\documentclass[11pt]{article}

\usepackage[T1]{fontenc}
\usepackage[utf8]{inputenc}

\usepackage{fancyhdr}

\usepackage{comment}
\usepackage{booktabs}
\usepackage{nicefrac}
\usepackage{microtype}
\usepackage{lipsum}
\usepackage{graphicx}
\usepackage{multirow}
\usepackage{mathrsfs}
\usepackage[title]{appendix}
\usepackage{xcolor}
\renewcommand{\textcolor}[2]{#2}

\usepackage{textcomp}
\usepackage{manyfoot}
\usepackage{algorithm}
\usepackage{algorithmicx}
\usepackage{algpseudocode}
\usepackage{listings}

\usepackage{natbib}

\usepackage{mathtools}
\mathtoolsset{showonlyrefs}

\usepackage{url}
\usepackage{subcaption}

\usepackage{tikz}
\usetikzlibrary{shapes,decorations,arrows,calc,arrows.meta,fit,positioning,decorations.pathreplacing}
\usepackage{tikz-network}
\usepackage{pgfplots}
\usepgfplotslibrary{groupplots}
\usepackage[inline]{enumitem}
\usepackage{multicol}
\usepackage{times}

\usepackage{amsmath,amsfonts,amsthm,amssymb,bm,dsfont}
\usepackage[a4paper,top=2.5cm,bottom=2.5cm,left=2cm,right=2cm]{geometry}
\usepackage{numprint}
\npthousandsep{,}
\npthousandthpartsep{}
\npdecimalsign{.}

\theoremstyle{plain}
\newtheorem{theorem}{Theorem}
\newtheorem{proposition}[theorem]{Proposition}

\DeclareMathOperator*{\argmax}{argmax}
\DeclareMathOperator*{\argmin}{argmin}
\DeclareMathOperator*{\med}{med}

\newcommand{\titledoc}{Online Bayesian changepoint detection for network Poisson processes with community structure}
\newcommand{\titleshort}{Online Bayesian changepoint detection for network Poisson processes with community structure}

\providecommand{\keywords}[1]{{\small{\textbf{\textit{Keywords ---}} #1}}}

\usepackage{authblk}

\author{Joshua Corneck}
\author{Edward A. K. Cohen}
\author{James S. Martin}
\author{Francesco Sanna Passino}

\affil{Department of Mathematics, Imperial College London \\ 180 Queen's Gate, SW7 2AZ, London}

\date{}

\title{\huge\textbf{\titledoc}}

\allowdisplaybreaks

\usepackage[colorlinks=true,linkcolor=blue,citecolor=blue,urlcolor=blue]{hyperref}
\numberwithin{equation}{section}

\begin{document}

\maketitle

\begin{center}
\textit{This is a pre-print of an article published in Statistics and Computing. The final authenticated version is available online at: \url{https://doi.org/10.1007/s11222-025-10606-w}.}
\end{center}

\begin{abstract}
Network point processes often exhibit latent structure that govern the behaviour of the sub-processes. It is not always reasonable to assume that this latent structure is static, and detecting when and how this driving structure changes is often of interest. In this paper, we introduce a novel online methodology for detecting changes within the latent structure of a network point process. We focus on block-homogeneous Poisson processes, where latent node memberships determine the rates of the edge processes. We propose a scalable variational procedure which can be applied on large networks in an online fashion via a Bayesian forgetting factor applied to sequential variational approximations to the posterior distribution. The proposed framework is tested on simulated and real-world data, and it rapidly and accurately detects changes to the latent edge process rates, and to the latent node group memberships, both in an online manner. In particular, in an application on the Santander Cycles bike-sharing network in central London, we detect changes within the network related to holiday periods and lockdown restrictions between 2019 and 2020.
\end{abstract}


\section{Introduction}\label{sec1}

 Network data describe complex relationships between a large number of entities, called \textit{nodes}, linked via \textit{edges}. In particular, much of such data is continuous-time valued and dynamic in nature, with nodes often exhibiting a latent \textit{community structure} that drives interactions \citep{amini2013}. By allowing temporal point processes to exist on the edges of the network, and not solely focusing on the adjacency matrix, such interaction data can be more accurately understood. Common examples of these network structures include email exchanges within a company \citep{enron2004} or posts on a social media platform \citep{socialnetworks2017}, where employer hierarchy or political leanings, respectively, could drive connections between individuals. Furthermore, it is not always reasonable to assume that such structure is static. In the case of email exchanges, a change in an employee's role within their organisation would alter the latent community structure of the network. Similarly, in social networks, a sporting or political event would likely affect the frequency of interactions only between users with particular characteristics and topics of interest. Detecting such changepoints is crucial in many real-world applications, such as in the context of computer networks, where modifications of network structure can signify malicious behaviour \citep[see, for example,][]{Hallgren23}. The detection of such changes in an online manner will be the focus of this paper.

Much of the original work on network models focused on simple, binary connections between nodes. One of the most common of these modelling approaches is the Stochastic Block Model \citep[SBM;][]{Holland}, which models connections between nodes via a latent community structure on a network. This model has undergone significant development, with theoretical results extending to include node degree heterogeneity \citep{Karrer11, amini2013, Gao2018-ey}, valued edges \citep{mariadassou-2010}, dynamic community structure in discrete time \citep{matias-2017, yang-2011, xu-2014}, and missing data \citep{Mariadassou2020-fg}; refer to \cite{Lee2019-qz} for a comprehensive review. In spite of their simplicity, SBMs can be used to consistently estimate and approximate much more complex network objects, such as exchangeable random graphs and their associated graphons, and thus provide a good approximation for any exchangable random graph model \citep{airoldi2013}. However, most of these advancements are based on aggregated, discrete time data.

There are several statistical methods for analysing dynamic networks, normally defined as graphs with time-evolving edge-connections. Interested readers are referred to \cite{Holme2015-me} for a review. However, most of this methodology does not capture instantaneous interactions, but relies instead on aggregations \citep[see, for example,][]{Shlomovich22}. To combat this, \cite{Matias2018-rx} developed a model for recurrent interaction events, which they call the semiparametric stochastic block model. In this framework, interactions are modelled by a conditional inhomogeneous Poisson process, with intensities depending solely on the latent group memberships of the interacting individuals. \cite{Perry13} uses a Cox multiplicative intensity model with covariates to model the point processes observed on each edge. %In a similar vein, \cite{miscouridou-2018} proposed a class of models for capturing sparsity, degree heterogeneity and community structure on a network point process. 
\cite{SannaPassino23} models the edge-specific processes via mutually exciting processes with intensities depending only on node-specific parameters. These methodologies handle complex temporal data, but they are offline methods, and work only in the case of a static latent network structure. Attempts have been made at capturing networks with a dynamic latent structure, but attention has focused on discrete time data, such as binary or weighted edges. \cite{matias-2017} proposed a method for frequentist inference for a model which extends the SBM to allow for dynamic group memberships. In that work, the dynamics of the groups are modelled by a discrete time Markov chain. \cite{heard-2010} developed a two-stage offline method for anomaly detection in dynamic graphs. However, to the best of our knowledge, there currently exists no methodology for online detection of changepoints in the context of continuous-time network data.

The recent work of \cite{Fang2020-ix} represents the first attempt at online estimation and community detection of a network point process. The authors build upon the foundation laid by \cite{Matias2018-rx} and extend it to an online setting, but maintain the assumption of a static latent structure. While \cite{Fang2020-ix} do offer suggestions as to how their framework could be extended to incorporate latent dynamics, the methodology developed therein requires knowing both the adjacency matrix and the number of latent groups a priori.

In this work, we propose a novel online Bayesian changepoint detection algorithm for network point processes with a latent community structure among the nodes. Our methodology is based on utilising \textit{forgetting factors} within a Bayesian context to sequentially update the variational approximation to the posterior distribution of the model parameters when new data are observed within a stream. We focus on what we will refer to as a dynamic Bayesian block-homogeneous network Poisson process, where dynamic refers to the fact that the latent structure of the network is time-varying. As an added benefit, our method is able to accurately infer the community structure and obtain a piecewise recovery of the conditional intensity function when we consider inhomogeneous Poisson processes on the edges. Extensions to our method are proposed to infer either the adjacency matrix of the network or the number of latent groups, each in conjunction with the edge rates. We also extend to handle cases where new groups are created or existing ones are merged.
% Extensions to our method are also proposed to simultaneously infer both the adjacency matrix of the network and the number of latent groups in an online manner, and handle cases where new groups are created or existing groups are merged into one another.

The remainder of this article is structured as follows: Section~\ref{sec:bdbhpp} describes the dynamic Bayesian block-homogeneous Poisson process model used in this work, and the possible local and global changes to the network structure occurring on such a network model. Section~\ref{sec:online_VB}
discusses the proposed online variational Bayesian inference approach via Bayesian forgetting factor, which is used to sequentially approximate the posterior distribution on the stream. The performance of the proposed inferential procedure is then tested in Section~\ref{sec:sims} on simulated data, and on real-world data from the Santander Cycles bike-sharing network in Section~\ref{sec:santander}, followed by a discussion.

\section{Bayesian dynamic block homogeneous Poisson process}\label{sec:bdbhpp}

\subsection{The model} \label{sec:model}

We consider a stochastic process on a network, %consisting in 
which produces dyadic interaction data between a set of $N$ nodes observed over time. %We model this behaviour using a graph $\mathcal{G} = (\mathcal{V}, \mathcal{E})$, 
Let $\mathcal{G} = (\mathcal{V}, \mathcal{E})$ be a graph,
where the set $\mathcal{V} := \{1,\dots,N\}$ corresponds to nodes, whereas the edge set $\mathcal{E}\subseteq\mathcal{V}\times\mathcal{V}$ contains pairs representing interactions between nodes in $\mathcal V$. %Note that $\mathcal{G}$ need not be complete, and let $\mathcal{R}$ denote the complete set of dyads. 
In particular, we write $(i,j)\in\mathcal E$ if there is a connection from node $i\in\mathcal{V}$ to node $j\in\mathcal{V}$. 
%node $i\in\mathcal{V}$ is connected with node $j\in\mathcal{V}$. 
Furthermore, we denote the set of all possible edges between nodes as $\mathcal{R}=\mathcal{V}\times\mathcal{V}$. 
The graph $\mathcal G$ could be equivalently represented via its adjacency matrix $\boldsymbol A=\{a_{ij}\}\in\mathbb\{0,1\}^{N\times N}$ such that $a_{ij}=\mathbb{I}_{\mathcal{E}}\{(i,j)\}$, where $\mathbb I_\cdot\{\cdot\}$ denotes the indicator function.
%We work with directed interactions and no self-interactions,
% We further assume that the graph is directed and has no self-loops, which indicates that $(i,j)\in\mathcal E$ does not necessarily imply $(j,i)\in\mathcal E$, and $(i,i) \notin \mathcal{E}$ for all $i\in\mathcal{V}$ respectively. Therefore, the adjacency matrix $\boldsymbol A$ of $\mathcal G$ is hollow and generally non-symmetric. 

%As in an SBM, individuals, and thus elements of $\mathcal{V}$, are assumed to be distributed among a set of groups $
We assume that each node $i\in\mathcal{V}$ belongs to a group $z_i\in\mathcal{K}$, where $\mathcal{K} := \{1,\dots,K\}$ %, which model a latent structure on the network, 
and we let $z=(z_1,\dots,z_N)\in \mathcal{K}^N$ be a priori independent and identically distributed with 
$z_i \sim \mathrm{Categorical}(\pi)$ where $\pi = (\pi_1,\dots,\pi_K)$, such that $\pi_k\geq 0$ for all $k\in\mathcal{K}$, and $\sum_{k=1}^K \pi_k=1$.
%random variables denoting the group memberships of the nodes, with $\mathbb{P}(z_i = k) = \pi_k$. In a slight abuse of notation, one 
We also associate to $z_i$ a binary vector $\tilde{z}_{i} = (\tilde z_{i1},\dots,\tilde z_{iK}) \in \{0,1\}^K$, with $\tilde z_{ik} = \mathbb{I}_{\{k\}}(z_i)$, % if and only if $z_i = k$, so that $z_i$ has a categorical distribution, $z_i \sim \mathrm{Categorical}(\pi)$ where $\pi = (\pi_1,\dots,\pi_K)$. 
and we %will sometimes 
write $\tilde{\boldsymbol{Z}} = \{\tilde z_{ik}\}_{i\in\mathcal{V}, k\in\mathcal{K}} \in \{0,1\}^{N\times K}$ to denote the matrix of group memberships. Denote by $\Pi = (\Pi_1,\dots,\Pi_K)$ a vector of the proportion of $N$ attributed to each group. 
%with $N\Pi_k \in \mathbb{N}$ for all $k\in\mathcal{K}$ and $\sum_{k\in\mathcal{K}}\Pi_k = 1$. We work in a 
Under a Bayesian framework, we place a Dirichlet prior distribution on $\pi$ with parameter $\gamma^0\in \mathbb{R}^K_+$. %, denoted as $\mathrm{Dirichlet}\left(\gamma^0\right)$. %, with $\gamma^0 \in \mathbb{R}^d_+$.

We observe a marked point process on the network, consisting of a stream of triplets $(i_\ell, j_\ell, t_\ell)\in\mathcal E\times\mathbb R_+,\ \ell=1,2,\dots$, denoting directed interactions from node $i_\ell$ to node $j_\ell$ at time $t_\ell$, where $t_\ell\leq t_{\ell^\prime}$ for $\ell<\ell^\prime$. The associated edge-specific counting process is denoted by $x_{ij}(\cdot)$, where, for all $(i,j)\in\mathcal{E}$, 
\begin{equation}
    x_{ij}((0,t]) \equiv x_{ij}(t) := \sum_{\ell} 
    \mathbb I_{\{(i,j)\}}\{(i_\ell,j_\ell)\}\mathbb I_{(0,t]}(t_\ell).
    %\mathbb{I}\{(i_\ell,j_\ell,t_\ell) \in \{i\}\times\{j\}\times [0,t)\},
\end{equation} 
If $(i,j)\notin\mathcal{E}$, $x_{ij}(t)=0$ for the entire observation period. 
In this work, we model the counting process $x_{ij}(\cdot)$ as a Poisson process with rate $\lambda_{z_iz_j}\in\mathbb R_+$, conditional on $(i,j)\in\mathcal{E}$ and on the node group memberships $z_i$ and $z_j$. Also, we place independent conjugate gamma prior distributions on the event rates $\lambda=\{\lambda_{km}\}\in\mathbb R_+^{K\times K}$. %, so that the rate on an edge are driven by the enclosing nodes (and the direction of the edge). Write $x(t) = \{x_{ij}(t)\}_{i,j\in\mathcal{V}\times\mathcal{V}}$ to be the full network counting process, which by superposition is a homogeneous Poisson process. We place a gamma prior on $\lambda = \{\lambda_{km}\}_{k,m\in\mathcal{K}} \in \mathbb{R}^{K\times K}_+$. As such, the model is initialised as
In summary, the full model can be expressed as follows:
\begin{align}
	x_{ij}(t) \mid z_i, z_j, \lambda_{z_iz_j} &\sim \mathrm{Poisson}(\lambda_{z_iz_j}t), &
    \text{for all } (i,j) \in \mathcal{E},\  \label{eqn:poiss_network} \\
	\lambda_{km} &\sim \mathrm{Gamma}\left(\alpha_{km}^0, \beta_{km}^0\right),
	&\text{for all } k, m \in \mathcal{K}, \label{eqn:lam_prior} \\
	z_i \mid \pi &\sim \mathrm{Categorical}(\pi),
	&\text{for all } i \in \mathcal{V}, \label{eqn:z_init} \\
	\pi &\sim \mathrm{Dirichlet}\left(\gamma^0\right), \label{eqn:pi_prior}
\end{align}
where $\gamma^0 = \left(\gamma_1^0,\dots,\gamma_K^0\right) \in \mathbb{R}^K_+$, $\alpha^0 = \left\{\alpha^0_{km}\right\}_{k,m\in\mathcal{K}} \in \mathbb{R}^{K\times K}_+$ and $\beta^0 = \left\{\beta_{km}^0\right\}_{k,m\in\mathcal{K}} \in \mathbb{R}^{K\times K}_+$. Note that in \eqref{eqn:poiss_network}, we use that for a homogeneous Poisson process with rate $\lambda^\ast$, the number of counts for that process on an interval of length $t$ is Poisson distributed with rate $\lambda^\ast t$. The model for $x_{ij}(t)$ in \eqref{eqn:poiss_network} is known as the \textit{block homogeneous Poisson process} \citep[BHPP;][]{Fang2020-ix}. The full Bayesian model in \eqref{eqn:poiss_network}-\eqref{eqn:pi_prior} is represented in graphical model form in Figure~\ref{fig:dag_model}.%. The model as outlined in \eqref{eqn:lam_prior}-\eqref{eqn:poiss_network} 

\begin{figure}
    \centering
    \includegraphics[width=0.45\textwidth]{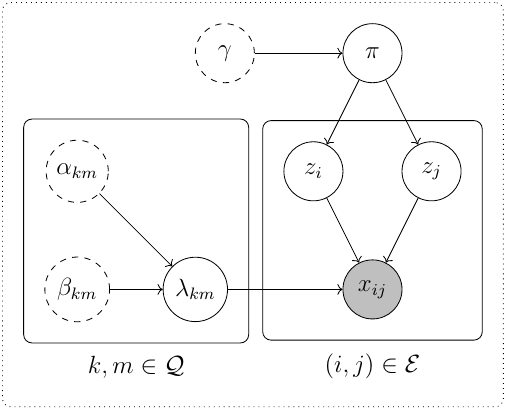}
    \caption{A directed acyclic graphical model representation of the model in Equations~\eqref{eqn:poiss_network}-\eqref{eqn:pi_prior}.} %Dashed circles represent fixed quantities, and thick circles stochastic ones. White and grey circles denote unobserved and observed quantities, respectively, and boxes grouping nodes indicate index repetition.}
    \label{fig:dag_model}
\end{figure}

Under the BHPP, the expected waiting time between events across the entire network is $\mathcal{O}(N^{-1})$, as described in the following proposition.

\begin{proposition}\label{prop:rate_growth}
    Let $\mathcal{G} = (\mathcal{V}, \mathcal{E})$ be a directed graph with no isolated nodes. %For the set-up considered in this work, where 
    If on each edge lives a homogeneous Poisson process with rate $\lambda_{z_iz_j}\in\mathbb R_+$, only dependent upon node memberships $z_i,z_j\in\mathcal{K}$, then the expected waiting time between arrivals for the full network counting process $x(t)=\sum_{(i,j)\in\mathcal E} x_{ij}(t)$ is $\mathcal{O}\left(N^{-1}\right)$.
\end{proposition}
\begin{proof}
    Let $\mathcal{G} = (\mathcal{V},\mathcal{E})$ with $|\mathcal{V}| = N$. Suppose that $\mathcal{G}$ has $M$ connected components. As we have no isolated nodes, the number of edges is minimised with $M=N/2$ and $|\mathcal{E}| = N/2$ if $N$ is even, or $M=(N-1)/2$ and $|\mathcal{E}| = (N+1)/2$, if $N$ is odd. Without loss of generality, assume $\lambda_{11} = \min_{k,m\in\mathcal{K}}\{\lambda_{km}\}$ and call the rate of the full network Poisson process $\lambda$. By superposition, it follows $\lambda \geq N\lambda_{11}/2$ or $\lambda \geq (N+1)\lambda_{11}/2$, for $n$ even and odd, respectively, and so the expected waiting time $w$, is bounded as $w \leq 2/\lambda_{11}N$.
\end{proof}

The aim of the present paper is to provide an online algorithm for detecting changepoints; such an algorithm must be able to perform inference at each time step before the arrival of any new data. A consequence of Proposition~\ref{prop:rate_growth} is that for networks of increasing size, this becomes infeasible for large $N$ when observing all events as they arrive, as an online algorithm would need to be at least as fast. This motivates a further assumption that data arrives as a stream of batched counts rather than a stream of continuous-valued event times. Let the batches be observed at times $L_r,\ r=1,2,\dots$, where $L_r<L_{r^\prime} $ for $r<r^\prime$. %$\mathcal{D}_m$ at times %$t=L_m$, taking the form
For $r=1,2,\dots$, each data batch takes the form
\begin{equation}
\mathcal{D}_r := \left\{\left(i_\ell, j_\ell, t_\ell\right): L_{r-1} < t_\ell \leq L_r\right\}; \quad L_{0} \equiv 0.
\end{equation}
%where $\left(i_\ell, j_\ell, t_\ell\right) \in \mathcal{E} \times [L_{m-1},L_m)$ denotes a directed interaction from node $i_\ell$ to node $j_\ell$ at time $t_\ell$. 
We consider the case of a constant time $\Delta$ between batches, so that $L_r - L_{r-1} = \Delta$ for all $r$. %The
Also, we denote the complete set of data available at time $L_r$ as $\mathcal{D}_{1:r} := \{\mathcal{D}_1,\dots, \mathcal{D}_r \}$. This setup is illustrated in Figure \ref{fig:notation_diagram}. 
\begin{figure*}[t]
    \centering
    \includegraphics[width=\textwidth]{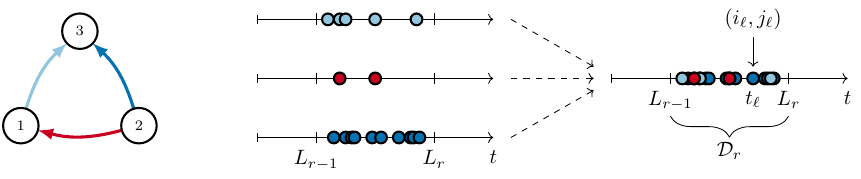}
    \caption{Illustration of the setup in Section~\ref{sec:model} for a network with $N=3$ and $\mathcal{E}=\{(1,3),(2,1),(2,3)\}$. Circles represent event times, coloured according to their edge. Arrival times to each edge %in $\mathcal{E}$ 
    are collated into $\mathcal{D}_r$.}
    \label{fig:notation_diagram}
\end{figure*}

Initially, we assume that both the adjacency matrix $\boldsymbol A$ and the number of latent groups $K$ are known \textit{a priori}, similar to the setup in \cite{Fang2020-ix}. However, we note that assuming this knowledge is limiting in practical applications, since these quantities are usually unknown. Therefore, Sections \ref{subsec:unknown_adj} and \ref{subsec:unknown_groups} discuss extensions of the BHPP model and corresponding inference procedures with an unknown adjacency matrix and an unknown number of groups.

\subsection{Changes to the network}

In real-world applications, it is unrealistic to assume that the latent group memberships $z=(z_1,\dots,z_N)$ and rate matrix $\lambda=\{\lambda_{km}\}_{k,m\in\mathcal{K}}$ are constant across the entire observation period. 
%We consider changes in latent group membership $z$, and changes to the rate matrix $\lambda$. 
Our objective is therefore to develop a framework which could be used to detect \textit{global} and \textit{local} changes to the model parameters.

For \textit{local} changes, we consider modifications to the latent group structure, allowing for any $n\leq N$ nodes to change their group membership at any time $t^\prime$. In contrast to the approach of \cite{matias-2017}, in which the membership of node $i$ over the observation window behaves according to a discrete-time, irreducible and aperiodic Markov chain, we place no model on how or when nodes change and allow for any (finite) number of changes in a given interval. Furthermore, we allow for changes in continuous time, again in contrast to the discrete-time formulation of \cite{matias-2017}. Our main objective is only to detect if and when nodes have changed memberships, and not to estimate the underlying group membership switching process. 
% Figure \ref{fig:dynamic_graph_structure} illustrates the effect of a single node switch in the case of $N=4$ nodes.

% \begin{figure*}[!t]
%     \centering
%     \begin{subfigure}[h!]{0.3\textwidth}
%         \centering
%         \includegraphics[width=0.6\textwidth]{figures/tikz_pictures/fig3_1.pdf}
%         \caption{Observed graph structure}
%     \end{subfigure}
%     \hspace{5mm}
%     \begin{subfigure}[h!]{0.3\textwidth}
%         \centering
%         \includegraphics[width=0.6\textwidth]{figures/tikz_pictures/fig3_2.pdf}
%         \caption{Latent group structure}
%     \end{subfigure}
%     \hspace{5mm}
%     \begin{subfigure}[h!]{0.3\textwidth}
%         \centering
%         \includegraphics[width=0.6\textwidth]{figures/tikz_pictures/fig3_3.pdf}  
%         \caption{Change in latent group structure}
%     \end{subfigure}
%     \caption{
%     Visualisation of an observed graph $\mathcal G$ with $N=4$, with a change in the latent group structure. 
%     %The left-hand panel shows the observed graph structure. The middle panel shows the latent group structure, and the right-hand panel the result of a change to the latent structure. 
%     The changing node is coloured in grey, and edges are coloured according to the groups they connect.}
%     \label{fig:dynamic_graph_structure}
% \end{figure*}

For \textit{global} changes, we consider jumps in the intensity between any groups $k,m\in \mathcal{K}$  at some time $t^\prime\in\mathbb R_+$ of the form $\lambda_{km}\mapsto \lambda_{km}^\prime$ ($\lambda_{km}\neq \lambda_{km}^\prime$). Again, we are interested in assessing if the block-specific intensity $\lambda_{km}$ has changed, and we do not estimate or posit assumptions on the mechanism that leads to such changes.

\section{An online variational Bayesian estimation procedure}
\label{sec:online_VB}

We present an online inference procedure for tracking the time-evolving latent structure and parameters of the BHPP model described in Section~\ref{sec:model}. We focus first on the set-up where the full edge set $\mathcal{E}$ and the number of latent groups $K$ is known, as in \cite{Fang2020-ix}. %Once a framework is established for this setting, we extend to more complex cases. 
Next, we consider when $\mathcal{E}$ is unknown, and then we move to the case of an unknown number of groups, using a Bayesian nonparametric approach. 

\subsection{Variational Bayesian approximation}
\label{subsec:variational_bayes}

Variational Bayesian (VB) inference has the objective to approximate a posterior distribution when it is not analytically tractable \citep{wang2019}, offering a faster estimation approach when compared to Markov Chain Monte Carlo methods \citep{blei-2017}. This makes it more suited to the requirements of an online learning framework.
We adopt the terminology of \citet{blei-2017} in distinguishing between local latent variables $z_1,\dots,z_N$, that scale with the number of nodes, and global latent variables $\theta = (\pi,\lambda)$, whose dimension, written $S\in\mathbb N$, does not change with $N$. 

In a variational approach, one posits a family of distributions $\mathcal F^{N+S}$ on the parameter space $(\theta,z)$, %parameterised in some convenient way, 
and seeks to select the component $q^\ast(\theta,z)\in \mathcal F^{N+S}$ closest to the true posterior $p(\theta,z\mid x)$ in the sense of the Kullback-Leibler (KL) divergence, where $x$ denotes the observed data. %We work with 
A mathematically convenient choice for $\mathcal F^{N+S}$ is the mean-field variational family, the set of factorisable distributions of the form
\begin{equation}
\mathcal{F}^{N+S} = \left\{q:q(\theta,z) = \prod_{i=1}^N q_{z_i}(z_i)\prod_{s=1}^S q_{\theta_s}(\theta_s)\right\}.
\end{equation}
%with $N$ being the number of local latent variables, and $S$ the number of global latent variables. 
The approximating distribution $q^*(\theta,z)\in\mathcal{F}^{N+S}$ is then given as%, finding $q^*(\theta,z)$ as
\begin{equation}
q^\ast(\theta,z) = \argmin_{q \in \mathcal{F}^{N+S}} \mathrm{KL}\left[q(\theta,z)\ \| \ p(\theta,z|x)\right].
\label{eqn:VB_objective}
\end{equation}
% In practice, $q^\ast(\theta,z)$ is found by solving the equivalent problem 
% \begin{equation}
% q^\ast(\theta,z) = \argmax_{q \in \mathcal{F}^{N+S}} \mathrm{ELBO}(q),
% \label{eqn:ELBO_objective}
% \end{equation}
% where the evidence lower bound (ELBO) is defined as
% \begin{equation}
%     \mathrm{ELBO}(q) := \mathbb{E}_q\left\{\log \frac{p(x,\theta,z)}{q(\theta,z)}\right\} = -\int q(\theta,z)\log\frac{q(\theta,z)}{p(\theta,z,x)}d \theta d z.
% \end{equation}
 Despite every $q(\theta,z) \in \mathcal{F}^{N+S}$ taking a simple product structure, the solution $q^\ast(\theta,z)$ to \eqref{eqn:VB_objective} is usually not available analytically. A popular method for approximating the global minimum $q^\ast(\theta,z)$ of the KL-divergence is the Coordinate Ascent Variational Inference (CAVI) algorithm \citep{bishop-2006, blei-2017}, which iteratively updates each component of $q(\theta,z)$ while keeping the others fixed. Letting $\phi$ denote the $(N+S)$-dimensional vector of latent variables $\phi = (z,\theta)$, CAVI proceeds as:
 \begin{equation}
   \hat{q}_{\phi_j}^{(t+1)}(\phi_j) = \argmin_{q_{\phi_j}} \mathrm{KL}\left[q_{\phi_j}(\phi_j) \prod_{i<j}\hat{q}_{\phi_i}^{(t+1)}(\phi_i)\prod_{i>j}\hat{q}_{\phi_i}^{(t)}(\phi_i)\ \Bigg|\Bigg|\ p(\phi|x)\right],\qquad 
  \label{eqn:cavi_target}
 \end{equation}
for $j \in \{1,\dots,N+S\},\ t=0,1,2,\dots$
 until some convergence criterion is met. %A single pass through the algorithm consists of $N+S$ update steps, where we sweep through a fixed ordering of the latent variables, updating as
 Under the mean-field variational approximation, \eqref{eqn:cavi_target} has the following solution \citep{bishop-2006}:
 \begin{equation}
 \hat{q}_{\phi_j} \propto \exp\{\mathbb{E}_{-\phi_j}\left[\log p(x,\phi)\right]\},
 \label{eq:cavi_update_sol}
 \end{equation}
where the notation $\mathbb{E}_{-\phi_j}$ used to denote the expectation with respect to all components of $\phi$ except $\phi_j$ \citep{blei-2017}. %The update procedure in full is summarised in Algorithm \ref{alg:cavi}. 
As per convention with variational inference algorithms, we will drop the subscript on $q_{\phi_j}(\phi_j)$ for notational convenience, and take the argument of the variational component to index it \citep[see, for example,][]{blei-2017}. 

It should be noted that CAVI is only guaranteed to achieve a local minimum, and hence is sensitive to intialisation \citep{Zhang2020-re, blei-2017}.
%\begin{remark}
Furthermore, with $\phi = (z,\theta)$, updates in \eqref{eq:cavi_update_sol} have an explicit ordering.  
The ordering of the updates can also have implications on the convergence properties of the algorithm \citep{Ray22}, similarly to Gibbs sampling \citep[see, for example,][]{vanDyk08}.

\subsection{Online VB for the dynamic BHPP}%DPP-SBM}

Upon receiving the latest batch $\mathcal D_r$ at time $L_r$, the aim of an online algorithm is to update the estimates of the BHPP parameters based on the entire history $\mathcal{D}_{1:r}$ of the process. To ensure finite and constant complexity and memory, a truly online algorithm should be \textit{single pass} \citep[see, for example][]{Bifet18}: each data batch should be inspected only once, and summaries of previous data batches $\mathcal{D}_{1:(r-1)}$ are used in conjunction with $\mathcal{D}_r$ to update parameter estimates.

The prior-posterior distribution updates within a Bayesian framework offers a natural way to propagate information forward in time. The posterior distribution at time $L_{r-1}$ contains the information that was learned from $\mathcal{D}_{1:(r-1)}$. A naive approach therefore is to pass this through as the new prior distribution for the update at time $L_r$. However, this will perform poorly in the case of changing latent structure. As the number of batches increases, the new prior distributions become more concentrated around the current best estimates. This results in an updated posterior distribution that is largely dominated by its prior distribution (corresponding to the posterior distribution up to the previous batch) and which is increasingly insensitive to new data and changes in the data generating process.

We therefore propose to flatten the posterior obtained at $L_{r-1}$, via temperature parameters, at the point of passing it through as the prior distribution for the update at $L_r$. This flattening step is in effect down-weighting previous observations and can therefore be considered as a Bayesian analogue of the \emph{forgetting factor} procedure that has been widely used in frequentist online literature \citep[see, for example,][]{Haykin-2002,bodenham-2017}. In this way, the parameter estimates are quicker to respond to changes in the latent rate structure.

Due to the independence of counts from a Poisson process within non-overlapping time windows, the posterior density %distribution
at \textcolor{red}{time step} $r$ under the BHPP model factorises as
\begin{align}
p(z, \theta \mid \mathcal D_{1:r}) &\propto p(z \mid \pi) p(\pi) p(\lambda) \prod_{\ell=1}^r p(\mathcal D_{\ell} \mid z,\lambda)  \\
&\propto p(\mathcal D_{r}\mid z,\lambda) \times p(z, \theta \mid \mathcal D_{1:(r-1)}),\qquad
\label{eqn:posterior}
\end{align}
where the prior distributions are chosen to be conjugate, as per \eqref{eqn:lam_prior} and \eqref{eqn:pi_prior}.
From \eqref{eqn:posterior}, the posterior distribution at \textcolor{red}{time step} $r-1$ can be interpreted as the prior for the posterior distribution at \textcolor{red}{time step} $r$.

The posterior density %distribution 
in \eqref{eqn:posterior} is not available in closed form and, due to the marginal density %distribution 
$p(\mathcal{D}_{1:r})$ being unavailable, can only be evaluated up to a normalising constant. %without marginalising out $z$, which requires summing over exponentially many terms with growing $N$
We thus consider the VB approach with the mean-field variational family, deploying the CAVI algorithm to compute approximate posterior distributions. Denoting the approximate posterior at time step $r-1$ by $\hat{q}^{(r-1)}(\theta,z)$, we pass through as the prior for time step $r$ the \textit{tempered} density%distribution 
\begin{align}
    \hat{q}^{(r-1)}_{\delta}(\theta,z)&= \hat{q}^{(r-1)}_{\delta_\pi}(\pi) \times \prod_{k,m\in\mathcal{K}} \hat{q}^{(r-1)}_{\delta_\lambda}(\lambda_{km})\prod_{i\in\mathcal{V}} \hat{q}^{(r-1)}_{\delta_z}(z_i) \label{eq:prior_structure} \\ &\propto \hat{q}^{(r-1)}(\pi)^{\delta_\pi} \prod_{k,m\in\mathcal{K}}\hat{q}^{(r-1)}(\lambda_{km})^{\delta_{\lambda}}\prod_{i\in\mathcal{V}}\hat{q}^{(r-1)}(z_i)^{\delta_{z}},
    \label{eq:tempered_approximation}
\end{align}
controlled by the \textit{Bayesian forgetting factor} (BFF) $\delta=(\delta_\lambda,\delta_z,\delta_\pi)\in (0,1]^3$. Here, $C_{r-1}$ is a normalising constant to ensure a valid density. This procedure is visualised in Figure~\ref{fig:pipeline}.

\begin{figure*}[t]
    \centering
    \includegraphics[width=0.8\textwidth]{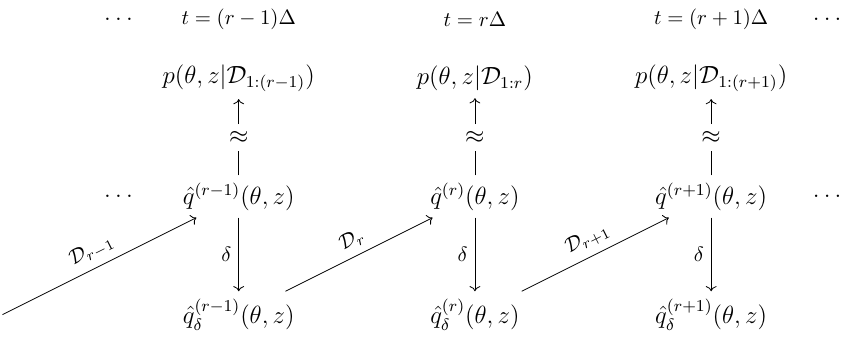}
    \caption{A %pictorial
    representation of the transformation of the CAVI approximation from the $r$-th time step via temperature parameters $\delta$ and its use as a prior for the $(r+1)$-th time step. The overlaying of $\mathcal{D}_{r}$ on the diagonal arrows indicates that the transformed CAVI posterior is combined with incoming data to yield the next approximate posterior.}
    \label{fig:pipeline}
\end{figure*}

The effect of the BFF in down-weighting previous data and placing greater emphasis on the latest batch $\mathcal{D}_r$ is illustrated in Figure \ref{fig:motivating}. A fully-connected network is simulated on the time interval $[0,5]$ with $N=500$ nodes and $K=2$ groups. An instantaneous rate change in $\lambda_{11}$ is made at $t=1$. The left panel of Figure~\ref{fig:motivating} shows the results for naive approach (without a forgetting factor), whereas the right panel includes a BFF and demonstrates a much faster response to the change and quicker convergence of the posterior mean to the true value.\\
\textcolor{red}{A BFF of 0.1 was found to work well in our simulations, except in the case of very low rates. In this case, where we might see no counts in an interval, one would select $\delta_\lambda$ closer to 1 to allow the prior to inform the estimates. However, our methodology does not require that $\Delta$ be fixed, and so in the case of low counts, one could instead allow for a longer waiting period before updating.}

\begin{figure*}[t]
     \centering
     \includegraphics[width=\textwidth]{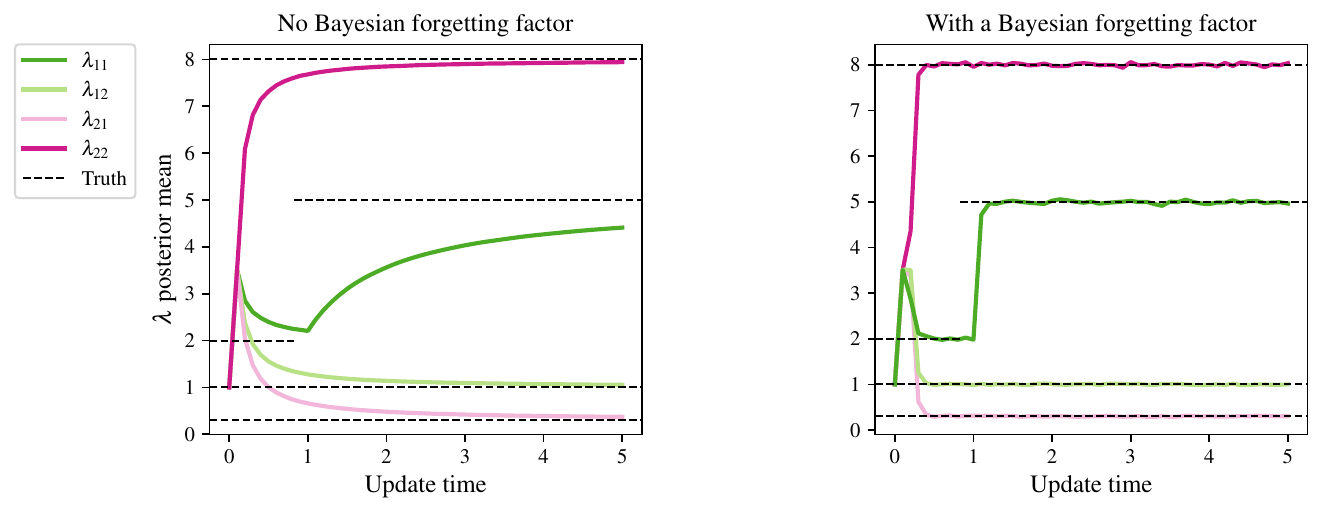}
    \caption{Estimated posterior mean of $\lambda\in\mathbb R_+^{K\times K}$ for a fully connected graph with $N=500$ nodes simulated on $[0,5]$, with $K=2$ groups, $\Pi = (0.6, 0.4)$, $\lambda_{22}=8,\ \lambda_{12}=1,\ \lambda_{21}=0.3$, and $\lambda_{11}=2$. At $t=1$, $\lambda_{11}$ changes to $5$. The four coloured lines are the posterior means of the components $\lambda$, plotted against time step. The black dashed horizontal lines are the true %, initial components 
    values of $\lambda$ at each time point. $\delta=0.1$ in the right-hand panel.%, the vertical black, dashed vertical line is the changepoint time. %The two panels show inference with and without a Bayesian forgetting factor.
    }
    \label{fig:motivating}
\end{figure*}

\subsection{Online CAVI updates}

We assume that each component of \eqref{eq:prior_structure} takes the same form of its corresponding complete conditional distribution under the standard BHPP model in \eqref{eqn:poiss_network}--\eqref{eqn:pi_prior}, raised to a power and normalised: 
\begin{align*}
    \hat{q}^{(r-1)}_{\delta_\lambda}(\lambda_{km}) &= \frac{1}{C_{\lambda_{km},r-1}}%\prod_{k,m\in\mathcal{K}}
    \mathrm{Gamma}\left(\alpha_{km}^{(r-1)}, \beta_{km}^{(r-1)}\right)^{\delta_\lambda}, 
    &\text{for all } k,m\in\mathcal{K},  \\ 
    \hat{q}^{(r-1)}_{\delta_z}(z_i) &= \frac{1}{C_{z_i,r-1}}\mathrm{Categorical}\left(\tau_i^{(r-1)}\right)^{\delta_z}, 
    &\text{for all } i\in\mathcal{V},  \\
    \hat{q}^{(r-1)}_{\delta_\pi}(\pi) &= \frac{1}{C_{\pi,r-1}}\mathrm{Dirichlet}\left(\gamma^{(r-1)}\right)^{\delta_\pi}.
\end{align*}
\textcolor{red}{The above notation indicates that we take the probability density or mass function of the relevant distribution, raise it to a power and then normalise}. $C_{\lambda_{km},r-1},C_{z_i,r-1}$ and $C_{\pi,r-1}$ are normalising constants, ensuring that the densities are valid, and $\delta=(\delta_\lambda,\delta_z,\delta_\pi)\in (0,1]^3$ are temperature parameters specific to each class of latent variables.
Suppose that data is observed on the interval $I_r := (L_{r-1}, L_r]$, for $r\in\mathbb{N}$, and denote the count on edge $(i,j)$ during this interval by $x_{ij}(I_r)$. Under the prior structure in \eqref{eq:prior_structure}, it is shown in Appendix \ref{app:cavi} that the CAVI sequential updates are:
\begin{enumerate}
    \item $\hat{q}^{(r)}(\lambda_{km}) = \mathrm{Gamma}\left(\alpha^{(r)}_{km}, \beta^{(r)}_{km}\right)$ for all $k,m\in \mathcal{K}$ where $\alpha^{(r)}_{km}, \beta^{(r)}_{km}$ are defined as:
    \begin{align}
    \alpha^{(r)}_{km} &= \delta_\lambda\left(\alpha^{(r-1)}_{km} - 1\right) +\sum_{(i,j)\in\mathcal{E}}\tau_{ik}^{(r-1)}\tau_{jm}^{(r-1)}x_{ij}(I_r) + 1, \\
    \beta^{(r)}_{km} &= \delta_\lambda\beta^{(r-1)}_{km} + \Delta \sum_{(i,j)\in\mathcal{E}}\tau_{ik}^{(r-1)}\tau_{jm}^{(r-1)}.\qquad \label{eqn:online_lam_update}
    \end{align}
    \item $\hat{q}^{(r)}(z_i) = \mathrm{Categorical}\left( \tau_i^{(r)}\right)$ for all $i\in\mathcal{V}$ where $\tau_i^{(r)} = (\tau_{i1}^{(r)},\dots,\tau_{iK}^{(r)})$, with $\sum_{k\in\mathcal{K}} \tau_{ik}^{(r)} = 1$, satisfies the relation:
    \begin{align} 
    \tau_{ik}^{(r)} &\propto \exp\Bigg\{\delta_z\left[\psi\left(\gamma^{(r-1)}_k\right) - \psi\left(\sum_{m\in\mathcal{K}}\gamma^{(r-1)}_m\right)\right] + x_{ii}(I_r)\left[\psi\left(\alpha_{kk}^{(r)}\right) - \log\left(\beta_{kk}^{(r)}\right)\right] - \Delta\frac{\alpha_{kk}^{(r)}}{\beta_{kk}^{(r)}} \notag
    \\
    & + \sum_{m\in\mathcal{K}}\Bigg[\sum_{j:(i,j)\in\mathcal{E}}\tau_{jm}^{(r)}\bigg(x_{ij}(I_r)\left\{\psi\left(\alpha^{(r)}_{km}\right) -\log\left(\beta^{(r)}_{km}\right)\right\} -  \Delta\frac{\alpha^{(r)}_{km}}{\beta^{(r)}_{km}}\bigg) \\
    &\times \bigg(1 - \mathbb{I}_{\{k\}}(m)\mathbb{I}_{\{i\}}(j)\bigg) + \sum_{j':(j',i)\in\mathcal{E}}\tau_{j'm}^{(r)}\bigg(x_{j'i}(I_r)\left\{\psi\left(\alpha^{(r)}_{mk}\right) - \log\left(\beta^{(r)}_{mk}\right)\right\}-\Delta\frac{\alpha^{(r)}_{mk}}{\beta^{(r)}_{mk}}\bigg)\\
    &\times \bigg(1 - \mathbb{I}_{\{k\}}(m)\mathbb{I}_{\{i\}}(j')\bigg)\Bigg]\Bigg\}. \label{eqn:online_z_update}
    \end{align}
    \item $\hat{q}^{(r)}(\pi) = \mathrm{Dirichlet}\left\{\left(\gamma^{(r)}_1,\dots,\gamma^{(r)}_K\right)\right\}$ where $\gamma^{(r)}_k$ is defined as:
    \begin{align}
    \gamma^{(r)}_k = \delta_\pi\left(\gamma_k^{(r-1)} - 1\right) + \delta_z\sum_{i\in\mathcal{V}}\tau_{ik}^{(r)} + 1.\qquad \label{eqn:online_pi_update}
    \end{align}
\end{enumerate}

Here $\psi(\cdot)$ is the digamma function. Note that the parameters $\tau_{ik}^{(r)},\ i\in\mathcal{V},\ k\in\mathcal{K}$, are jointly optimised via a fixed point solver for \eqref{eqn:online_z_update}, iteratively moving through the rows of the matrix $\tau = \{\tau_{ik}\}_{i\in\mathcal{V},\ k\in\mathcal{K}} \in [0,1]^{N\times K}$ until convergence. The fixed point solver is initialised at the $r$-th time step by using the value of $\tau$ outputted from the $(r-1)$-th \textcolor{red}{time step} as the starting point, which also partially circumvents the problem of label switching \citep{jasra2005}. %Furthermore, the convergence criterion for $\tau$ can be specified by the user.

% The procedure would need to be run at every time step $\Delta$, when a new data batch is observed. The CAVI approximations from the run at $t=r\Delta$ are given in  \eqref{eqn:online_z_update}-\eqref{eqn:online_lam_update} of Proposition \ref{prop:online_cavi_updates}. Algorithm \ref{alg:online_cavi} describes the full inference procedure based upon the updates derived in Proposition \ref{prop:online_cavi_updates}. 

Algorithm \ref{alg:online_cavi} describes the full inference procedure based upon the updates in equations \eqref{eqn:online_lam_update}--\eqref{eqn:online_pi_update}.

\begin{algorithm}
\caption{Online VB Procedure for BHPP}
\begin{algorithmic}[1]
\item Initialise $ \alpha^{(0)}$, $\beta^{(0)}$, $\gamma^{(0)}$ and set $\tau_{ik} = 1/K$ for all $i \in \mathcal{V}$ and $k\in\mathcal{K}$. 
\For {$r=1,2,\dots$}
    \For {$\ell = 1,2,\dots,N_\mathrm{CAVI}$}
        \While{$\tau$ not converged}
            \For {$j = 1$ to $N$}
                \State Update $\tau_j = (\tau_{j1},\dots,\tau_{jK})$ as in equation \eqref{eqn:online_z_update}, intialising as the output of update $r-1$.
                \State Normalise $\tau_j$ such that $\sum_{k=1}^K \tau_{jk} = 1$.
            \EndFor
        \EndWhile
        \State Update $\gamma^{(r)}$, $\alpha^{(r)}$ and $\beta^{(r)}$ as in  \eqref{eqn:online_lam_update}--\eqref{eqn:online_pi_update}.
    \EndFor
\EndFor
\end{algorithmic}
\label{alg:online_cavi}
\end{algorithm}

\subsection{Detecting changepoints}

In this work, we aim to detect two types of changepoints: changes in the  point process rates and changes in the group memberships. 

To detect changes in the point process rates, we compute the KL-divergence between the approximate posteriors $\hat{q}^{(r)}(\lambda_{km})$ and $\hat{q}^{(r-s)}(\lambda_{km})$, for $s=1,2,\dots,\kappa$, where $\kappa\in\mathbb{N}$ is a lag parameter. We compare the current approximation with estimates beyond that of the previous update, up to a maximal lag of $\kappa$, to avoid incorrectly flagging outliers as changes. Figure \ref{fig:kl_lag} illustrates the behaviour of the KL-divergence for a lag of 1 and a lag of 2. In Figure \ref{fig:KL-lag-no-cp}, we see that using only a lag of 1 would result in the false identification of a change. There is a trade-off to be made between confidence and speed of detection. We found that a maximal lag of $\kappa = 2$ provided a good balance.

To compute the KL-divergence, we note that $\hat{q}^{(r)}(\lambda_{km})$ is gamma distributed for all $r\in\mathbb{N}$. For two random variables $X_1\sim \mathrm{Gamma}(\alpha_1,\beta_1)$ and $X_2\sim\mathrm{Gamma}(\alpha_2,\beta_2)$ with probability distributions $p_1$ and $p_2$ respectively, their KL-divergence is: 
\begin{align}
    \mathrm{KL}(p_1\ \|\ p_2) &= \alpha_2\log\frac{\beta_1}{\beta_2} - \log \frac{\Gamma(\alpha_1)}{\Gamma(\alpha_2)} + (\alpha_1 - \alpha_2)\psi(\alpha_1) - (\beta_1 - \beta_2)\frac{\alpha_1}{\beta_1}.
    \label{eqn:KL_gamma}
\end{align}
Two burn-in times are required to make use of the KL-divergence. In particular, $B_1\in\mathbb{N}$ time steps are needed for the algorithm to converge from initialisation, \textcolor{red}{which are then discarded.} $B_2\in\mathbb{N}$ further \textcolor{red}{time steps are} needed to obtain enough samples to flag changes. During the initial $B_1 + B_2$ \textcolor{red}{time steps}, stationarity is assumed, \textcolor{red}{which is a standard and necessary assumption (see, for example, \cite{bodenham-2017})}. After $B_1 + B_2$ \textcolor{red}{time steps}, we propose to use the median absolute deviation (MAD) to evaluate the presence of a changepoint, as it is robust to outliers \citep{leys2013}. In general, for a dataset $\mathcal{Y} := \{y_i\}_{i=1}^n$, the MAD is defined as
\begin{equation}
\mathrm{MAD} = \med\big(\{|y_i - \med(\mathcal{Y})|\}_{i=1}^n\big),
\end{equation}
where $\med(\cdot)$ denotes the median of the elements in a set. An observation $y_j \in \mathcal{Y}$ is classified as an outlier if $|y_j - \med(\mathcal{Y})| > W\cdot \mathrm{MAD}$, for a choice of threshold $W\in\mathbb{R}_+$. \textcolor{red}{After a change is flagged, only a further $B_2$ time steps are required to collect a new dataset for flagging changes from this new state.}

In our proposed methodology, the elements in the set $\mathcal{Y}$ are the realised values of the KL-divergence between different gamma distributions, approximating the posterior distributions for the parameters $\lambda_{km},\ k,m\in\mathcal{K}$ at different time windows. \textcolor{red}{For each pair $(k,m)\in \mathcal{K} \times \mathcal{K}$, we let $\mathcal{X}_{km}^r$ store the samples used to check for a changepoint at time step $r$. For $r=B_1 + B_2 + 1$, define
\begin{align}
    \mathcal{X}_{km}^r := \left\{p_{km}^{r,1},\dots,p_{km}^{r,B_2}:p_{km}^{r,\ell} = q^{(B_1 + \ell)}(\lambda_{km})\text{ for } \ell=1,\dots,B_2\right\}. \notag
\end{align}
We then construct a set of KL-divergences $\mathcal{
Y}_{km}^{r,\kappa}$ as
\begin{align}
    \mathcal{Y}_{km}^{r,\kappa} := &\Big\{\mathrm{KL}\left[p_{km}^{r,\ell} \mid\mid p_{km}^{r,\ell - s}\right] : 1 \leq s \leq \kappa,\ s+1 \leq \ell \leq B_2\Big\}.
\end{align}
We then flag $r>B_1 + B_2$ as an outlier if 
\begin{equation}
    \left|\mathrm{KL}\left[q^{(r)}(\lambda_{km}) \mid \mid p_{km}^{r,B_2}\right] - \mathrm{med}\left(\mathcal{Y}_{km}^{r,\kappa}\right)\right| > W\cdot \mathrm{MAD}.
\end{equation}
If $q^{(r)}(\lambda_{km})$ is not an outlier, we define $\mathcal{X}_{km}^{r+1}$ as
\begin{align}
    \mathcal{X}_{km}^{r+1} := &\left\{p_{km}^{r+1,1}, \dots, p_{km}^{r+1,B_2} : p_{km}^{r+1,\ell-1} = p_{km}^{r,\ell} \text{ for } 2 \leq \ell \leq B_2, \text{and } p_{km}^{r+1,B_2} = q^{(r)}(\lambda_{km})\right\},
\end{align}
otherwise, we set $\mathcal{X}_{km}^{r+1} = \mathcal{X}_{km}^r$. This update procedure is illustrated in Figure \ref{fig:X_construction}. A changepoint is then flagged at update $r$ if $\kappa$ sequential outliers are detected under this procedure. If a changepoint is detected at $r^*$, under the assumption of a subsequent $B_2$ \textcolor{red}{time steps} of stationarity, we define
\begin{align}
    \mathcal{X}_{km}^{r^* + B_2} &:= \left\{p_{km}^{r^* + B_2,1},\dots,p_{km}^{r^* + B_2,B_2} :p_{km}^{r^* + B_2, \ell} = q^{(r^* + \ell)}(\lambda_{km}),\ \ell=1,\dots,B_2\right\}, \notag
\end{align}
and the above procedure is repeated. 
}
\begin{figure}
    \centering
    \includegraphics[width=0.4\textwidth]{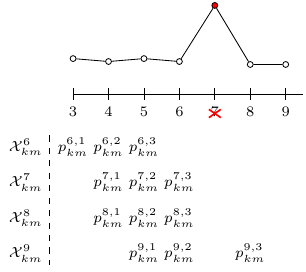}
    \caption{An illustration of how $\mathcal{X}_{km}^r$ is updated for a stream with an outlier at $r=7$.}
    \label{fig:X_construction}
\end{figure}
% In particular, we write $y^{r,s}_{km} := \mathrm{KL}[q^{(r)}(\lambda_{km})\ \|\  q^{(r-s)}(\lambda_{km})]$ and construct $\mathcal{Y}_{km}^s := \{y^{r,s}_{km}\}_{r\geq B_1 + B_2 + s}$. A changepoint is flagged at $t = r\Delta$ if for all $s=1,2,\dots,\kappa$, we have: 
% \begin{equation}
% |y_{km}^{r+s-1, s} - \med(\mathcal{Y}_{km}^s)| > L\cdot \mathrm{MAD}.    
% \end{equation}
% It must be remarked that if we choose a maximal lag of $\kappa$, we can only flag a changepoint a minimum of $\kappa - 1$ time steps after the changepoint has occurred. 

We now provide an argument for why the stream must be reset post-change. Suppose a change to the latent rate $\lambda$ occurs at some time $L_{r} < t' < L_{r+1}$, and that pre-change $\lambda \sim \mathrm{Gamma}(\alpha, \beta)$, whereas $\lambda \sim \mathrm{Gamma}(\alpha^\prime, \beta^\prime)$ post-change, for $\alpha\neq\alpha^\prime$ and $\beta\neq\beta^\prime$. Pre-change, our algorithm provides two consecutive CAVI estimates $q^{(r-1)}(\lambda)$ and $q^{(r)}(\lambda)$, with shapes and rates $\alpha^{(r-1)}, \beta^{(r-1)}$ and $\alpha^{(r)}, \beta^{(r)}$, respectively, and post-change, we have CAVI estimates $q^{(r+1)}(\lambda)$ and $q^{(r+2)}(\lambda)$ with shapes and rates $\alpha^{(r+1)}, \beta^{(r+1)}$ and $\alpha^{(r+2)}, \beta^{(r+2)}$. For ease of analysis, we suppose that %the error between the consecutive estimates doesn't change with the change, so 
$\alpha^{(r)} - \alpha^{(r-1)} = \alpha^{(r+2)} - \alpha^{(r+1)} = \Delta_\alpha$, and $\beta^{(r)} - \beta^{(r-1)} = \beta^{(r+2)} - \beta^{(r+1)} = \Delta_\beta$. 

We consider the ratio of the KL-divergences between consecutive estimates pre and post-change, and expand as $\Delta_\alpha \to 0$ to get
\begin{align}
    &\frac{\mathrm{KL}(q^{(r+2)}\ \|\ q^{(r+1)})}{\mathrm{KL}(q^{(r)}\ \|\ q^{(r-1)})} =  \frac{\alpha^{(r+1)}\left(\Delta_\beta / \beta^{(r+1)} - \log\left(1 + \Delta_\beta / \beta^{(r+1)}\right)\right)}{\alpha^{(r-1)}\left(\Delta_\beta / \beta^{(r-1)} - \log\left(1 + \Delta_\beta / \beta^{(r-1)}\right)\right)} \mathcal{O}(\Delta_\alpha).
    \label{eqn:KL_expansion_alpha}
\end{align}
Similarly, taking $\Delta_\beta \to 0$ yields
\begin{align}
    &\frac{\mathrm{KL}(q^{(r+2)}\ \|\ q^{(r+1)})}{\mathrm{KL}(q^{(r)}\ \|\ q^{(r-1)})} = \frac{\log\left\{\Gamma(\alpha^{(r+1)})/\Gamma(\alpha^{(r+1) }+ \Delta_\alpha)\right\} + \Delta_\alpha\psi(\alpha^{(r+1)})}{\log \left\{\Gamma(\alpha^{(r-1)})/\Gamma(\alpha^{(r-1)} + \Delta_\alpha)\right\} + \Delta_\alpha\psi(\alpha^{(r-1)})} +\mathcal{O}(\Delta_\beta).
    \label{eqn:KL_expansion_beta}
\end{align}
% -(\alpha_1 + \Delta_\alpha)\log\left(1 + \frac{\Delta_\beta}{\beta_1}\right) + \alpha_1\frac{\Delta_\beta}{\beta_1} + \mathcal{O}(\Delta_\alpha^2).
% If we expand similarly as $\Delta_\beta \to 0$, we obtain
% \begin{equation}
%     \mathrm{KL}(p_1\ \|\ p_2) = - \Delta_\alpha\frac{\Delta_\beta}{\beta_1} - \log \frac{\Gamma(\alpha_1)}{\Gamma(\alpha_1 + \Delta_\alpha)} - \Delta_\alpha\psi(\alpha_1) + \mathcal{O}(\Delta_\beta^2).
%     \label{eqn:KL_expansion_beta}
% \end{equation}
Observation of \eqref{eqn:KL_expansion_alpha}  indicates that to first order the ratio depends upon only the relative size of the difference between their rates, that is $\Delta_\beta / \beta$, and not directly upon the magnitude of $\beta$. This is not problematic only if the magnitude of $\Delta_\beta$ changes linearly with $\beta$. Also, \eqref{eqn:KL_expansion_beta} shows that the ratio depends directly on the magnitude of the scale $\alpha$. It follows that after the identification of a changepoint at $t = r_c\Delta$, $r_c\in\mathbb{N}$, the stream of realised KL-divergence values must be discarded, and the algorithm wait a further $B_2$ \textcolor{red}{time steps} to obtain enough samples for the new set $\mathcal{Y}_{km}^s := \{y_{km}^{r,s}\}_{r \geq r_c+ B_2 + s}$ which is used to detect any subsequent changes. A further benefit of consider a maximal lag $\kappa > 1$ is that more KL-divergence samples are obtained for the same number of observations.

To detect changes in the group memberships, the approximation $\prod_{i\in\mathcal{V}}\hat{q}(z_i)$ to the posterior of the latent group memberships in \eqref{eqn:online_z_update} provides a node-level vector $\tau_i$ ($i=1,...,N)$ of group membership probabilities. Computing $\argmax_k \tau_{ik}$ provides a group assignment for node $i$, and by comparing assignments between runs, one has a natural way of flagging changes. However, this approach is unsatisfactory as it does not account for the magnitude of the change in probability. 
% For example, suppose that $K=2$, and $\tau_{i1} = 1 - \tau_{i2}$ alternates between slightly above and slightly below 0.5 on consecutive updates. This na\"ive approach would flag multiple changes, despite there being only small changes in the probability of group assignment. 
The KL-divergence also cannot be used for flagging group changes, since $\mathrm{KL}(p_1\ \|\ p_2)$ is only defined only when $p_2 > 0$ for all values in the support of $p_1$. %For $q_1,q_2\in \Delta^n$, the $n$-dimensional simplex, %the KL-divergence between $X_1 \sim \mathrm{Categorical}(q_1)$ and $X_2 \sim \mathrm{Categorical}(q_2)$, with probability distributions $p_1$ and $p_2$, respectively, is
%\begin{equation}
%    \mathrm{KL}(p_1 \ \|\ p_2) = \sum_{i=1}^n q_{1i}\log\frac{q_{1i}}{q_{2i}}.
%\label{eqn:KL_categorical}
%\end{equation}
%When $q_{1i} \neq q_{2i},\ q_{2i} = 0$, the relevant term is defined as $\infty$, and for $q_{1i} = q_{2i} = 0$, it takes the value 0. Suppose $K=2$ and $\tau^{(r-1)}(z_i) = (1,0)$ and $\tau^{(r)}(z_i) = (0.99, 0.01)$, then $\mathrm{KL}(\hat{q}^{(r)}(z_i) \ \|\ \hat{q}^{(r-1)}(z_i)) = \infty$.
Instead, we utilise the Jensen-Shannon divergence \citep[JS;][]{lin1991}, defined for two distributions $p_1,p_2$ as 
\begin{equation}
    \mathrm{JS}(p_1\ \|\ p_2) = \frac{1}{2}\mathrm{KL}(p_1\ \|\ p) + \frac{1}{2}\mathrm{KL}(p_2\ \|\ p),
    \label{eqn:JS_div}
\end{equation}
where $p := (p_1 + p_2)/2$ is a 50-50 mixture distribution of $p_1$ and $p_2$. The JS-divergence is symmetric in its arguments, and avoids the undefined values encountered with the KL-divergence. %For the categorical random variables $X_1$ and $X_2$ defined above, this is given by
For $q_1,q_2\in \Delta^n$, where $\Delta^n$ denotes the $n$-dimensional simplex, the JS-divergence between discrete random variables $X_1 \sim \mathrm{Categorical}(q_1)$ and $X_2 \sim \mathrm{Categorical}(q_2)$, with probability mass functions $p_1$ and $p_2$ respectively, takes the following form:
\begin{align}
    \mathrm{JS}(p_1 \ \|\ p_2) &= \frac{1}{2}\left(\sum_{i=1}^n q_{1i}\log\frac{q_{1i}}{(q_{1i} + q_{2i})/2} +\sum_{i=1}^n q_{2i}\log\frac{q_{2i}}{(q_{1i} + q_{2i})/2}\right).
\label{eqn:JS_categorical}
\end{align}
%The values of this stream can be very small, and so a log transformation is taken. 
The MAD is then used to detect changes to the stream of logged values, \textcolor{red}{in an analogous way to the latent rates,} but with two additional requirements on the group probabilities:
\begin{align}
&\argmax_k \tau^{(r)}_{ik} \neq \argmax_k \tau^{(r-s)}_{ik}, & \text{for all } s=1,\dots,\kappa,   \label{eqn:tau_stream_add_req_1}\\
&\argmax_k \tau^{(r-s_1)}_{ik} = \argmax_k \tau^{(r-s_2)}_{ik},  &\text{for all }s_1,s_2 \in \{1,\dots,\kappa\}. \label{eqn:tau_stream_add_req_2}
\end{align}
If the MAD condition and both \eqref{eqn:tau_stream_add_req_1} and \eqref{eqn:tau_stream_add_req_2} are met, then a changepoint is flagged. 
The conditions \eqref{eqn:tau_stream_add_req_1} and \eqref{eqn:tau_stream_add_req_2} ensure that changes are flagged only if the most likely group assignment changes after a window of $\kappa$ \textcolor{red}{time steps} where the most likely group assignment was stable. 
It should noted that for this stream, as the probabilities are constrained to sum to 1, we do not %need to 
reset these values after a changepoint is detected.\\ %\textcolor{red}{I don't think?}
\textcolor{red}{Initialisation of the approximate posterior parameters at update $r$ as the values obtained from update $r-1$ is a heuristic approach to avoid label switching that we demonstrate works well. An alternative approach would be to instead consider the set of permutations $S_K$ of $[K]$. At each update, one could minimise the KL/JS-divergence over all $\sigma \in S_K$, which would identify the relevant permutation should a label switch have occurred. This would correspond to calculating $\min_{\sigma \in S_K}\mathrm{KL}[q^{(r)}(\lambda_{km}) \mid \mid p_{\sigma(k)\sigma(m)}^{r,B_2}]$ and identify whether there has been label switching from that. However, such an approach to provide additional robustness would require $\mathcal{O}(K!)$ computations, which could be undesirable for an online algorithm.\\
Note that identifiability issues may arise in the case that a rate and membership change occur simultaneously. However, in simulations conducted with a small proportion of nodes changing groups ($\approx 25\%$) at the same time as a change to the latent rates, no identifiability issues were encountered.}

\begin{figure}
    \centering
    \begin{subfigure}[h!]{0.45\textwidth}
        \centering 
        \includegraphics[width=\textwidth]{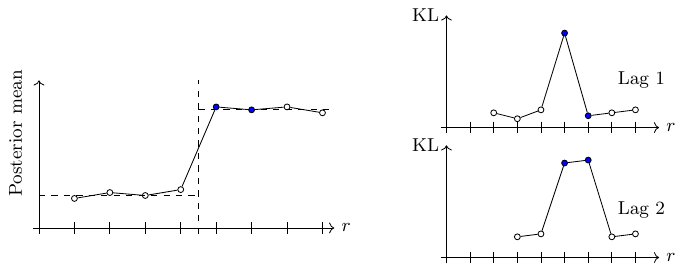}
        \caption{A true changepoint.}
        \label{fig:KL-lag-cp}
    \end{subfigure}
    \hspace{.025\textwidth}
    \unskip\ \vrule\ 
    \hspace{.025\textwidth}
    \begin{subfigure}[h!]{0.45\textwidth}
        \centering
        \includegraphics[width=\textwidth]{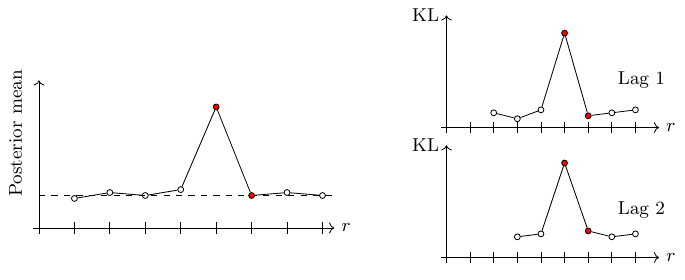}
        \caption{An outlier.}
        \label{fig:KL-lag-no-cp}
    \end{subfigure}
    \caption{%An illustration
    Illustration of the behaviour of the KL-divergence with different lags for a changepoint and an outlier.}
    \label{fig:kl_lag}
\end{figure}

\subsection{Online VB for the dynamic BHPP with an unknown underlying graph}
\label{subsec:unknown_adj}

Until now, it has been assumed that the graph is fully connected. We relax this assumption and extend the framework to the setting where $\boldsymbol A=\{a_{ij}\}\in\{0,1\}^{N\times N}$ is unknown. Assuming $a_{ij} = 1$ when there is no edge adversely affects estimation of the latent rates; instead of the absence of an event indicating the absence of an edge, it instead just leads to a lower rate estimate. This motivates a method that can take into account the possibility that $a_{ij}=0$. The work here builds upon the sparse setup of \cite{Matias2018-rx}.

For each $(i,j)\in\mathcal{R}$, an additional latent variable $a_{ij} \in \{0,1\}$ is introduced. A new latent group membership $z_i'\in\mathcal{K}' := \{1,\dots, K'\}$ is assigned to each $i\in\mathcal{V}$, and $a_{ij}$ is Bernoulli-distributed conditional upon $z_i'$ and $z_j'$, according to a stochastic blockmodel \citep[SBM;][]{Holland}. Each $a_{ij}$ captures the probability that $(i,j)\in \mathcal{R}$ is also an element of $\mathcal{E}$, and functions to reduce the contribution of edges with no arrivals to the estimates of the underlying rates. The new model is then expressed as:
\begin{align}
    x_{ij}(t) \mid a_{ij}, z_i, z_j, \lambda_{z_iz_j} &\sim \mathrm{Poisson}(a_{ij}\lambda_{z_iz_j}t), 
    &\hfill\text{for all } i,j\in\mathcal{V}, \label{eqn:poiss_network_adj} \\
    \lambda_{km} &\sim \mathrm{Gamma}\left(\alpha_{km}^0, \beta_{km}^0\right),
    &\text{for all } k,m \in \mathcal{K},\label{eqn:lam_prior_adj} \\
    a_{ij} \mid z_i', z_j', \rho_{z_i'z_j'} &\sim \mathrm{Bernoulli}(\rho_{z_i'z_j'}),
    &\text{for all } i,j \in \mathcal{V}, \label{eqn:a_prior_adj}\\
    \rho_{k'm'} & \sim \mathrm{Beta}\left(\eta_{k'm'}^0, \zeta_{k'm'}^0\right), 
    &\text{for all } k', m' \in \mathcal{K}',  \label{eqn:rho_prior_adj}\\
    z_i\mid\pi &\sim \mathrm{Categorical}(\pi), & \text{for all } i \in \mathcal{V}, \label{eqn:z_adj}\\ 
    z_i'\mid\mu &\sim \mathrm{Categorical}(\mu), 
    &\text{for all } i \in \mathcal{V}, \label{eqn:z_prime_adj}\\ 
    \pi &\sim \mathrm{Dirichlet}\left(\gamma^0\right), \label{eqn:pi_prior_adj}\\
    \mu &\sim \mathrm{Dirichlet}\left(\xi^0\right), \label{eqn:mu_prior_adj}
\end{align}
\textcolor{red}{We refer to this model as the stochastic blockmodel BHPP (SBM-BHPP).} We allow the most general setting in which the group structure that governs the rates is separate from that which determines edge connection probabilities. We assume that the graph is static, and so changes to the memberships that drive the edge-processes should not affect the determination of the graph adjacency matrix. The decoupling of the edge-process and edge-connection group memberships ensure this. Furthermore, as has been noted, SBMs provide good approximations to any exchangeable random graph model \citep{airoldi2013}, and they therefore offer a logical model for the graph generation mechanism. This new model structure is represented in graphical model form in Figure~\ref{fig:dag_model_adj}.

\begin{figure*}
    \centering
    \includegraphics[width=0.7\textwidth]{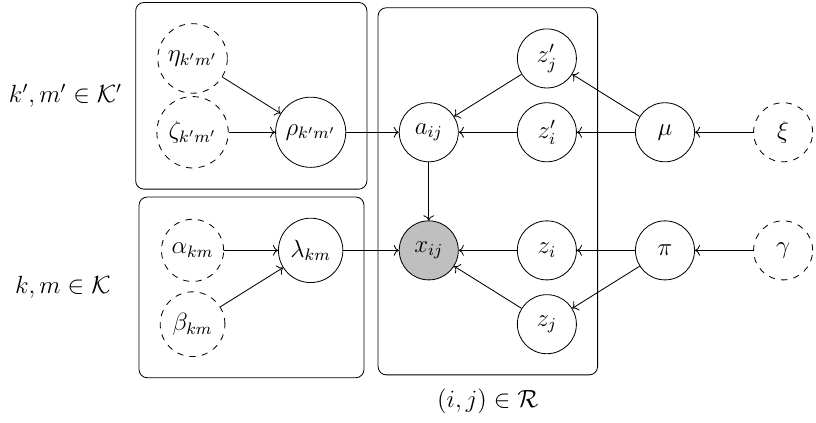}
    \caption{A directed acyclic graphical model representation of the model in  \eqref{eqn:poiss_network_adj}-\eqref{eqn:mu_prior_adj}.}
    \label{fig:dag_model_adj}
\end{figure*}

Given the graph is considered to be static, if $x_{ij}(t)>0$ for some $t$, then $(i,j) \in \mathcal{E}$ (corresponding to $a_{ij}=1$) with probability one. Otherwise, if $x_{ij}(t) = 0$ then we either have no edge (corresponding to $a_ij = 0$), or a non-null intensity process with no events in $(0,t]$. The difference with the Section \ref{sec:model} is that we now have the additional local latent variables, $z$, $z'$ and $\boldsymbol{A}$, all of which grow with $N$ (linearly for $z$ and $z'$, and quadratically for $\boldsymbol{A}$). 

To implement variational inference, the full conditional distribution $p(a_{ij}\mid z_i, z_j,z_i',z_j',\lambda_{z_iz_j},\rho_{z_i'z_j'},\mathcal{D}_{1:r})$ must be written down exactly, as per~\eqref{eq:cavi_update_sol}. Using the fact that inter-arrival times from a Poisson process are exponentially distributed, \cite{Matias2018-rx} shows that the conditional posterior probability of an edge $(i,j) \in \mathcal{E}$ in the case of static $\lambda$ takes the form 
\begin{multline}
    p(a_{ij} = 1\mid\mathcal{D}_{1:r}, z_{i}=k, z_{j} = m, z_{i}^\prime=k^\prime, z_{j}^\prime = m^\prime,
    \lambda_{km}, \rho_{k'm'}) =  \mathbb{I}\{x_{ij}(r\Delta) > 0\} \\ + \frac{\mathbb{I}\{x_{ij}(r\Delta) = 0\}\rho_{k'm'}\exp\left(-\lambda_{km}r\Delta\right)}{1 - \rho_{k'm'} + \rho_{k'm'}\exp\left(-\lambda_{km}r\Delta\right)}.
    \label{eqn:true_a_posterior}
\end{multline}
Since $\boldsymbol A$ is static, the posterior for each $a_{ij}$ should be computed using the full data $\mathcal{D}_{1:r}$. Note that the conditional distribution in ~\eqref{eqn:true_a_posterior} depends on $\mathcal{D}_{1:r}$ only via the counting process $x_{ij}(\cdot)$ at time $r\Delta$.
For the model in \eqref{eqn:poiss_network_adj}-\eqref{eqn:mu_prior_adj}, the true posterior distribution after $r$ batches factorises as
\begin{align}
    p(\boldsymbol{A},z,z',\theta  \mid\mathcal{D}_{1:r}) \propto p(\mathcal{D}_r\mid\boldsymbol{A},z,z'\theta)\times p(\boldsymbol{A}, z,z',\theta\mid\mathcal D_{1:(r-1)}),
    %p(\boldsymbol{A}\mid\mathcal{D}_{1:r},z,\theta) \times p(z,\theta\mid\mathcal{D}_{1:r}) \\
    %&\propto \prod_{(i,j)\in\mathcal{R}}p(a_{ij}\mid\mathcal{D}_{1:r},z_i,z_j,\theta) \times p(\mathcal{D}_r\mid z,\theta) \times p(z,\theta\mid\mathcal{D}_{1:(r-1)}).
    \label{eqn:unknown_adj_posterior}
\end{align}
similarly to the posterior distribution in \eqref{eqn:posterior}.
Therefore, inference on this model can be performed in much the same way as inference on the original model. However, as the true conditional posterior for $\boldsymbol{A}$ is tractable, the update procedure is slightly different. Unlike \cite{Matias2018-rx}, we cannot simply take $q^{(r)}(a_{ij})$ to be Bernoulli distributed with probability as in \eqref{eqn:true_a_posterior} as $\lambda$ is, in general, not static. Instead, to account for variations in $\lambda$ over the full observation window, we propose an approximation to the posterior distribution of $a_{ij}$ at $t=r\Delta$, conditional upon $z_{i}=k, z_{j}=m, z_{i}^\prime=k^\prime, z_{j}^\prime = m^\prime$, of the form %, implying that the distribution can be updated 
$\hat q^{(r)}(a_{ij}) = \mathrm{Bernoulli}\left(\sigma_{ij}^{(r)}\right)$, where
\begin{align}
    %\hat q^{(r)}(a_{ij}) = \mathrm{Bernoulli}\left(\sigma_{ij}^{(r)}\right), \ \mathrm{with} \ 
    \sigma_{ij}^{(r)} = \mathbb{I}\{x_{ij}(r\Delta) > 0\}
 + \frac{\mathbb{I}\{x_{ij}(r\Delta) = 0\}\hat{\rho}_{k'm'}^{(r)}\exp\left(-\Delta\sum_{\ell=0}^{r}\hat{\lambda}^{(\ell)}_{km}\right)}{1 - \hat{\rho}_{k'm'}^{(r)} + \hat{\rho}_{k'm'}^{(r)}\exp\left(-\Delta\sum_{\ell=0}^{r}\hat{\lambda}^{(\ell)}_{km}\right)}.\qquad 
    \label{eqn:sigma_update}
\end{align}
Here, $\hat{\rho}_{k'm'}^{(\ell)}$ and $\hat{\lambda}_{km}^{(\ell)}$ denote the mean of our approximate posterior distribution for $\rho_{k'm'}^{(\ell)}$ and $\lambda_{km}^{(\ell)}$, respectively, at update $\ell$. At update $r$, we propose a variational approximation of the form $q^{(r)}(\boldsymbol A, z,z', \theta) = \prod_{(i,j)\in\mathcal{R}}\mathrm{Bernoulli}(\sigma_{ij}^{(r)}) \times q(\theta, z, z')$, where $\sigma_{ij}^{(r)}$ is as in \eqref{eqn:sigma_update} and $q(\theta,z,z')\in\mathcal{F}^{2N+S}$. A two-step estimation procedure follows naturally, wherein $q(\theta,z,z')$ is approximated using CAVI as was done previously, and then $\sigma^{(r)}_{ij}$ is updated using \eqref{eqn:sigma_update} and the CAVI approximations to the posteriors of $\lambda$ and $\rho$. 

Note that since $a_{ij}$ for $i,j\in\mathcal{V}$, $z_i'$ for $i \in\mathcal{V}$, $\rho_{k'm'}$ for $k',m'\in\mathcal{K}'$ and $\mu$ are considered \textit{static} parameters, forgetting factors for these quantities are not required.

Consider the BHPP model with unknown graph structure, as in \eqref{eqn:poiss_network_adj}-\eqref{eqn:mu_prior_adj}, with global parameters $\theta=(\lambda,\pi,\mu,\rho)$ and local parameters $z$, $z'$ and $\boldsymbol{A}$. At step $r-1$, we approximate the posterior density $p(\boldsymbol{A},z,z',\theta\mid\mathcal D_{1:(r-1)})$ by $\hat{q}^{(r-1)}(\boldsymbol A, z, z', \theta)$, which is the product of the optimal CAVI solution $\hat{q}^{(r-1)}(z,z',\theta)$ and the fixed-form approximation $\hat{q}^{(r-1)}(\boldsymbol A)$ of \eqref{eqn:sigma_update}. For the BHPP model with known graph structure, all latent variables could change, and thus in \eqref{eq:tempered_approximation} all components were tempered before being passed as the prior for \textcolor{red}{time steps} $r$. However, in the case of an unknown graph structure, as the graph is assumed to be static, only the components of the dynamic latent variables are tempered when constructing the prior for \textcolor{red}{time step} $r$. We pass through as the prior for \textcolor{red}{time step} $r$ the partially tempered density 
\begin{align}
    \hat{q}^{(r-1)}_\delta(\boldsymbol{A},z,\theta) &= 
    \prod_{k,m\in\mathcal{K}} \hat{q}^{(r-1)}(\rho_{km}) \times \prod_{i\in\mathcal{V}} \hat{q}^{(r-1)}(z_i') \times \prod_{(i,j)\in\mathcal{R}} \hat{q}^{(r-1)}(a_{ij}) \times
    \hat{q}^{(r-1)}(\mu) \\
    & \times\prod_{i\in\mathcal{V}} \hat{q}^{(r-1)}_{\delta_z}(z_i) \times \prod_{k,m\in\mathcal{K}} \hat{q}^{(r-1)}_{\delta_\lambda}(\lambda_{km}) \times \hat{q}^{(r-1)}_{\delta_\pi}(\pi).\qquad 
    \label{eq:prior_structure_adj}
\end{align}
    %Consider the model in but with priors of the form
We assume that each component of \eqref{eq:prior_structure_adj} takes the same form as its corresponding complete conditional distribution under the BHPP model with unknown graph structure of \eqref{eqn:poiss_network_adj}-\eqref{eqn:mu_prior_adj}. Additionally, the distributions for $\lambda$, $\pi$ and $z$ are raised to a power and normalised. In summary:
\begin{align*}
    \hat{q}^{(r-1)}(a_{ij}) &= \mathrm{Bernoulli}\left(\sigma_{ij}^{{(r-1)}}\right), &\text{for all }i,j\in\mathcal{V}, \\
    \hat{q}^{(r-1)}_{\delta_\lambda}(\lambda_{km}) &=\frac{1}{C_{\lambda_{km},r-1}}\mathrm{Gamma}\left(\alpha_{km}^{(r-1)}, \beta_{km}^{(r-1)}\right)^{\delta_\lambda}, 
    &\text{for all }k,m\in\mathcal{K},\\
    \hat{q}^{(r-1)}(\rho_{k'm'}) &= \mathrm{Beta}\left(\eta^{(r-1)}_{k'm'}, \zeta^{(r-1)}_{k'm'}\right), &\text{for all }k',m'\in\mathcal{K}', %^{\delta_\rho} 
    \\
    \hat{q}^{(r-1)}(z_i) &= \frac{1}{C_{z_i,r-1}}\mathrm{Categorical}\left(\tau_i^{(r-1)}\right)^{\delta_z}, &\text{for all }i\in\mathcal{V},
    \\
    \hat{q}^{(r-1)}(z_i') &= \mathrm{Categorical}\left(\nu_i^{(r-1)}\right), &\text{for all }i\in\mathcal{V}, \\
    \hat{q}^{(r-1)}_{\delta_\pi}(\pi) &= \frac{1}{C_{\pi,r-1}}\mathrm{Dirichlet}\left(\gamma^{(r-1)}\right)^{\delta_\pi}, \\
    \hat{q}^{(r-1)}(\mu) &= \mathrm{Dirichlet}\left(\xi^{(r-1)}\right). 
\end{align*}
Suppose that data is again observed on the interval $I_r := (L_{r-1}, L_r]$, for $r\in\mathbb{N}$. Under the prior structure in \eqref{eq:prior_structure_adj}, the CAVI sequential updates take the form:
\begin{enumerate}
    \item $\hat{q}^{(r)}(\lambda_{km}) = \mathrm{Gamma}\left(\alpha_{km}^{(r)}, \beta_{km}^{(r)}\right)$ where for all $k,m\in\mathcal{K}$ we define $\alpha_{km}^{(r)}$ and $\beta_{km}^{(r)}$ as
    \begin{align}
    \alpha_{km}^{(r)} &= \delta_\lambda\left(\alpha^{(r-1)}_{km} - 1\right) \sum_{(i,j)\in\mathcal{R}}\tau^{(r-1)}_{ik}\tau^{(r-1)}_{jm}\sigma_{ij}^{(r-1)}x_{ij}^{(r)} + 1, \\
    \beta_{km}^{(r)} &= \delta_\lambda\beta^{(r-1)}_{km} + \sum_{(i,j)\in\mathcal{R}}\tau^{(r-1)}_{ik}\tau^{(r-1)}_{jm}\sigma_{ij}^{(r-1)}. \notag\label{eqn:lam_update_incomplete}
    \end{align}
    \item $\hat{q}^{(r)}(\rho_{k'm'}) = \mathrm{Beta}\left(\eta^{(r)}_{k'm'}, \zeta^{(r)}_{k'm'}\right)$ where for all $k',m'\in\mathcal{K}'$, we define $\eta_{k'm'}^{(r)}$ and $\zeta_{k'm'}^{(r)}$ as
    \begin{align}
    \eta_{k'm'}^{(r)} &= %\delta_\rho
    \eta_{k'm'}^{(r-1)} + \sum_{(i,j)\in\mathcal{R}}\nu^{(r-1)}_{ik'}\nu^{(r-1)}_{jm'}\sigma_{ij}^{(r-1)} + 1, \\
    \zeta_{k'm'}^{(r)} &= %\delta_\rho
    \zeta_{k'm'}^{(r-1)} + \sum_{(i,j)\in\mathcal{R}}\nu^{(r-1)}_{ik'}\nu^{(r-1)}_{jm'}\left(1 - \sigma_{ij}^{(r-1)}\right). \label{eqn:rho_update_incomplete}
    \end{align}
    \item $\hat{q}^{(r)}(z_i) = \mathrm{Categorical}\left(\tau_i^{(r)}\right)$ for all $ i\in\mathcal{V}$ where $\tau_i^{(r)} = (\tau_{i1}^{(r)},\dots,\tau_{iK}^{(r)})$, with $\sum_{k\in\mathcal{K}} \tau_{ik}^{(r)} = 1$, satisfies the relation:
    \begin{align}
    \tau_{ik}^{(r)} &\propto \exp\Bigg\{\delta_z\left[\psi\left(\gamma^{(r-1)}_k\right) - \psi\left(\sum_{\ell=1}^K \gamma^{(r-1)}_\ell\right)\right] + \sum_{j\in\mathcal{V}}\sum_{m\in\mathcal{K}} \tau_{jm}^{(r)} \bigg(1 - \mathbb{I}_{\{i\}}(j)\mathbb{I}_{\{k\}}(m)\bigg)  \\
    &\times \Bigg[-\Delta\left(\sigma_{ij}^{(r-1)}\frac{\alpha^{(r)}_{km}}{\beta^{(r)}_{km}} + \sigma_{ji}^{(r-1)}\frac{\alpha^{(r)}_{mk}}{\beta^{(r)}_{mk}}\right) + x_{ij}(I_r)\ \sigma_{ij}^{(r-1)}\left(\psi\left(\alpha^{(r)}_{km}\right) - \log\left(\beta^{(r)}_{km}\right)\right) \\ 
    &+ x_{ji}(I_r)\ \sigma_{ji}^{(r-1)}\left(\psi\left(\alpha^{(r)}_{mk}\right) - \log\left(\beta^{(r)}_{mk}\right)\right)\Bigg] + \sigma_{ii}^{(r-1)}\bigg[x_{ii}(I_r)\left(\psi\left(\alpha^{(r)}_{kk}\right) - \log\left(\beta^{(r)}_{kk}\right)\right) \\ & -\Delta\frac{\alpha_{kk}}{\beta_{kk}}\bigg]\Bigg\}.
    \end{align}
    \item $\hat{q}^{(r)}(z_i') = \mathrm{Categorical}\left(\nu_i^{(r)}\right)$ for all $ i\in\mathcal{V}$ where $\nu_i^{(r)} = (\nu_{i1}^{(r)},\dots,\nu_{iK'}^{(r)})$, with $\sum_{k'\in\mathcal{K}'} \nu_{ik'}^{(r)} = 1$, satisfies the relation:
    \begin{align}
    \nu_{ik'}^{(r)} &\propto \exp\Bigg\{\psi\left(\xi^{(r)}_{k'}\right) - \psi\left(\sum_{\ell'=1}^{K'} \xi^{(r)}_{\ell'}\right) + \sum_{j\in\mathcal{V}}\sum_{m'\in\mathcal{K}'} \nu_{jm'}^{(r)}\bigg(1 - \mathbb{I}_{\{k'\}}(m')\mathbb{I}_{\{i\}}(j)\bigg) \\&\times \Bigg[\sigma_{ij}^{(r)}\left(\psi\left(\eta^{(r)}_{k'm'}\right) - \psi\left(\eta^{(r)}_{k'm'} + \zeta^{(r)}_{k'm'}\right)\right) + \left(1 - \sigma_{ij}^{(r)}\right)\left(\psi\left(\zeta^{(r)}_{k'm'}\right) - \psi\left(\zeta_{k'm'}^{(r)} + \eta_{k'm'}^{(r)}\right)\right)\\ &+ \sigma_{ji}^{(r)}\left(\psi\left(\eta^{(r)}_{m'k'}\right) - \psi\left(\eta^{(r)}_{m'k'} + \zeta^{(r)}_{m'k'}\right)\right) + \left(1 - \sigma_{ji}^{(r)}\right)\left(\psi\left(\zeta^{(r)}_{m'k'}\right) - \psi\left(\zeta_{m'k'}^{(r)} + \eta_{m'k'}^{(r)}\right)\right)\Bigg] \\
    &+ \sigma_{ii}^{(r)}\bigg[\psi\left(\eta^{(r)}_{k'k'}\right) - \psi\left(\eta^{(r)}_{k'k'} + \zeta^{(r)}_{k'k'}\right)\bigg]+ \left(1 - \sigma_{ii}^{(r)}\right)\bigg(\psi\left(\zeta^{(r)}_{k'k'}\right) - \psi\left(\eta^{(r)}_{k'k'} + \zeta^{(r)}_{k'k'}\right)\bigg)\Bigg\}.
    \end{align}
    \item $\hat{q}^{(r)}(\pi) = \mathrm{Dirichlet}\left(\gamma^{(r)}\right)$ where for all $k\in\mathcal{K}$, we define $\gamma^{(r)}_k$ as:
    \begin{align}
    \gamma^{(r)}_k = \delta_\pi\left(\gamma_k^{(r-1)} - 1\right) + \delta_z\sum_{i\in\mathcal{V}}\tau^{(r)}_{ik} + 1. \label{eqn:pi_update_incomplete}
    \end{align}
    \item $\hat{q}^{(r)}(\mu) = \mathrm{Dirichlet}\left(\xi^{(r)}\right)$ where for all $k'\in\mathcal{K}'$, we define $\xi_{k'}^{(r)}$ as:
    \begin{align}
    \xi^{(r)}_{k'} = \xi_{k'}^{(r-1)}+ \sum_{i\in\mathcal{V}}\nu^{(r)}_{ik'}.\label{eqn:nu_update_incomplete}
    \end{align}
    \item $\hat q^{(r)}(a_{ij}) = \mathrm{Bernoulli}\left(\sigma_{ij}^{(r)}\right)$, where $\sigma_{ij}^{(r)}$ is defined in \eqref{eqn:sigma_update} for all $i,j \in\mathcal{V}$.
\end{enumerate}

\subsection{Online VB for the dynamic BHPP with an unknown number of groups}
\label{subsec:unknown_groups}

In real-world applications, the number of groups $K$ is usually unknown, and it must be estimated from the observed data. To address this issue, we adopt a Bayesian non-parametric approach, proposing a Dirichlet process prior \citep{ferguson1973} on the group memberships. 
In particular, we replace the Dirichlet prior distributions in \eqref{eqn:pi_prior} and \eqref{eqn:pi_prior_adj} with a Griffiths-Engen-McCloskey \citep[GEM;][]{pitman2002} prior distribution, with parameter $\nu$, written $\pi\sim\mathrm{GEM}(\nu)$. This prior distribution corresponds to an infinite limit of a Dirichlet distribution: $\mathrm{GEM}(\nu)=\lim_{K\to\infty}\mathrm{Dirichlet}(\nu \boldsymbol{1}_K/K)$, where $\boldsymbol{1}_K$ is the vector of ones of length $K$. The full model becomes: 
\begin{align}
    x_{ij}(t) \mid z_i, z_j, \lambda_{z_iz_j} &\sim \mathrm{Poisson}(\lambda_{z_iz_j}t),&\text{for all } (i,j)\in\mathcal{E}, \label{eqn:poisson_GEM} \\
    \lambda_{km} &\sim \mathrm{Gamma}(\alpha, \beta),&\text{for all } k,m\in\mathbb{N}, \label{eqn:lambda_GEM} \\
    z_i\mid\pi &\sim \mathrm{Categorical}(\pi), &\text{for all } i\in\mathcal{V}, \label{eqn:z_GEM}\\
    \pi & \sim \mathrm{GEM}(\nu),\label{eqn:u_GEM}
\end{align}
where $\pi$ represents an infinite sequence $\pi_1,\pi_2,\dots$ such that $\pi_k\geq0$ for all $k=1,2,\dots$ and $\sum_{k=1}^\infty\pi_k=1$. 
The GEM prior distribution also corresponds to the distribution of proportions obtained under a stick-breaking representation \citep{sethuraman1994} of a Dirichlet process. Therefore, the proportions $\pi$ can be reparametrised as a product of variables $u_1,u_2,\dots\in[0,1]$ drawn from independent beta distributions, as follows:
\begin{align}
    u_k \sim \mathrm{Beta}(1,\nu), \quad
    \pi_k := u_k\prod_{\ell=1}^{k-1}(1 - u_\ell),
    \text{for } k=1,2,\dots.&\qquad
    \label{eq:stick_break}
\end{align}
This decomposition is particularly useful to derive an online variational inference algorithm for the BHPP model with GEM priors (GEM-BHPP), following \cite{blei2006}. In particular, within a mean-field approximation $q(\lambda,u,z)=q(\lambda)\times q(u)\times q(z)$ for the posterior distribution $p(\lambda,u,z\mid\mathcal D_{1:r})$, a variational approximation $q(u)=\prod_{\ell=1}^\infty q(u_\ell)$ is posited directly on $u_1,u_2,\dots$, rather than $\pi$. This approximation is truncated at a level $L\in\mathbb N$, implying that $q(u_\ell)=\delta_0(u_\ell)$ for $\ell=L+1,L+2,\dots$, effectively resulting in an $L$-dimensional probability vector from \eqref{eq:stick_break}.  

It should be pointed out how this approach differs from the model set-up we considered earlier.
In the GEM-BHPP, we make no truncation in the model set-up, but rather only in the variational approximation to it. This is in contrast to the method we presented in Sections~\ref{sec:model}, wherein one would be using a finite-dimensional Dirichlet distribution directly in the prior structure.

We introduce the notation $\omega_i^{(r)}$, $\nu_i^{(r)}$, $\alpha_{km}^{(r)}$ and $\beta_{km}^{(r)}$, for $i,k,m\in\{1,\dots,L\}$. If we define $\omega^{(0)}_i \equiv 1$, $\nu_i^{(0)} \equiv \nu$ for all $i\in\{1,\dots,L\}$ and $\alpha_{km}^{(0)} \equiv \alpha$ and $\beta_{km}^{(0)}\equiv \beta$, for all $k,m\in\{1,\dots,L\}$, the first time step takes the form of the model specified in  \eqref{eqn:poisson_GEM}--\eqref{eqn:u_GEM}. Using this notation, we can formulate our inference procedure based on the same process of sequential CAVI updates. 

Consider the GEM-BHPP model in \eqref{eqn:u_GEM}--\eqref{eqn:poisson_GEM} with a stick-breaking reparametrisation of $\pi$ as in \ref{eq:stick_break}, global parameters $\theta=(\lambda,u)$ and local parameters $z$. The posterior distribution $p(\theta, z\mid \mathcal D_{1:(r-1)})$ is approximated by the optimal CAVI solution $\hat{q}^{(r-1)}(\theta,z)$. Similarly to \eqref{eq:tempered_approximation}, this CAVI approximation is tempered to form $\hat{q}^{(r-1)}_\delta(\theta,z)$, taking the following product form:
\begin{align}
    \hat{q}^{(r-1)}_\delta(\theta,z) =  \prod_{k=1}^L\prod_{m=1}^L \hat{q}^{(r-1)}_{\delta_\lambda}(\lambda_{km}) \times\prod_{i\in\mathcal{V}} \hat{q}^{(r-1)}_{\delta_z}(z_i) \times \prod_{\ell=1}^L\hat{q}^{(r-1)}_{\delta_u}(u_\ell). \qquad
    \label{eq:prior_structure_gem}
\end{align}
This tempered density is then passed through as the prior for \textcolor{red}{time step} $r$. We assume that each component of \eqref{eq:prior_structure_gem} takes the same form of its corresponding complete conditional distribution under the GEM-BHPP model in \eqref{eqn:poisson_GEM}--\eqref{eqn:u_GEM}, raised to a power and normalised: 
\begin{align}
    &\hat{q}^{(r-1)}_{\delta_\lambda}(\lambda_{km}) = \frac{1}{C_{\lambda_{km},r-1}}%\prod_{k,m\in\mathcal{K}}
    \mathrm{Gamma}\left(\alpha_{km}^{(r-1)}, \beta_{km}^{(r-1)}\right)^{\delta_\lambda}, &\text{for all }k,m\in\{1,\dots,L\}, \\ 
    &\hat{q}^{(r-1)}_{\delta_z}(z_i) = \frac{1}{C_{z_i,r-1}}\mathrm{Categorical}(\tau^{(r-1)})^{\delta_z},&\text{for all }i\in\mathcal{V}, \\
    &\hat{q}^{(r-1)}_{\delta_u}(u_\ell) = \frac{1}{C_{u_\ell,r-1}}\mathrm{Beta}\left(\omega_\ell^{(r-1)}, \nu_\ell^{(r-1)}\right)^{\delta_u},&\text{for all } \ell\in\{1,\dots,L\},
\end{align}
where $C_{\lambda_{km},r-1},C_{z_i,r-1}$ and $C_{u_\ell,r-1}$ are normalising constants, ensuring that the densities are valid, and $\delta=(\delta_\lambda,\delta_z,\delta_u)\in (0,1]^3$ are temperature parameters specific to each class of latent variables.
Suppose that data is observed on the interval $I_r := (L_{r-1}, L_r]$, for $r\in\mathbb{N}$, and denote the count on edge $(i,j)$ during this interval by $x_{ij}(I_r)$. Under the prior structure in \eqref{eq:prior_structure_gem}, the CAVI sequential updates then take the form
\begin{enumerate}
    \item $\hat{q}^{(r)}(\lambda_{km}) = \mathrm{Gamma}(\alpha^{(r)}_{km}, \beta^{(r)}_{km})$ for all $k,m\in \mathcal{K}$ where $\alpha^{(r)}_{km}, \beta^{(r)}_{km}$ are defined as:
    \begin{align}
    \alpha_{km}^{(r)} &= \delta_\lambda\left(\alpha_{km}^{(r-1)} - 1\right) + \sum_{(i,j)\in\mathcal{E}}\tau_{ik}^{(r-1)}\tau_{jm}^{(r-1)}x_{ij}(I_r) + 1, \\
    \beta_{km}^{(r)} &= \delta_\lambda\beta_{km}^{(r-1)} + \Delta\sum_{(i,j)\in\mathcal{E}}\tau_{ik}^{(r-1)}\tau_{jm}^{(r-1)}.\qquad \label{eqn:online_lam_GEM}
    \end{align}
    \item $\hat{q}^{(r)}(z_i) = \mathrm{Categorical}\left(\tau_i^{(r)}\right)$ for all $i\in\mathcal{V}$ where $\tau_i^{(r)} = (\tau_{i1}^{(r)},\dots,\tau_{iL}^{(r)})$, with $\sum_{k\in\mathcal{L}} \tau_{ik}^{(r)} = 1$, satisfies the relation:
    \begin{align} 
    \tau_{ik}^{(r)} &\propto \exp\Bigg\{\delta_z\left[\psi\left(\omega_k^{(r-1)}\right)- \psi\left(\omega_k^{(r-1)} + \nu_k^{(r-1)}\right) + \sum_{\ell=1}^{k-1} \bigg(\psi\left(\nu_\ell^{(r-1)}\right) -\psi\left(\omega_\ell^{(r-1)} + \nu_\ell^{(r-1)}\right)\bigg)\right]\\
    & + \sum_{m=1}^{L}\Bigg[ \sum_{j:(i,j)\in\mathcal{E}} \tau_{jm}^{(r)}\bigg(1 - \mathbb{I}_{\{k\}}(m)\mathbb{I}_{\{i\}}(j)\bigg)\bigg(x_{ij}(I_r)\left[\psi\left(\alpha_{km}^{(r)}\right) - \log\left(\beta_{km}^{(r)}\right)\right]- \Delta\frac{\alpha_{km}^{(r)}}{\beta_{km}^{(r)}}\bigg)\\
    & + \sum_{j':(j',i)\in\mathcal{E}} \tau_{j'm}^{(r)}\bigg(1 - \mathbb{I}_{\{k\}}(m)\mathbb{I}_{\{i\}}(j')\bigg)\bigg(x_{j'i}(I_r)\left[\psi\left(\alpha_{mk}^{(r)}\right) - \log\left(\beta_{mk}^{(r)}\right)\right]- \Delta\frac{\alpha_{mk}^{(r)}}{\beta_{mk}^{(r)}}\bigg)\Bigg] \\
    & + x_{ii}(I_r)\left[\psi\left(\alpha_{kk}^{(r)}\right) - \log\left(\beta_{kk}^{(r)}\right)\right] - \Delta\frac{\alpha_{kk}^{(r)}}{\beta_{kk}^{(r)}}\Bigg\}. \qquad \label{eqn:tau_fp_GEM}\end{align}
    \item $\hat{q}^{(r)}(u_i) = \mathrm{Beta}\left(\omega^{(r)}_i,\nu^{(r)}_i\right)$ for all $i \in \{1,\dots,L\}$ where $\omega^{(r)}_i$ and $\nu^{(r)}_i$ are defined as: 
    \begin{align}
    \omega^{(r)}_i &= \delta_u\left(\omega^{(r-1)}_i - 1\right) + \delta_z\sum_{j\in\mathcal{V}}\tau_{ji}^{(r)} + 1, \\
    \nu^{(r)}_i &= \delta_u\left(\nu^{(r-1)}_i - 1\right) + \delta_z\sum_{j\in\mathcal{V}}\sum_{k=i+1}\tau_{jk}^{(r)} + 1.
    \label{eqn:online_pi_GEM}
\end{align}
\end{enumerate}

A subtle problem arises with the updates in \eqref{eqn:online_lam_GEM}. In the case that the truncation parameter, $L$, is larger than the true number of groups, the update procedure will allocate some groups zero nodes. If $\ell \in \{1,\cdots,L\}$ is one such group, then for all $i\in\mathcal{V}$ $\tau_{i\ell}^{(r)} = 0$ for every $r$ after the algorithm has converged, up until a change occurs. For $\delta_\lambda \in (0,1)$, it follows from the rate update in \eqref{eqn:online_lam_GEM} that $\lim_{r\to\infty}\beta_{\ell m}^{(r)} = \lim_{r\to\infty} \beta_{m \ell}^{(r)} = 0$ for all $m \in \{1,\cdots,L\}$, and thus that the posterior mean of $\lambda_{\ell m}$ and $\lambda_{m \ell}$ diverges. This divergence causes problems in \eqref{eqn:tau_fp_GEM} where this mean appears. To circumvent this problem, we must replace $\delta_\lambda$ with the set $\{\delta^{km}_\lambda\}_{k,m=1}^L$. That is, we introduce a specific forgetting factor for each group-to-group rate. For each, $k,m\in\{1,\dots,L\}$, we initially set $\delta^{km}_L \equiv \delta_\lambda$, but monitor the sum $\sum_{(i,j)\in\mathcal{E}} \tau_{ik}^{(r-1)}\tau_{jm}^{(r-1)}$ with increasing $r$. Then, for all $k,m\in \{1,\dots,L\}$ such that $\sum_{(i,j)\in\mathcal{E}} \tau_{ik}^{(r-1)}\tau_{jm}^{(r-1)} < \epsilon$, we set $\delta_{km} = 1$, where $\epsilon$ is some threshold. This intervention prevents the exponential decay of $\beta_{km}^{(r)}$ to 0 as $r\to\infty$. Once the threshold is exceeded, the relevant BFFs are returned to $\delta_\lambda$. Experimentation found that $0.1$ is a good choice for $\epsilon$. 

\begin{figure}
    \centering
    \includegraphics[width=0.4\textwidth]{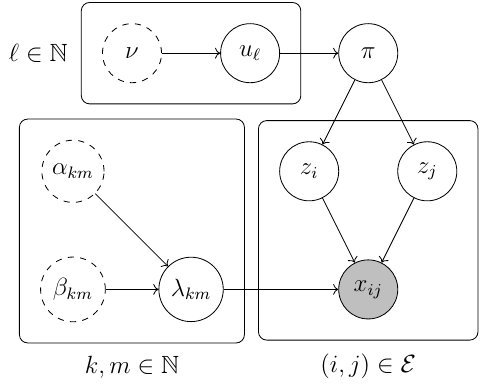}
    \label{fig:dag_model_GEM}
    \caption{A directed acyclic graph of the model given by  \eqref{eqn:poisson_GEM}-\eqref{eqn:u_GEM}.}
\end{figure}

\section{Simulation studies}
\label{sec:sims}

We evaluate the online changepoint detection algorithm of Section \ref{sec:online_VB} using simulated data. \textcolor{red}{Five} simulation studies are conducted. In Section \ref{subsec:sim1}, we examine membership and rate recovery in the case that $\boldsymbol{A}$ and $K$ are known when a varying proportion of nodes swap between the two groups. In Section \ref{subsec:sim2}, membership recovery is investigated when $\boldsymbol{A}$ is known, but $K$ changes throughout the observation window. In Section \ref{subsec:sim3}, we consider two the effect of a decreasing lag between two instantaneous changes to $\lambda$ when both $\boldsymbol{A}$ and $K$ are known. \textcolor{red}{In} Section \ref{subsec:sim4}, the effect of increasing sparsity in case of known $K$ but unknown $\boldsymbol{A}$ is examined. \textcolor{red}{Finally, in Section \ref{subsec:sim5}, we examine the recovery of group memberships with varying latent rates.}

Unless otherwise stated, we consider a network with $N=500$ nodes and $K=2$ latent groups, initialised with $\Pi = (0.6, 0.4)$. For convenience, we define the intra-inter-group rate matrix
\begin{equation}
\lambda_0 = \begin{pmatrix} 2 & 1 \\ 0.3 & 8 \end{pmatrix}. 
\label{eq:lambda0}
\end{equation}
The update interval is set to $\Delta = 0.1$ time units throughout, and every experiment is repeated 50 times, with the results averaged. We cycle 3 times over the CAVI and fixed point equations. All hyperparameters are initialised as 1, except for $\gamma$ and $\xi$, which are initialised uniformly on $[0.95, 1.05]^K$ and $[0.95, 1.05]^{K'}$, respectively. For all $i,j\in\mathcal{V}$, $\sigma_{ij}$ is initialised as $1/2$, and $\tau_i$ as $\boldsymbol{1}_K / K$, when $K$ is known, and as $ \boldsymbol{1}_L/L$ when we infer the number of groups at a truncation level of $L$. Unless otherwise stated, $\textcolor{red}{W_{JS}}$ is set to be 2 and $\textcolor{red}{W_{KL}}, B_1$ and $B_2$ to 10.

\subsection{Latent group membership recovery}
\label{subsec:sim1}
The first experiment takes a fully-connected network with rate matrix $\lambda_0$ as given in \eqref{eq:lambda0}. At time $t=3$, $P\%$ of nodes from group 1 swap to group 2, with $P \in \{1, 10, 25, 50, 75, 95\}$. To evaluate latent membership recovery, the inferred group memberships at time $t$ are compared to the true memberships using the adjusted rand index \citep[ARI,\ ][]{Hubert1985}. The ARI takes values between $-1$ and $1$, with a value of $1$ indicating perfect agreement (up to label switching), $0$ random agreement, and $-1$ complete disagreement. By computing the ARI at each update, rather than taking an average over $[0,t]$, we can examine the smoothness of recovery, and the reaction of the algorithm to changepoints. 

\begin{figure*}
     \centering
     \begin{subfigure}[t]{\textwidth}
         \centering
         \includegraphics[width=0.85\textwidth]{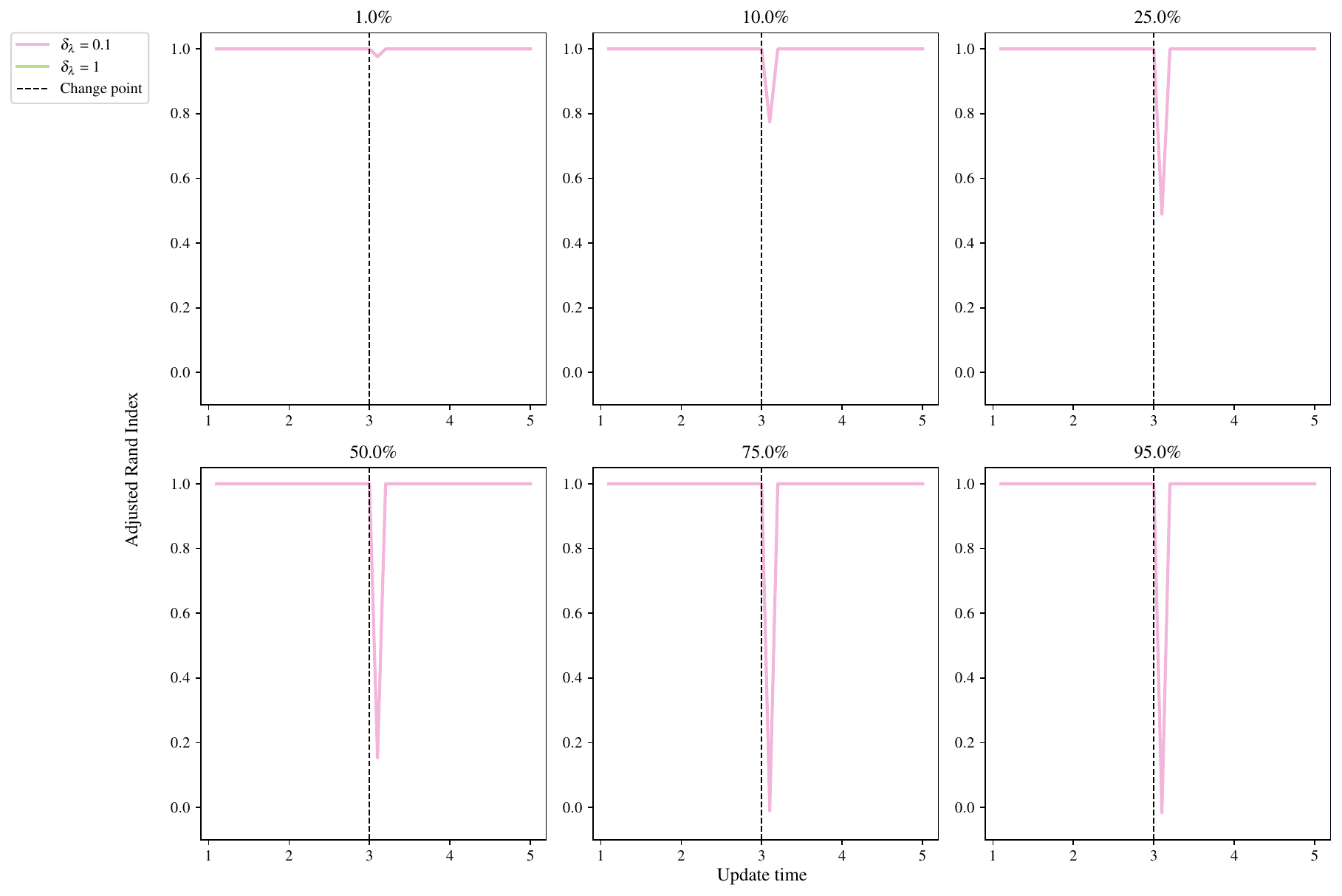}
         \caption{Mean ARI of the repetitions against update time for a varying proportion of group 1 nodes changing to group 2. % at the time indicated by the black, dashed vertical line. %Each panel is for a different percentage of nodes that swap. 
         %The green line is for no BFF and the pink for a BFF of 0.1.
         }
         \label{fig:group_percent_ARI}
     \end{subfigure}
     \vfill
     \begin{subfigure}[b]{\textwidth}
         \centering
         \includegraphics[width=0.85\textwidth]{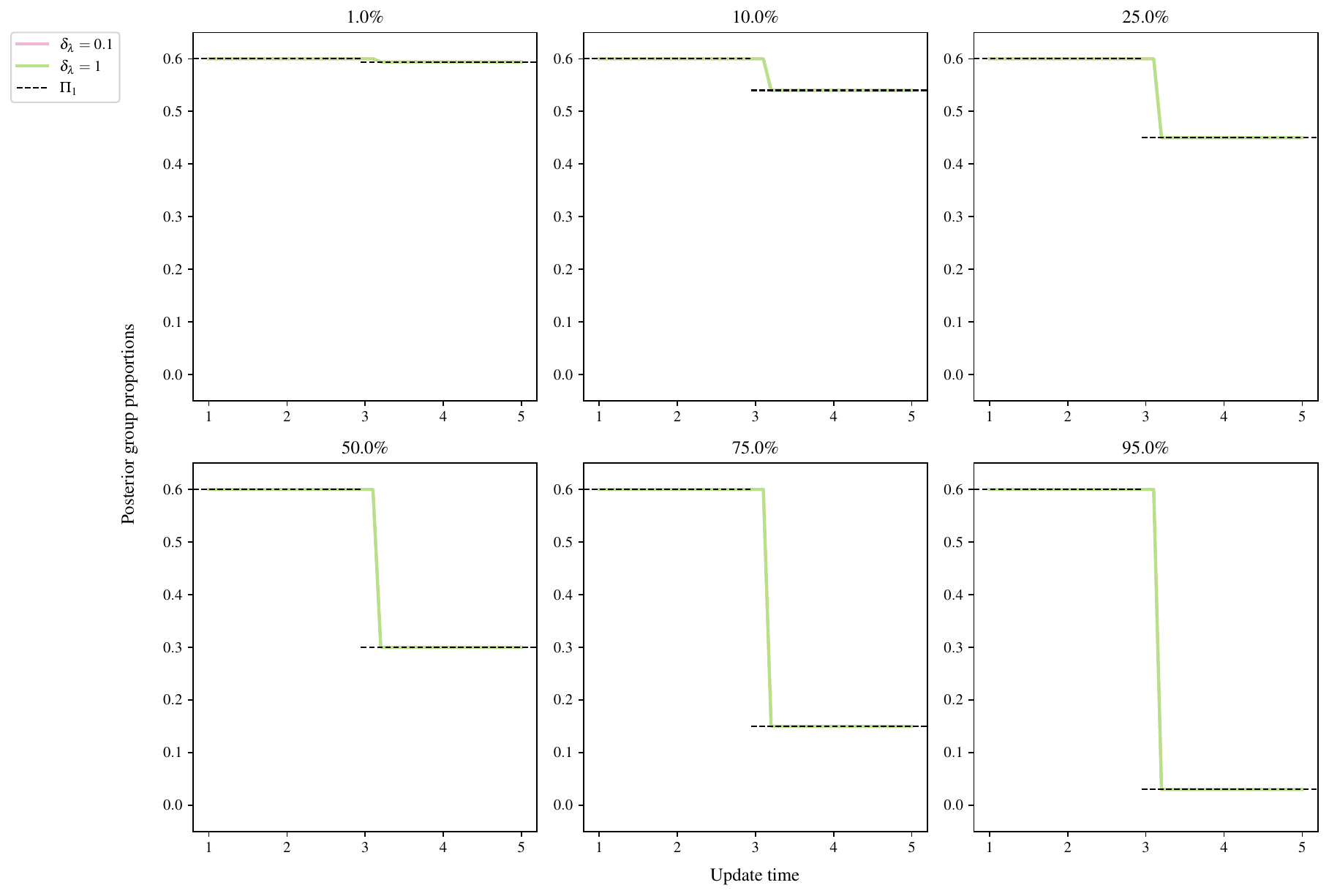}
         \caption{Mean over the repetitions of the proportion of nodes in group 1 with time. %The thick green line is for no BFF and the pink for a BFF of 0.1.
         The results obtained when using no BFF (green line) and a BFF of 0.1 (pink line) are shown together on each panel.
         The true proportion of group 1 either side of the changepoint is plotted as black, horizontal lines..}
         \label{fig:group_percent_prop}
     \end{subfigure}
    \caption{%The variation in d
    Detection of group membership changes for a varying percentage of nodes swapping from group 1 to group 2 at $t=3$. The panel titles give the percentage of nodes that swap from group 1 to 2.}
    \label{fig:group_percent_swap}
\end{figure*}

Figure \ref{fig:group_percent_ARI} shows the inference procedure is stable across all values of $P$, where the plots begin from $B_1$ \textcolor{red}{time steps}. The ARI is steady before and after the change except for the \textcolor{red}{time step} immediately after the change. Figure \ref{fig:group_percent_prop} shows that the algorithm quickly converges to the true proportions, both with and without a BFF. \textcolor{red}{Note that in both figures, the pink and green lines are indistinguishable from one another.} %An arbitrary difference between the labelling in the case of BFF and no BFF causes the opposite the lines. 
% An arbitrary change in the attributed group labels causes each line (BFF and no BFF) to switch between groups at the changepoint in the case of $P\in\{25,50,75,95\}$.
% While there is no appreciable effect on the ARI by the inclusion or exclusion of a BFF, it does play an important role in the estimation of the parameters. 
% In Figure \ref{fig:group_percent_lambda}, without a BFF, the posterior mean is significantly slower to converge to the true generating value than with a BFF. However, an increase in uncertainty is introduced by the BFF, as seen by the larger simulation intervals in Figure \ref{fig:group_percent_lambda}.

% \begin{figure}[t]
%      \centering
%      \includegraphics[width=0.9\textwidth]{figures/sim_study_1/posterior_means.pdf}
%      \caption{Posterior means and 95\% confidence interval for components $\lambda$ with update time, where $P$\% of nodes from group 1 swap to group 2, $P\in\{10,75,95\}$. %Intervals are obtained by averaging the 95\% bounds of each run. 
%      %The true values are black, dashed horizontal lines. The change occurs at the black, dashed vertical line. %Each column is a different percentage of nodes that swap from group 1 to 2. %, with the title giving the percentage. 
%      %Row 1 is for a BFF of 0.1, and row 2 for no BFF.
%      }
%      \label{fig:group_percent_lambda}
% \end{figure}

\subsection{Number of groups recovery}
\label{subsec:sim2}
We consider a fully connected network with two group membership changes: all nodes merging into group 1 at $t=2.5$, followed by the creation of a new group at $t=3.5$. \textcolor{red}{We use the GEM-BHPP model for this simulation.} Specifically, at $t=3.5$, $P\%$ of the nodes in group 1 remain, while the remainder create group 2. We again consider $P \in \{1, 10, 25, 50, 75, 95\}$, with $\lambda_0$ as the rate matrix, and ARI as our performance metric. We set $\textcolor{red}{W_{JS}} = 0$ here, which corresponds to taking the argmax of each $\tau_i$. 

Figure \ref{fig:group_percent_ARI_create} shows that the inference correctly groups the nodes in all cases pre-merger and between the merger and creation, both with and without a BFF. Using a BFF, the algorithm is seen to maintain a high ARI after the creation of a new group, in all cases except when $P=75$. The application without the BFF also fails in this case, but additionally it fails for $P=50$, and is slower to converge to the new groups when $P=25$ and 10. 

For cases where 50\% of nodes or less remain in group 1, the algorithm converges to the correct proportions, but with label-switching. This is observed in Figure \ref{fig:group_percent_lambda_create}, where the rates are seen to swap in the 10\% panel for the case of a BFF.

\begin{figure*}[p]
    \centering
    \includegraphics[width=0.9\textwidth]{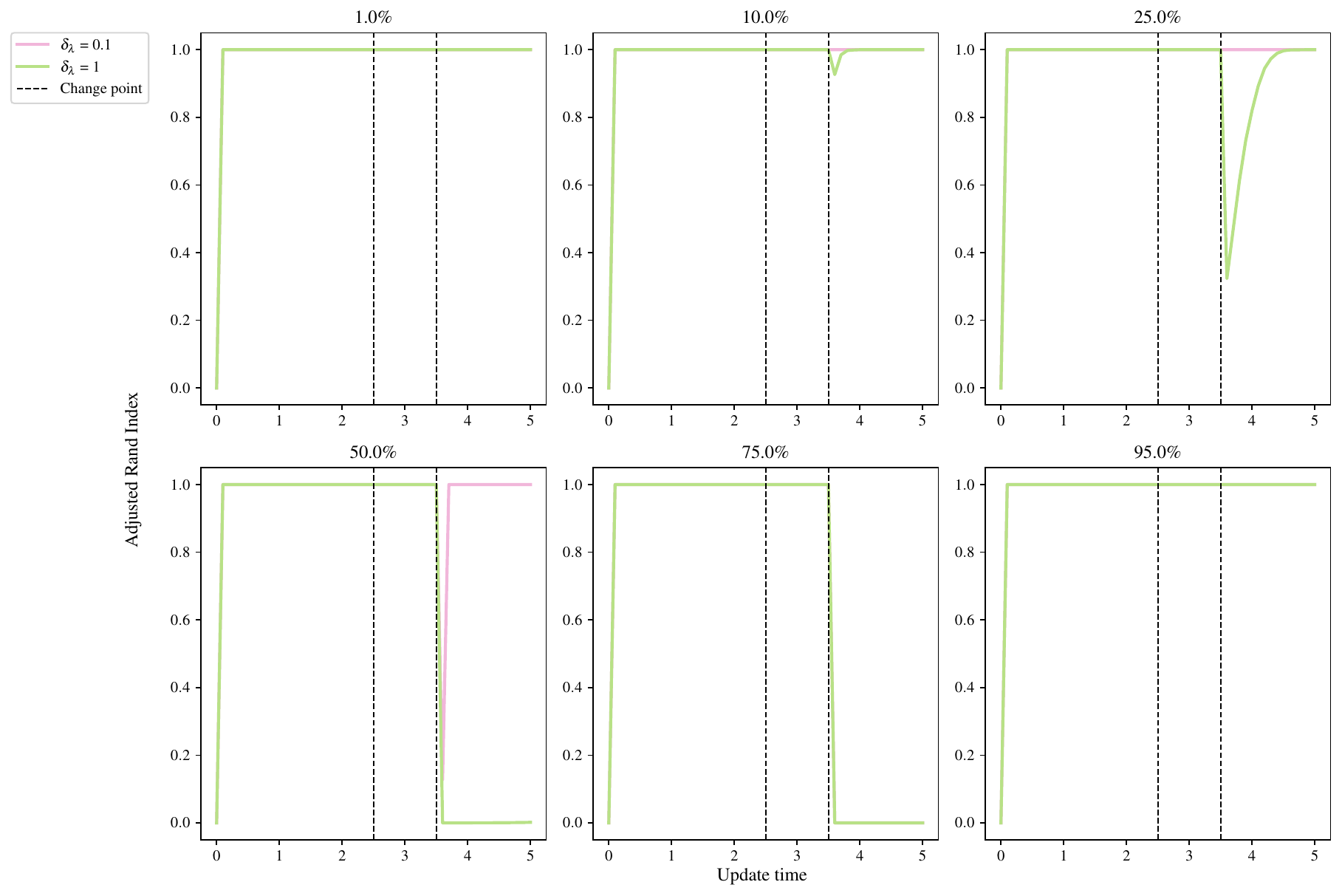}
    \caption{Mean ARI 
    against update time for the merger of group 2 into 1 at $t=2.5$, and the creation of group 2 at $t=3.5$, with the panel titles giving the percentage of nodes remaining in group 1 after $t=3.5$. The black, dashed vertical lines mark the changepoints. At $t=2.5$, all nodes in group 2 change to group 1, and at $t=3.5$, $P\%$ of group 1 nodes change to group 
     2, $P\in\{1,10,25,50,75,95\}$.}
     \label{fig:group_percent_ARI_create}

     \centering
     \includegraphics[width=0.9\textwidth]{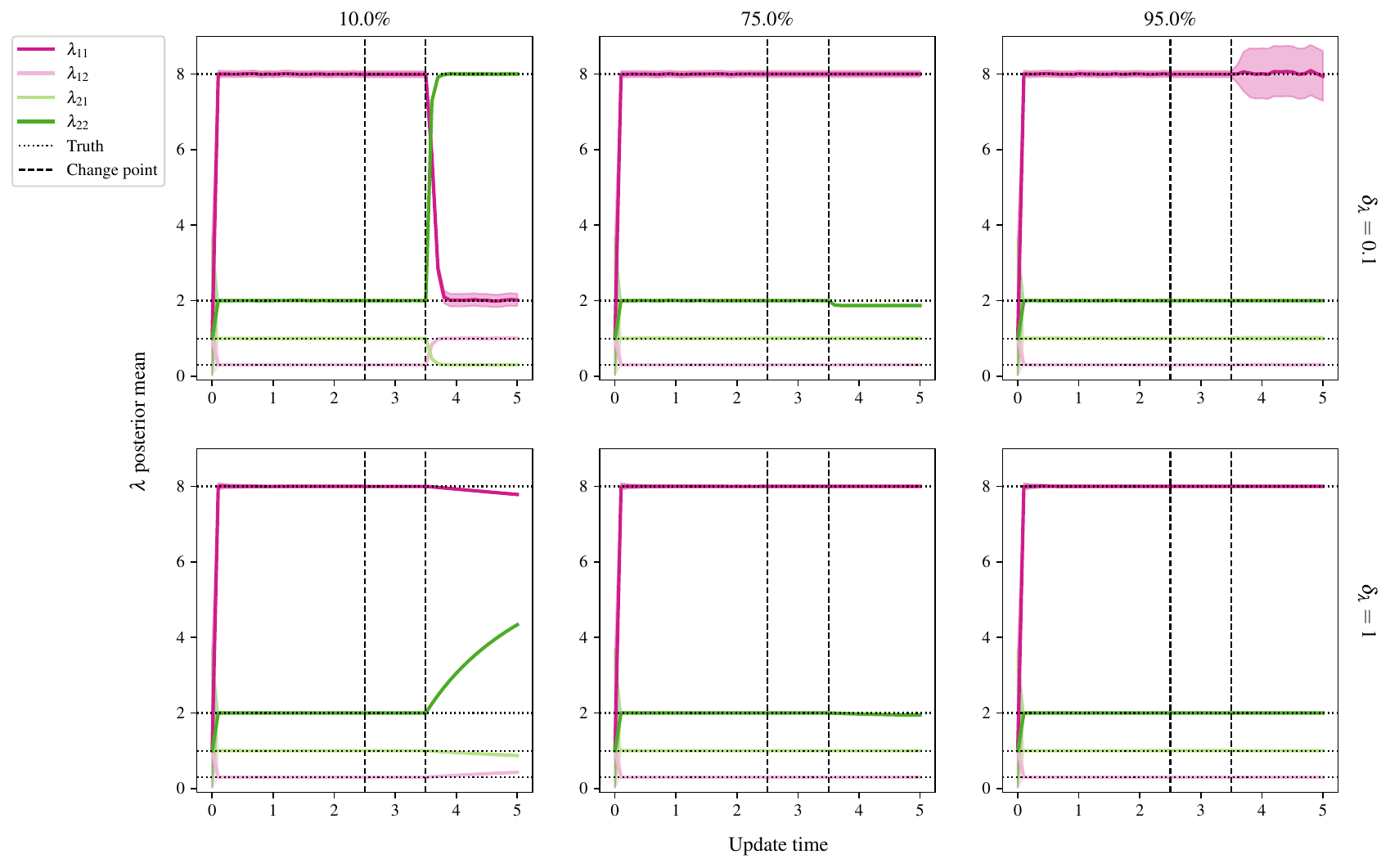}
     \caption{Posterior means and 95\% \textcolor{red}{simulation} interval for components $\lambda$ with update time. The black, dashed vertical lines mark the changepoints. At $t=2.5$, all nodes in group 2 change to group 1, and at $t=3.5$, $P\%$ of group 1 nodes change to group 2, $P\in\{10,75,95\}$.}
     \label{fig:group_percent_lambda_create}
\end{figure*}
 % \begin{subfigure}[b]{0.95\textwidth}
 %     \centering
 %     \includegraphics[width=0.9\textwidth]{figures/sim_study_3/group_props.pdf}
 %     \caption{Mean over 50 repetitions of the proportion of nodes in group 1. % with time. %The green line is for no BFF and the pink for a BFF of 0.1. 
 %     The true proportions of group 1 and 2 in the intervals around the changes points are plotted as black, horizontal lines (dotted for group 1 and dashed for group 2).}
 %     \label{fig:group_percent_prop_create}
 % \end{subfigure}

% \begin{figure}[t]
%      \centering
%      \includegraphics[width=0.9\textwidth]{figures/sim_study_3/posterior_means.pdf}
%      \caption{Posterior means and 95\% confidence interval for components $\lambda$ with update time. The black, dashed vertical lines mark the changepoints. At $t=2$, all nodes in group 2 change to group 1, and at $t=3$, $P\%$ of group 1 nodes change to group 2, $P\in\{10,75,95\}$. %Intervals are obtained by averaging the 95\% bounds of each run. 
%      %The true values are black, dashed horizontal lines. The change occurs at the black, dashed vertical line. %Each column is a different percentage of nodes that swap from group 1 to 2. %, with the title giving the percentage. 
%      %Row 1 is for a BFF of 0.1, and row 2 for no BFF.
%      }
% \label{fig:group_percent_lambda_create}
% \end{figure}

\subsection{Detection of changes to the matrix of rates}
\label{subsec:sim3}

On a fully connected network, we examine sequential changes to the rate matrix with decreasing time between the changes. The rate matrix maps as
\begin{align}
\lambda = \begin{pmatrix} 2 & 1 \\ 0.3 & 8 \end{pmatrix}
\quad \mapsto\quad 
\lambda' = \begin{pmatrix} 5 & 1 \\ 0.3 & 8 \end{pmatrix}
\quad \mapsto \quad 
\lambda'' = \begin{pmatrix} 3 & 1 \\ 0.3 & 8 \end{pmatrix},
\end{align}
with the first change at $t=3$ and the second at $t=3 + 0.1 M$. For each run, $M$ took a value in $\{1,2,3,4,5,10\}$. Here we do not reset the stream after a flagged change to the latent rates as we control their relative magnitude and wish to examine the performance of the algorithm with small latency between changes.

We evaluate the detection of latent changes using the proportion of changepoints correctly detected (CCD), and the proportion of detections that are not false (DNF), as in \cite{bodenham-2017}. Note that we aggregate across all group-to-group rates. Specifically, if there are $C$ changepoints, and we detect $D$, $T$ of which are correct, analogous to recall and precision, the authors define:
\begin{enumerate}
    \item $\mathrm{CCD} = T / C$, the proportion of changepoints correctly detected,
    \item $\mathrm{DNF} = T / D$, the proportion of detections that are not false.
\end{enumerate}
As noted by \cite{bodenham-2017}, CCD and DNF are preferred to their more intuitive counterparts (proportion of missed changepoints and proportion of false detections, respectively) as for the CCD and DNF, values closer to 1 indicate better performance than those closer to 0. Furthermore, these metrics are preferable over \textcolor{red}{the popular }average run length metrics \textcolor{red}{\cite[$\mathrm{ARL}0$ and $\mathrm{ARL}1$;][]{page1954}} as we want to capture the number of changes detected and missed.

Suppose $t^{\ast}_{n}$ and $t^{\ast}_{n+1}$ are consecutive true changepoints for $\lambda_{km}$, and let the most recently flagged changepoint by the algorithm for $\lambda_{km}$ be $t^{\prime}_{\ell} < t^{\ast}_{n}$. If the next flag for $\lambda_{km}$ occurs at $t^{\prime}_{\ell+1} > t^{\ast}_{n+1}$, that is $t^{\ast}_{n}$ is missed, then $t^{\prime}_{m+1}$ is classified as a correct detection of $t^{\ast}_{n+1}$. This is common practice in the literature. This approach is adopted by \cite{bodenham-2017}, whereas other authors implement a softmax rule for classification \cite[see, for example,][]{alanqary2021change, Yamanishi2002, bodenham-2017}, adjusting for when multiple changes are flagged in the same window. 

\begin{figure*}[p]
    \centering
    \includegraphics[width=0.8\textwidth]{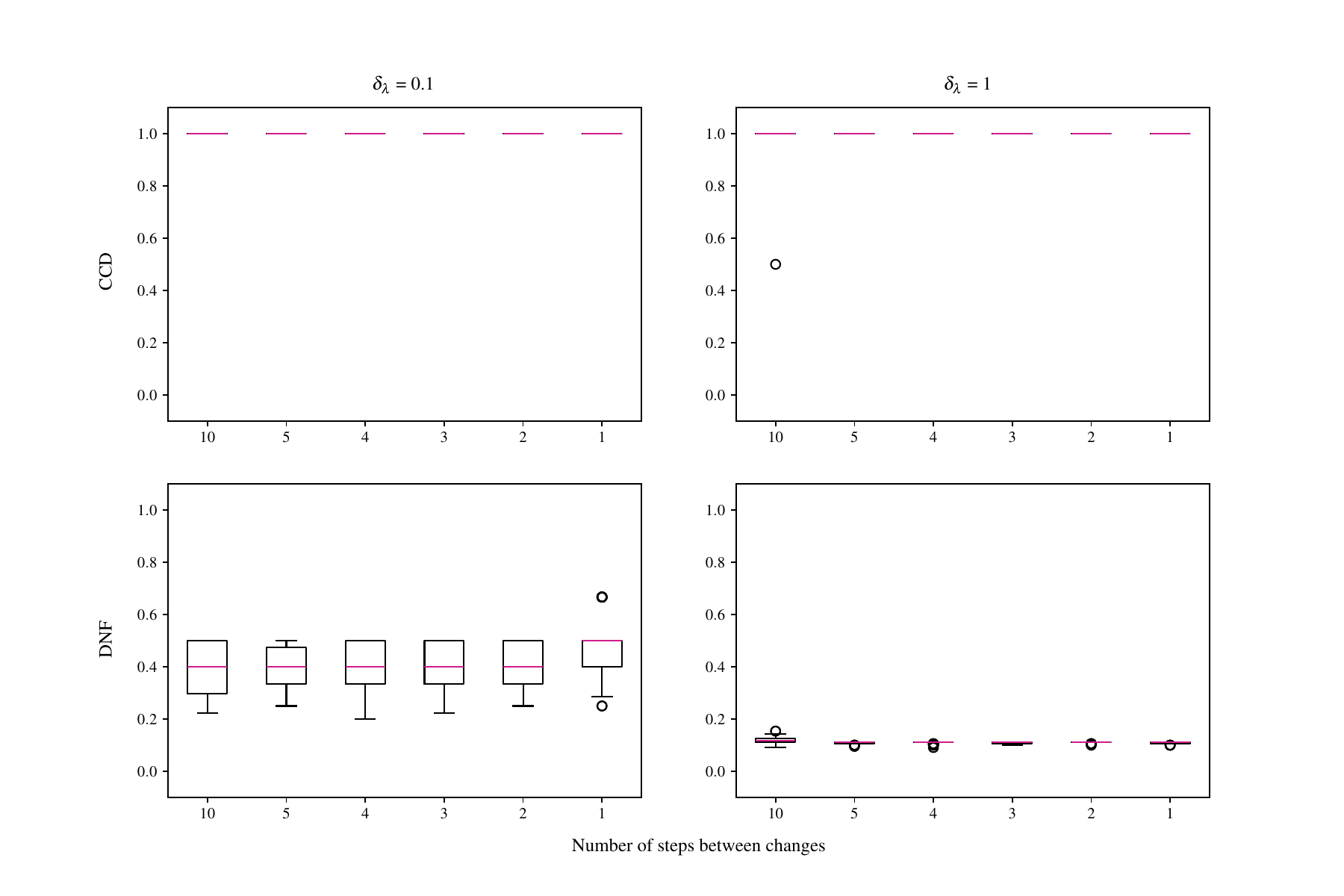}
    \caption{Boxplots of CCD and DNF over 50 runs for no BFF and a BFF of $\delta=0.1$. Each simulation has one change at $t=3$ and another a $t=3 + 0.1M$, where $M$ is on the horizontal axis.}
    \label{fig:CCD-DNF}

     \centering
     \includegraphics[width=0.9\textwidth]{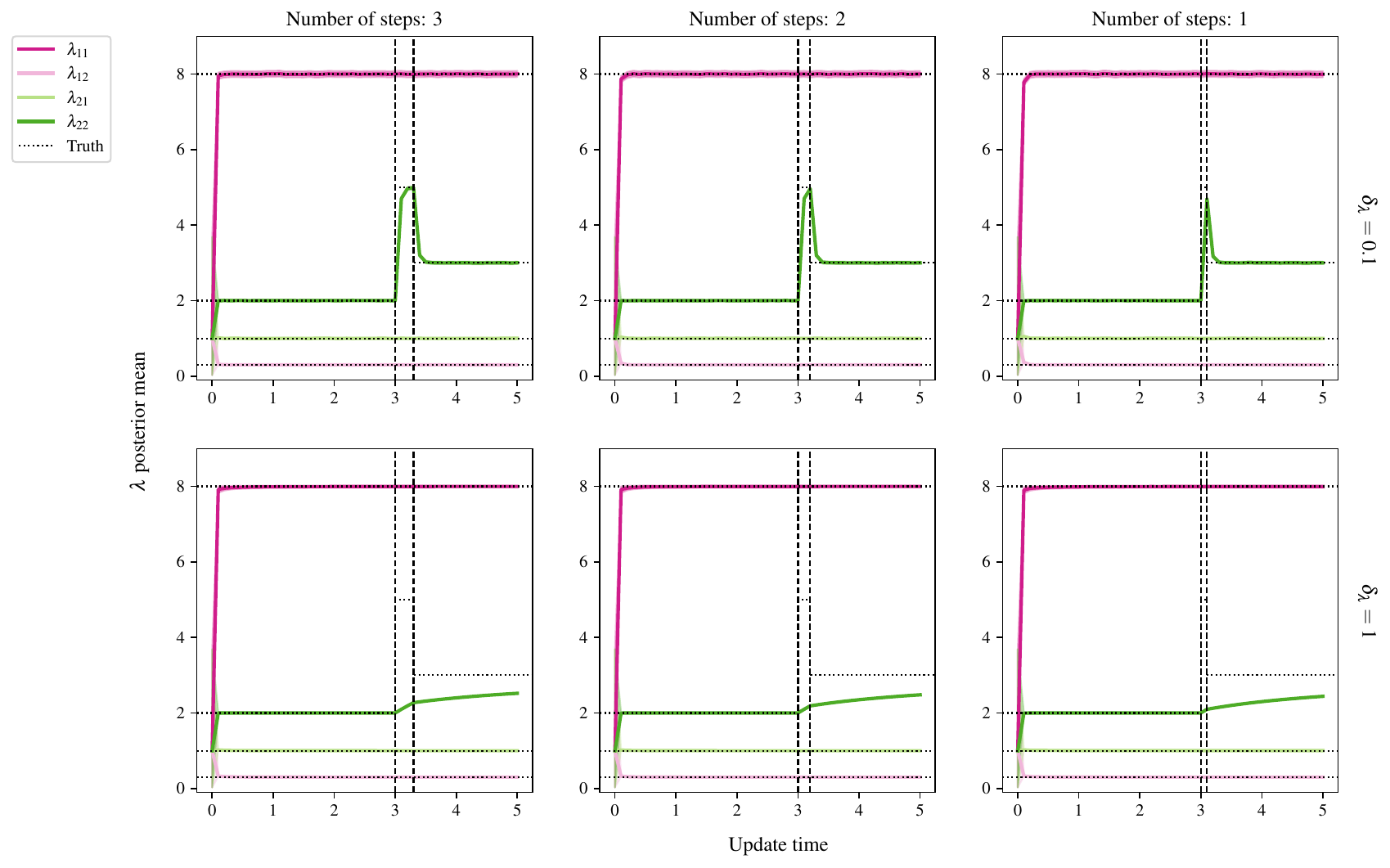}
     \caption{
     \textcolor{red}{Posterior means and 95\% \textcolor{red}{simulation} interval for components $\lambda$ with update time. The black, dashed vertical lines mark the changepoints. Each simulation has one change at $t=3$ and another a $t=3 + 0.1M$, where $M$ is the number of \textcolor{red}{time steps} in the column caption.}}
     \label{fig:rate_change_change_gap}
\end{figure*}

In Figure \ref{fig:CCD-DNF}, we see the CCD and DNF for each run. The CCD is consistently high for both a BFF and no BFF. This demonstrates that the algorithm consistently flags true changes with and without a BFF. The role of the BFF is shown in the DNF, where we see BFF yielding higher values. Without a BFF, multiple changes are flagged that did not occur.

\subsection{Effect of network sparsity}
\label{subsec:sim4}
We examine the effect of decreasing network sparsity on changepoint detection. We simulate a network with $K'=1$ latent connection groups, letting $\rho$, the group connection probability, vary over $\{0.01, 0.025, 0.05, 0.1, 0.25, 0.5\}$. For the point process groups, we retain the same $\Pi$ as before, and again, we set $\lambda = \lambda_0$. We simulate on $[0,25]$, with $\Delta = 0.1$, and at $t=10$, $25\%$ of nodes from group 1 swap to group 2. The inference procedure is run on the simulated network twice, once where the network is assumed fully-connected, and a second time \textcolor{red}{using the SBM-BHPP}. Both procedures are run with $\delta = 0.1$.

%The effect of the assumption of edges existing where they do not on the rate estimates is seen in 
Figure \ref{fig:post_means_infer_adj} demonstrates the effect of incorrectly assuming that the graph is fully connected: the rates are significantly underestimated. On the other hand, when $\boldsymbol{A}$ is inferred, the posterior means are much closer to the true values, even in the case of $1\%$ density. Furthermore, in Figure \ref{fig:ARI_infer_adj} the mean ARI of the membership recovery is seen to be higher for the inferred adjacency matrix. 

\begin{figure*}[p]
    \centering
    \includegraphics[width=0.9\textwidth]{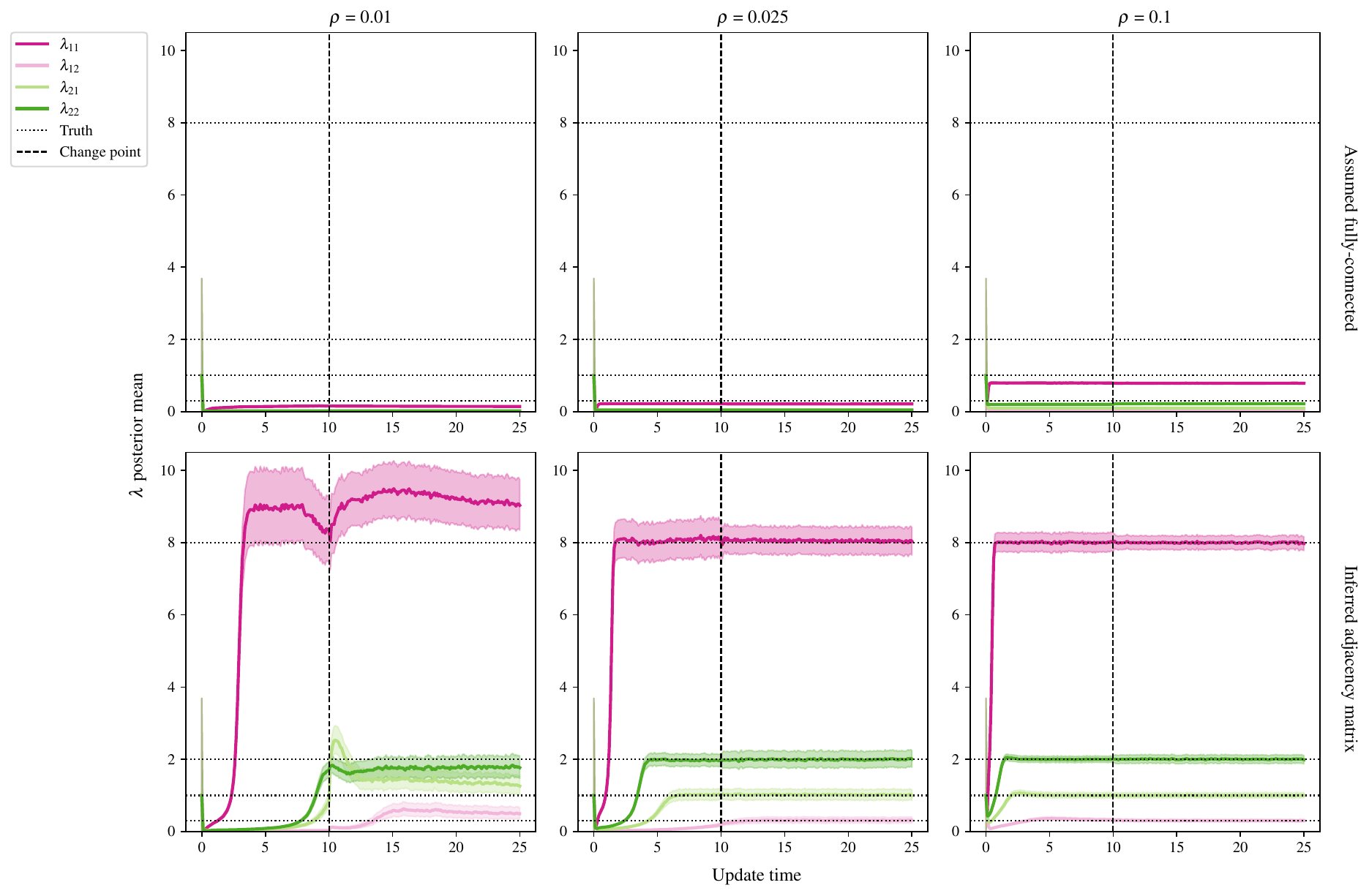}
    \caption{Posterior means and 95\% \textcolor{red}{simulation} interval for components $\lambda$ with update time. The black, dashed vertical line marks the changepoint, at which $25\%$ of nodes in group 1 swap to group 2. %Intervals are obtained by averaging the 95\% bounds of each run.
    %The true values are black, dashed horizontal lines. %Each column is for a different value of $\rho$. Row 1 is for when the network is assumed fully-connected, and row 2 for when the adjacency matrix is inferred.
    }
    \label{fig:post_means_infer_adj}
%\end{figure}

%\begin{figure}[p]
    \centering
    \includegraphics[width=.9\textwidth]{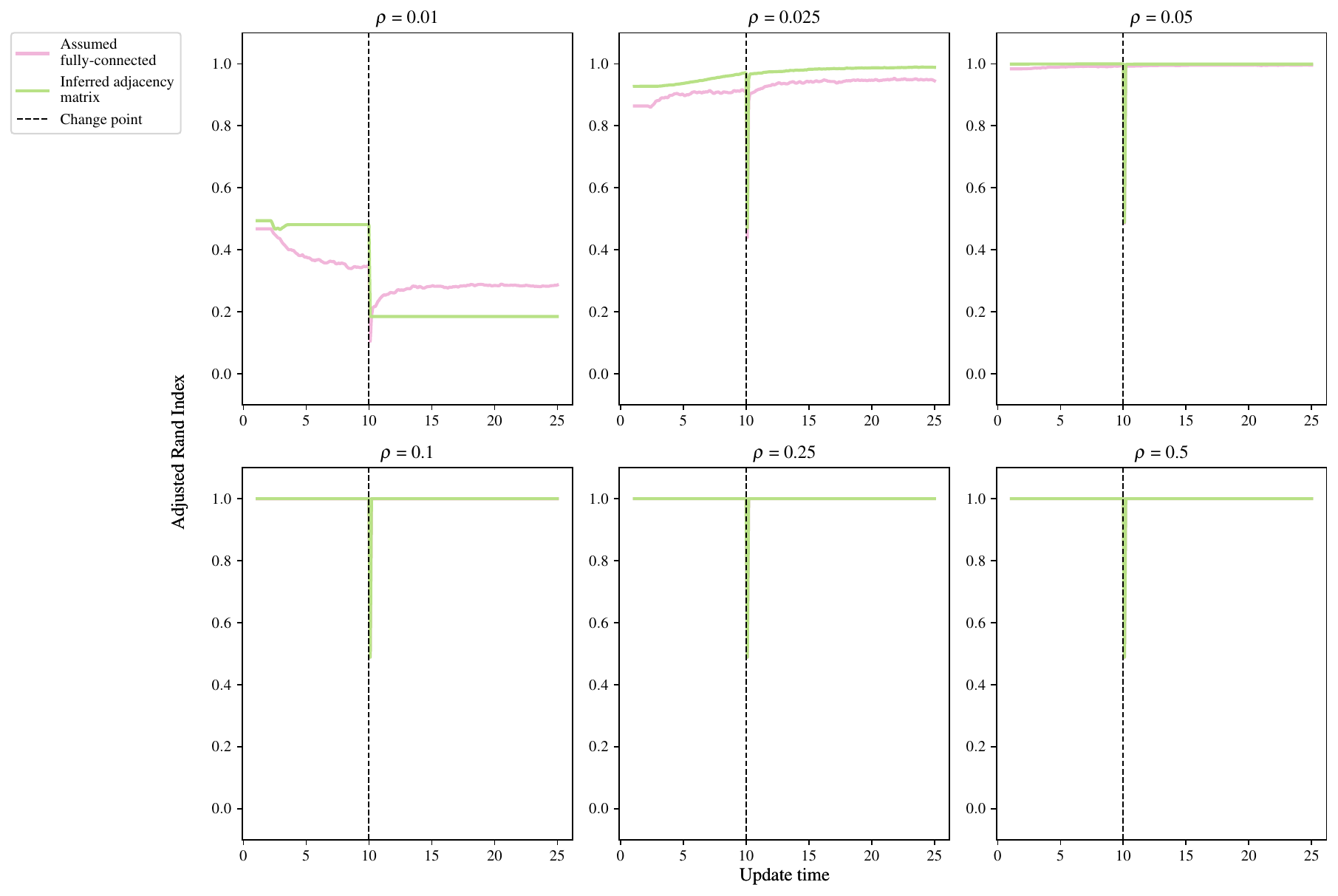}
    \caption{Mean ARI of the repetitions against update time. The black, dashed vertical line marks the changepoint, at which $25\%$ of nodes in group 1 swap to group 2.} %Each panel is for a different value of $\rho$.}
    \label{fig:ARI_infer_adj}
\end{figure*}

\subsection{\textcolor{red}{Detection of group changes with varying rates}}
\label{subsec:sim5}

\textcolor{red}{
On a fully connected network, we examine the recovery of group memberships with $K=3$ when the latent rates also vary. We simulate a network where the latent rates vary as follows:}

{\small\begin{align}
        \lambda(t) = \begin{pmatrix}
            2 s(t) + 5 & 0.1 s(t) + 0.2 & 0.05 s(t) + 0.1 \\
            0.2 c(t) + 1 & c(t) + 2 &  0.01s(t) + 0.8 \\
            0.1 c(t) + 0.9 & 0.1 s(t) + 0.5 & 0.01c(t) + 0.03
        \end{pmatrix}, \notag
\end{align}
}

\noindent{where $s(t)=\sin\left(2\pi t/5\right)$ and $c(t)=\cos\left(2\pi t/5\right)$.
In the simulation, we set $\Pi = (0.4, 0.4, 0.2)$ and at $t=3.05$, 25\% of the nodes change from group 1 to either group 2 or 3. We use a BFF of $0.1$ with the standard BHPP model. Figure \ref{fig:post_means_and_ARI} shows the posterior means of the rates and the ARI with update. We see that the rates are well-recovered and the ARI remains high either side of the change. 
\begin{figure*}[!h]
    \centering
    \includegraphics[width=0.9\textwidth]{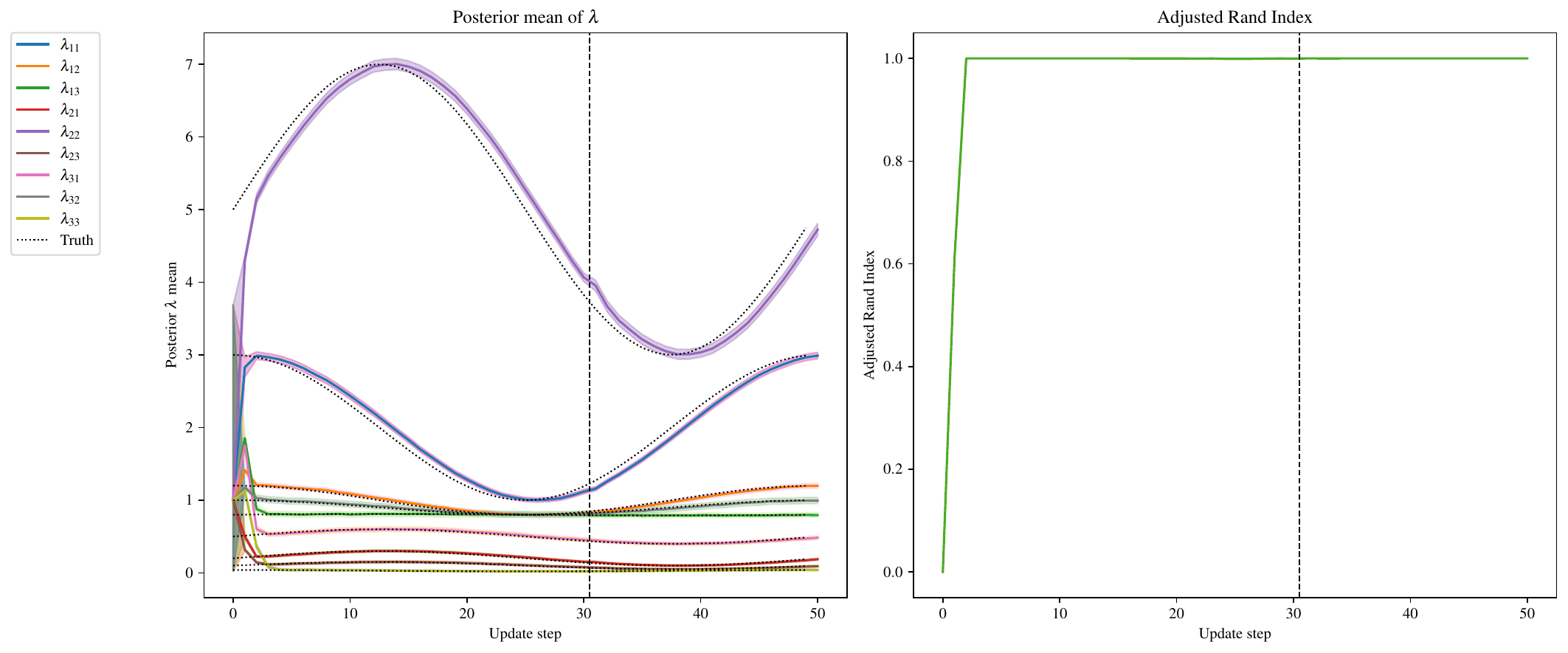}
    \caption{\textcolor{red}{Left-panel: posterior means and 95\% simulation interval for the components $\lambda$ with update time. The black, dotted lines are the true underlying rates. Right-panel: ARI with update time. The black, dashed vertical line marks the changepoint.
    }}
    \label{fig:post_means_and_ARI}
\end{figure*}
}

\section{Application to the Santander Cycles bike-sharing network}
\label{sec:santander}

The proposed online changepoint algorithms were tested on Transport for London (TfL) data from the London Santander Cycles bike-sharing network, which is publicly available online (\href{here}{https://cycling.data.tfl.gov.uk/}, powered by TfL Open Data). Each datum corresponds to a bike hire, and contains the start and end times of the journey, the IDs of the source and destination stations, the journey duration, a bike ID number, and an unique identifier for the journey. Considering the start and end stations as source and destination nodes, and the timestamp of the end of the journey as an arrival time to the directed edge from source to destination, the data forms a network point process. In this study, the data is aggregated into weekly counts to smooth the intensities of the point processes and weekly periodicities. \textcolor{red}{Despite the smoothing, we do not satisfy the homogeneous Poisson assumption, and thus look to infer changes to the latent groups and not to the latent rates. In Section \ref{subsec:sim5}, we demonstrate through simulation that we are able to do this.} We select a subset of the data from 2nd January 2019 until 15th July 2020 as this window contains significant COVID-19 related national events that can be used to check the performance of our algorithm. %Furthermore, stations (nodes) are added to the network over the period for which the data is available,  and we thus selected a window within which the number of stations is constant. 
In this time period, $N=791$ unique nodes are observed within $T=80$ weekly time windows, with updates every $\Delta = 1$ week time steps.

The online VB algorithm for the dynamic BHPP is run for the separate cases of \textcolor{red}{the SBM-BHPP and the GEM-BHPP}.
% an unknown graph structure, and an unknown number of groups. 
The number of groups in the case of an unknown adjacency matrix was set to match the number inferred by the implementation for an unknown number of groups, which was $K=6$. The adjacency matrix in the case of unknown $K$ was set to correspond to a fully connected graph: $\boldsymbol A = \boldsymbol{1}_{791 \times 791}$. 
We consider the task of detecting changes to the latent group structure of the bike sharing network. We run both algorithms with $\textcolor{red}{W_{JS}} = 1.55$, $B_1 = 25$, $B_2 = 10$, and $\kappa=2$, so that changes can be detected only after 
$B_1 + B_2 + \kappa = 37$ weeks of observations. We initialise the hyperparameters as described in Section \ref{sec:sims}.

\begin{figure*}[t]
     \centering     
     \includegraphics[width=0.8\textwidth]{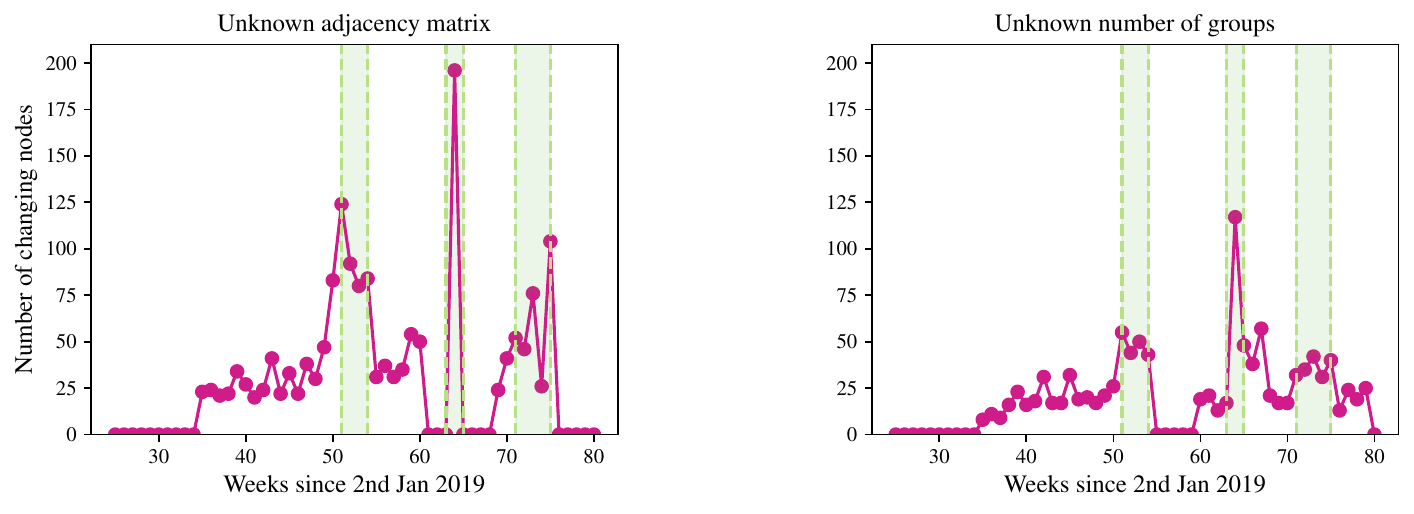}
    \caption{The number of nodes that change memberships at each update point from initialisation. The green regions correspond to the Christmas and New Year period of 2019, the introduction of the first UK COVID-19 lockdown, and the phased easing of these restrictions, respectively.}
    \label{fig:change_nodes_inf_graph}
\end{figure*}

Figure \ref{fig:change_nodes_inf_graph} shows that \textcolor{red}{both} algorithms flag multiple changes at each update, but that there are three regions where the number of flagged changes peaks. From left to right, these green regions correspond to the Christmas and New Year period of 2019, the introduction of the first UK COVID-19 lockdown on 19/03/2020, and the subsequent phased easing of restrictions from 01/06/2020. The algorithm reacts to these events, which are likely to cause changes to the network. 

\begin{figure*}[t]
    \centering
    \includegraphics[width=0.8\textwidth]{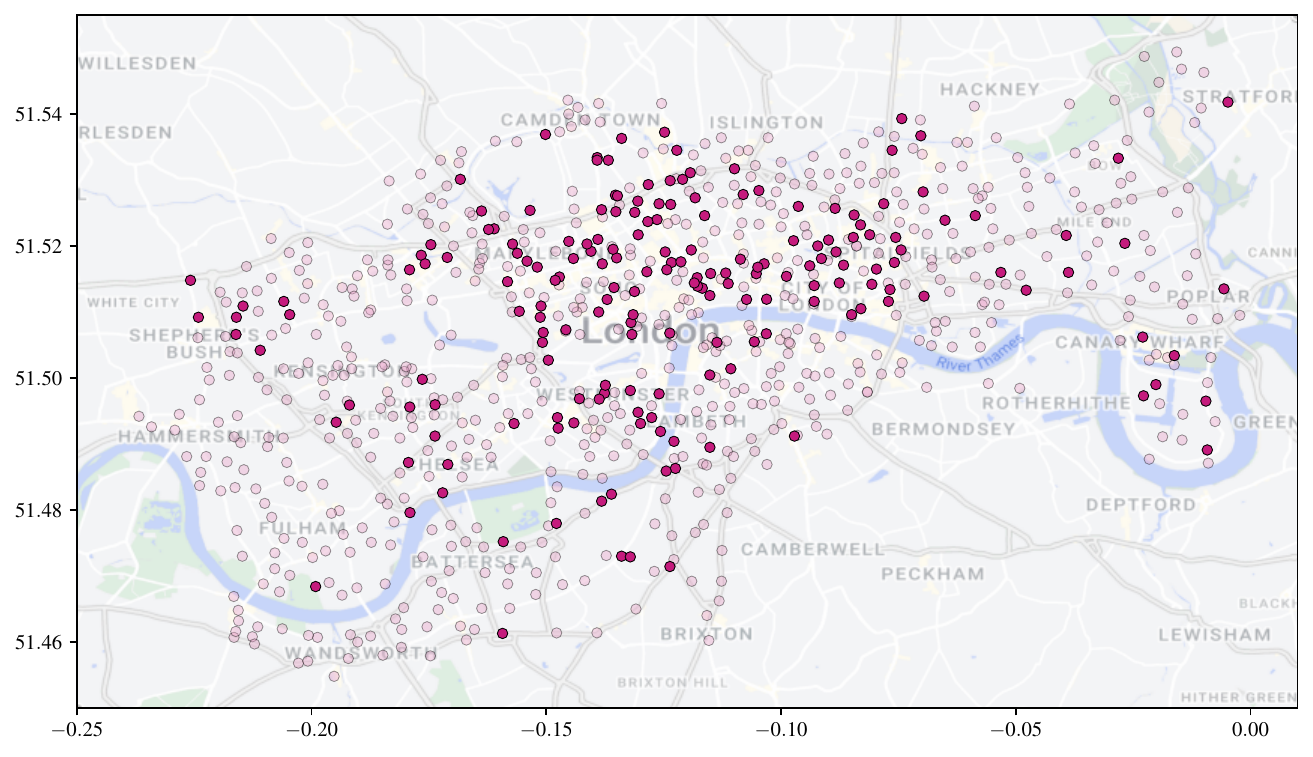}
    \caption{Station locations with those stations that change group at the onset of the UK COVID-19 lockdown coloured dark red.}
    \label{fig:post_lockdown}
\end{figure*}

Figure \ref{fig:post_lockdown} shows that the flagged changes \textcolor{red}{from the SBM-BHPP} corresponding to the onset of the lockdown are very concentrated around central London, which makes sense as ``work from home'' orders will have affected the use of these commuting bike stations. Additionally, stations around the Westfield Shopping Centre in Shepherd's Bush and the Canary Wharf financial district are flagged, representing an expected change due to the government restrictions.

Furthermore, Figure \ref{fig:spatial_dist} displays an example clustering from the \textcolor{red}{GEM-BHPP} at update point 50. There is some spatial clustering to the nodes, but each cluster contains mainly nodes with similar activity patterns. For example, the dark and light blue clusters represent popular nodes in central London and the Canary Wharf financial district, which are mainly used by commuters into these areas.
Similarly, the light and dark green clusters are nodes in West and East London (with the dark green cluster mostly covering West London, and the light green East London and the areas of Battersea and Wandsworth (south of the river Thames).
The pink cluster is around the boundary between the blue and green regions. A notable exception is represented by the red cluster, representing the most popular stations within the network, mainly used for leisure around Hyde Park and the Queen Elizabeth Olympic Park in Stratford. In addition, the red cluster also contains two additional stations in Battersea Park and Ravenscourt Park (west of Hammersmith). 
%Furthermore, the spatial clustering is not isolated to one region. We see different clustering of nodes of the same colour in different spatial regions, which indicates that the algorithm is able to detect groups driven not just by connected nodes, but also by similar edge behaviour.  

\begin{figure*}[t]
    \centering\includegraphics[width=0.8\textwidth]{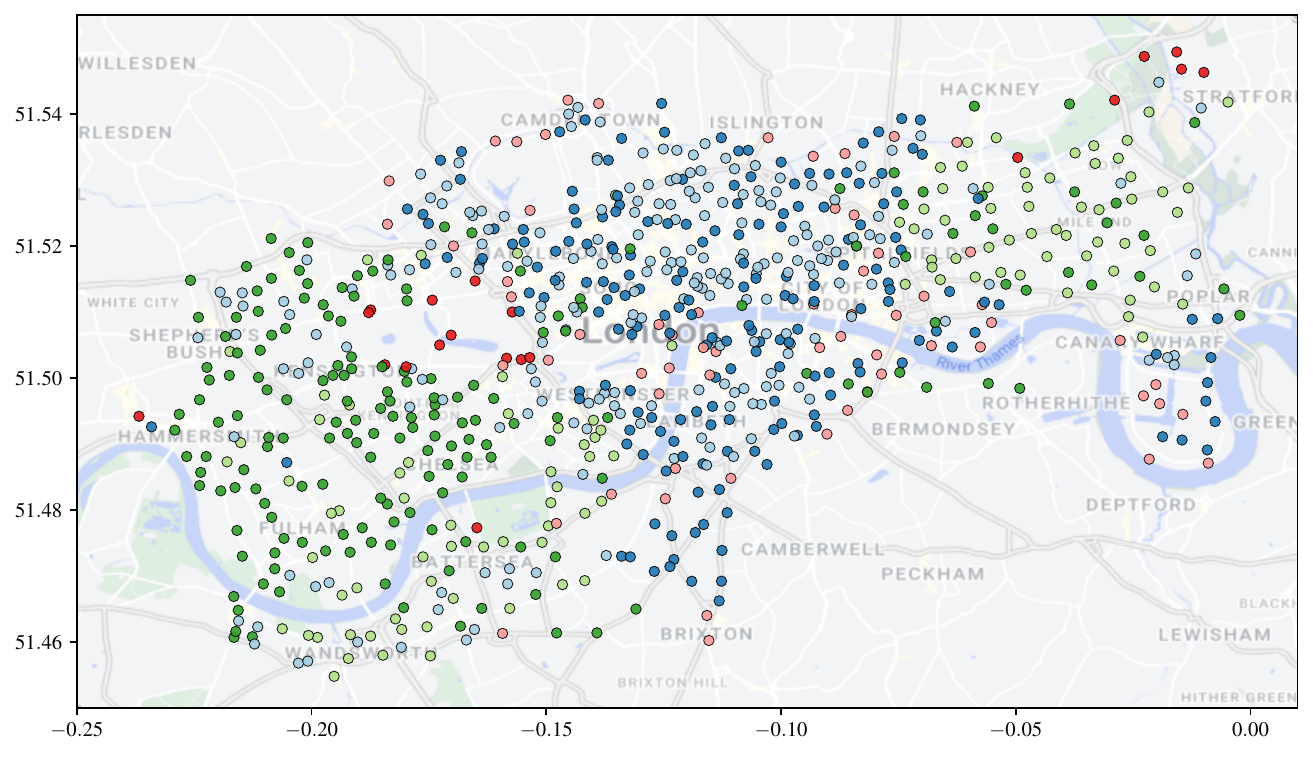}
    \caption{Station locations coloured by assigned group at update \textcolor{red}{time step} 50 of the dynamic BHPP with an unknown number of groups and $\boldsymbol{A} = \boldsymbol{1}_{791 \times 791}$.}
    \label{fig:spatial_dist}
\end{figure*}

\section{Conclusion and discussion}

We have presented a novel online Bayesian inference framework for detecting changes to the latent structure of a block-homogeneous Poisson process in which data arrives as batches in an event stream. Our methodology is scalable, and leverages a Bayesian forgetting factor framework to flag changes to the latent community structure and edge-process rates. The framework is extended to the cases where the adjacency matrix or the number of latent groups are unknown a priori. When tested on both real and simulated data, our methodology is seen to detect latent structure accurately and with minimal latency.

There are numerous ways in which this work could be extended. In particular, the frameworks for an unknown graph structure and unknown number of latent groups can readily be integrated to handle the case where neither is known a priori. It would also be of interest to incorporate seasonality into the model, perhaps within an online framework with longer memory. 

A further challenge would be to adapt the framework to allow for nodes to enter or leave the network during observation. Similarly, the adjacency matrix is assumed static in our methodology, and so an extension would allow for this to be dynamic, although this would likely cause identifiability issues in the case where there are also changes to the latent rates.

\section*{Acknowledgements}

\textcolor{red}{We would like to thank Professor Nick Heard and Professor Alessandra Luati for their valueable feedback on this manuscript.} Joshua Corneck acknowledges funding from the Engineering and Physical Sciences Research Council (EPSRC), grant number EP/S023151/1.
Ed Cohen acknowledges funding from the EPSRC NeST Programme, grant number EP/X002195/1.
Francesco Sanna Passino acknowledges funding from the EPSRC, grant number EP/Y002113/1. 

\section*{Code}

\textit{Python} code to implement the methodologies proposed in this article and reproduce the results is available in the Github repository {\color{red}\href{https://github.com/joshcorneck/dynamicBHPP.git}{\texttt{joshcorneck/dynamicBHPP}}}.

\bibliographystyle{rss}  
\bibliography{references} % Check references 

\newpage
\begin{appendices}
\section{Derivation of CAVI updates}\label{app:cavi}
Here we derive the CAVI approximating distributions for each parameter at the $r$th update. In deriving these expressions, we assume no ordering for the parameter updates, and will use a superscript of $(r-1)$ for all parameters but for the one whose expression is then being derived.

\subsection{Approximation of $q(z_i)$}

Ignoring any terms that do not contain $z_i$, we can derive the expectation as follows:
\begin{align}
        \mathbb{E}_{-z_i}&\{\log p(x,z,\lambda,\pi)\} = \mathbb{E}_{-z_i}\left\{\log p(x|z,\lambda) + \delta_z\log p(z|\pi)\right\} + \text{cst.}\\
        =&\ \mathbb{E}_{-z_i}\Bigg[\sum_{j:(i,j)\in\mathcal{E}}\sum_{k,m\in\mathcal{K}} \tilde z_{ik}\tilde z_{jm}\bigg(x_{ij}(I_r)\log(\lambda_{km}) - \Delta\lambda_{km}\bigg) + \\
        & \qquad \qquad \sum_{j':(j',i)\in\mathcal{E}}\sum_{k,m\in\mathcal{K}} \tilde z_{j'k}\tilde z_{im}\bigg(x_{j'i}(I_r)\log(\lambda_{km}) - \Delta \lambda_{km}\bigg) - \\
        &\qquad \qquad \sum_{k \in \mathcal{K}} \tilde z_{ik}\bigg(x_{ii}(t)\log(\lambda_{kk}) - \Delta\lambda_{kk}\bigg) + \delta_z \sum_{k\in\mathcal{K}} \tilde z_{ik}\log \pi_k \Bigg] + \text{cst.}\\
        =&\sum_{k\in\mathcal{K}}\tilde z_{ik}\Bigg[\delta_z \mathbb{E}_{\pi}\{\log \pi_k\} + x_{ii}(t)\mathbb{E}_{\lambda_{kk}}\{\log(\lambda_{kk})\} - \Delta\mathbb{E}_{\lambda_{kk}}\{\lambda_{kk}\} + \\
        & %\qquad \qquad 
        \sum_{m\in\mathcal{K}}\Bigg\{\sum_{j:(i,j)\in\mathcal{E}}\mathbb{E}_{z_j}\{\tilde z_{jm}\}\bigg(x_{ij}(I_r)\mathbb{E}_{\lambda_{km}}\{\log(\lambda_{km})\} - \Delta\mathbb{E}_{\lambda_{km}}\{\lambda_{km}\}\bigg)\bigg(1 - \mathbb{I}_{\{k\}}(m)\mathbb{I}_{\{i\}}(j)\bigg) + \\
        & %\qquad \qquad 
        \sum_{j':(j',i)\in\mathcal{E}}\mathbb{E}_{\tilde z_{j'}}\{\tilde z_{j'm}\}\bigg(x_{j'i}(I_r)\mathbb{E}_{\lambda_{mk}}\{\log(\lambda_{mk})\} - \Delta\mathbb{E}_{\lambda_{mk}}\{\lambda_{mk}\}\bigg)\bigg(1 - \mathbb{I}_{\{k\}}(m)\mathbb{I}_{\{i\}}(j')\bigg)\Bigg\}\Bigg] + \text{cst.}
\end{align}
Taking the exponential, it follows that the optimal choice of $\hat{q}(z_i)$ under a CAVI approximation is $\hat{q}^{(r)}(z_i) = \mathrm{Categorical}\left(z_i;\tau_i^{(r)}\right)$, where
\begin{align}
    \tau_{ik}^{(r)} &\propto \exp\Bigg\{\delta_z\Bigg[\psi\left(\gamma^{(r-1)}_k\right) - \psi\left(\sum_\ell \gamma^{(r-1)}_\ell\right)\Bigg] +  x_{ii}(t)\bigg(\psi\left(\alpha^{(r-1)}_{kk}\right) - \log\left(\beta^{(r-1)}_{kk}\right)\bigg) - \Delta\frac{\alpha^{(r-1)}_{kk}}{\beta^{(r-1)}_{kk}} + \\
    &\qquad \sum_{m\in\mathcal{K}}\Bigg[\sum_{j:(i,j)\in\mathcal{E}}\tau^{(r-1)}_{jm}\bigg(x_{ij}(I_r)\left(\psi\left(\alpha^{(r-1)}_{km}\right) - \log\left(\beta^{(r-1)}_{km}\right)\right) - \Delta\frac{\alpha^{(r-1)}_{km}}{\beta^{(r-1)}_{km}}\bigg)\bigg(1 - \mathbb{I}_{\{k\}}(m)\mathbb{I}_{\{i\}}(j)\bigg)\ + \\
    &\qquad \sum_{j':(j',i)\in\mathcal{E}}\tau_{j'm}^{(r-1)}\bigg(x_{j'i}(I_r)\left(\psi\left(\alpha^{(r-1)}_{mk}\right) - \log\left(\beta^{(r-1)}_{mk}\right)\right) - \Delta\frac{\alpha^{(r-1)}_{mk}}{\beta^{(r-1)}_{mk}}\bigg)\bigg(1 - \mathbb{I}_{\{k\}}(m)\mathbb{I}_{\{i\}}(j')\bigg)\Bigg]\Bigg\},
\end{align}
for $i\in\mathcal{V}$ and $k\in\mathcal{K}$.

\subsection{Approximation of $q(\pi)$}

Ignoring any terms that do not contain $\pi$, we compute the expectation as:
\begin{align}
    \mathbb{E}_{-\pi}&\left\{\log p(x,z,\lambda,\pi)\right\} = \mathbb{E}_{-\pi}\left\{\delta_z\log p(z|\pi) + \delta_\pi\log p(\pi)\right\} + \text{cst.} \\
    &= \sum_{k\in\mathcal{K}} \log \pi_k\left[\sum_{i\in\mathcal{V}}\delta_z\mathbb{E}_{z_i}\{\tilde z_{ik}\} + \delta_\pi\left(\gamma_k^{(r-1)} - 1\right)\right] + \text{cst.} 
\end{align}
Using the distributions derived for $\hat{q}^{(r-1)}(z_i)$, we see that the optional CAVI approximation is to take $\hat{q}^{(r)}(\pi) = \mathrm{Dirichlet}\left(\pi;\gamma^{(r)}\right)$,
where
\begin{equation}
    \gamma_k^{(r)} = \delta_\pi\left(\gamma_k^{(r-1)} - 1\right) + \delta_z\sum_{i\in\mathcal{V}} \tau^{(r-1)}_{ik} + 1,
\end{equation}
for $k\in\mathcal{K}$.
\subsection{Approximation of $q(\lambda_{km})$}

Ignoring any terms that do not contain $\lambda_{km}$, we derive the expectation as follows:
\begin{align}
    \mathbb{E}_{-\lambda_{km}}&\left\{\log p(x,z,\lambda,\pi)\right\} = \mathbb{E}_{-\lambda_{km}}\left\{\log p(x|z,\lambda) + \delta_\lambda\log p(\lambda)\right\} + \text{cst.} \\
    &=\ \sum_{(i,j)\in\mathcal{E}} \sum_{k,m\in\mathcal{K}}\bigg[\mathbb{E}_{z_i}\{\tilde z_{ik}\}\mathbb{E}_{z_j}\{\tilde z_{jm}\}(x_{ij}\log \lambda_{km} -
    \lambda_{km}\Delta)\bigg] + \\ 
    &\quad \delta_\lambda\sum_{k,m\in\mathcal{K}}\left[\left(\alpha^{(r-1)}_{km} - 1\right)\log \lambda_{km} - \beta^{(r-1)}_{km}\lambda_{km}\right] + \text{cst.} \\
    &= \sum_{k,m\in\mathcal{K}}\bigg\{\bigg(\delta_\lambda\left(\alpha^{(r-1)}_{km} - 1\right) + \sum_{(i,j)\in\mathcal{E}}\mathbb{E}_{z_i}\{\tilde z_{ik}\}\mathbb{E}_{z_j}\{\tilde z_{ik}\}x_{ij}(I_r)\bigg)\log \lambda_{km}\ - \\
    &\quad \left(\delta_\lambda\beta^{(r-1)}_{km} + \Delta\sum_{(i,j)\in\mathcal{E}}\mathbb{E}_{z_i}\{\tilde z_{ik}\}\mathbb{E}_{z_j}\{\tilde z_{ik}\}\right)\lambda_{km}\bigg\} + \text{cst.} 
\end{align} 
Using the distributions derived for $\hat{q}^{(r-1)}(z_i)$, it follows that the optimal CAVI approximation is to take $\hat{q}^{(r)}\left(\lambda_{km}\right) = \mathrm{Gamma}\left(\lambda_{km}; \alpha^{(r)}_{km}, \beta^{(r)}_{km}\right)$, where
\begin{align}
    \alpha^{(r)}_{km} &= \delta_\lambda\left(\alpha^{(r-1)}_{km} - 1\right) + \sum_{(i,j)\in\mathcal{E}}\tau^{(r-1)}_{ik}\tau^{(r-1)}_{jm}x_{ij}(I_r) + 1,\\
    \beta^{(r)}_{km} &= \delta_\lambda\beta^{(r-1)}_{km} + \Delta\sum_{(i,j)\in\mathcal{E}}\tau^{(r-1)}_{ik}\tau^{(r-1)}_{jm},
\end{align}
for $k,m\in\mathcal{K}$.

\section{Derivation of CAVI updates for an unknown adjacency matrix}
\label{app:cavi_unknown_graph}
Here we derive the CAVI approximating distributions for each parameter at the $r$th update. In deriving these expressions, we assume no ordering for the parameter updates, and will use a superscript of $(r-1)$ for all parameters but for the one whose expression is then being derived.

\subsection{Approximation of $q(z_i)$}

Ignoring any terms that do not contain $z_i$, we derive the expectation as:
\begin{align}
    \mathbb{E}_{-z_i}&\left\{\log p(x,z,a,\lambda,\pi)\right\} = \mathbb{E}_{-z_i}\left\{\log p(x|z,a,\lambda) + \delta_z\log p(z|\pi)\right\} + \text{cst.} \\
    &= \mathbb{E}_{-z_i}\Bigg\{\sum_{j\in\mathcal{V}}\sum_{k,m\in\mathcal{K}}\tilde z_{ik} \tilde z_{jm}\bigg[a_{ij}\left(x_{ij}(I_r)\log\lambda_{km} - \Delta\lambda_{km}\right)\bigg] + \\
    & \qquad \qquad \sum_{j'\in\mathcal{V}}\sum_{k,m\in\mathcal{K}}\tilde z_{j'k} \tilde z_{im}\bigg[a_{j'i}\left(x_{j'i}(I_r)\log\lambda_{km} - \Delta\lambda_{km}\right)\bigg] - \\
    &\qquad \qquad \sum_{k\in\mathcal{K}}\tilde z_{ik}\bigg[a_{ii}\left(x_{ii}(I_r)\log\lambda_{kk} - \Delta\lambda_{kk}\right)\bigg] + \delta_z\sum_{k\in\mathcal{K}}\tilde z_{ik}\log \pi_k\Bigg\} + \text{cst.}\\
    &= \mathbb{E}_{-z_i}\Bigg\{\sum_{j\in\mathcal{V}}\sum_{k,m\in\mathcal{K}}\tilde z_{ik}\tilde z_{jm}\bigg[a_{ij}\left(x_{ij}(I_r)\log\lambda_{km} - \Delta\lambda_{km}\right) + a_{ji}\left(x_{ji}(I_r)\log\lambda_{mk} - \Delta\lambda_{mk}\right)\bigg] - \\
    & \qquad \qquad \sum_{k\in\mathcal{K}}\tilde z_{ik}\bigg[a_{ii}\left(x_{ii}(I_r)\log\lambda_{kk} - \Delta\lambda_{kk}\right)\bigg] + \delta_z\sum_{k\in\mathcal{K}}\tilde z_{ik}\log \pi_k\Bigg\} + \text{cst.}\\
    &= \sum_{k\in\mathcal{K}}\tilde z_{ik}\Bigg\{\sum_{j\in\mathcal{V}}\sum_{m\in\mathcal{K}}\tilde z_{jm}\bigg[x_{ij}(I_r)\mathbb{E}_{a_{ij}}\{a_{ij}\}\mathbb{E}_{\lambda_{km}}\{\log\lambda_{km}\} + x_{ji}(I_r)\mathbb{E}_{a_{ji}}\{a_{ji}\}\mathbb{E}_{\lambda_{mk}}\{\log\lambda_{mk}\} \ - \\
    &\qquad \qquad \Delta\bigg(\mathbb{E}_{a_{ij}}\{a_{ij}\}\mathbb{E}_{\lambda_{km}}\{\lambda_{km}\} + \mathbb{E}_{a_{ji}}\{a_{ji}\}\mathbb{E}_{\lambda_{mk}}\{\lambda_{mk}\}\bigg)\bigg]\bigg(1 - \mathbb{I}_{\{k\}}(m)\mathbb{I}_{\{i\}}(j)\bigg) + \\
    &\qquad \qquad \mathbb{E}_{a_{ii}}\{a_{ii}\}\bigg(x_{ii}(I_r)\mathbb{E}_{\lambda_{kk}}\{\log\lambda_{kk}\} - \Delta\mathbb{E}_{\lambda_{kk}}\{\lambda_{kk}\}\bigg) + \delta_z\mathbb{E}_{\pi_k}\left\{\log \pi_k\right\}\Bigg\} + \text{cst.} 
\end{align}
Taking the exponential, and using the distributions derived for $\hat{q}^{(r-1)}(\lambda_{km})$ and the assumed form for $\hat{q}^{(r-1)}(a_{ij})$, it follows that the optimal CAVI approximation is to take $\hat{q}^{(r)}(z_i) = \mathrm{Categorical}\left(z_i;\tau^{(r)}_i\right)$, where
\begin{align}
    \tau^{(r)}_{ik} &\propto \exp\Bigg\{\sum_{j\in\mathcal{V}}\sum_{m\in\mathcal{K}}\tau^{(r-1)}_{jm}\bigg[\sigma^{(r-1)}_{ij}x_{ij}(I_r)\bigg(\psi\left(\alpha^{(r-1)}_{km}\right) - \log\left(\beta^{(r-1)}_{km}\right)\bigg)\ + \\
    &\quad \sigma^{(r-1)}_{ji}x_{ji}(I_r)\bigg(\psi\left(\alpha^{(r-1)}_{mk}\right) - \log\left(\beta^{(r-1)}_{mk}\right)\bigg) - \Delta\bigg(\sigma^{(r-1)}_{ij}\frac{\alpha^{(r-1)}_{km}}{\beta^{(r-1)}_{km}} + \sigma^{(r-1)}_{ji}\frac{\alpha^{(r-1)}_{mk}}{\beta^{(r-1)}_{mk}}\bigg)\bigg]\ \times \\
    & \quad\bigg(1 - \mathbb{I}_{\{k\}}(m)\mathbb{I}_{\{i\}}(j)\bigg)+ \sigma^{(r-1)}_{ii}\bigg(x_{ii}(I_r)\bigg(\psi\left(\alpha^{(r-1)}_{kk}\right) - \log\left(\beta^{(r-1)}_{kk}\right)\bigg) - \Delta\frac{\alpha^{(r-1)}_{kk}}{\beta^{(r-1)}_{kk}}\bigg)\ + \\
    &\quad \delta_z\Bigg[\psi\left(\gamma^{(r-1)}_k\right) - \psi\left(\sum_\ell \gamma^{(r-1)}_\ell\right)\Bigg]\Bigg\},
\end{align}
for $i\in\mathcal{V}$ and $k\in\mathcal{K}$.

\subsection{Approximation of $q(\lambda_{km})$}
Ignoring any terms that do not contain $\lambda_{km}$, we can derive the expectation as follows:
\begin{align}
    \mathbb{E}_{-\lambda_{km}}&\{\log p(x,z,z',a,\lambda,\rho,\pi,\mu)\} = \mathbb{E}_{-\lambda_{km}}\bigg\{\log p(x|z,a,\lambda) + \delta_\lambda\log p(\lambda)\bigg\} + \text{cst.} \\
    &= \sum_{(i,j)\in\mathcal{R}} \mathbb{E}_{z_i}\{\tilde z_{ik}\}\mathbb{E}_{z_j}\{\tilde z_{jm}\}\mathbb{E}_{a_{ij}}\{a_{ij}\}\bigg(x_{ij}(I_r)\log\lambda_{km} - \Delta\lambda_{km}\bigg) + \\ 
    &\quad \delta_{\lambda}\bigg(\bigg(\alpha_{km}^{(r-1)} - 1\bigg)\log\lambda_{km} - \beta_{km}^{(r-1)}\lambda_{km}\bigg) + \text{cst.} \\
    &= \bigg(\sum_{(i,j)\in\mathcal{R}} \mathbb{E}_{z_i}\{\tilde z_{ik}\}\mathbb{E}_{z_i}\{\tilde z_{jm}\}\mathbb{E}+{a_{ij}}\{a_{ij}\}x_{ij}(I_r) + \delta_\lambda\bigg(\alpha_{km}^{(r-1)} - 1\bigg)\bigg) \log\lambda_{km}\  - \\
    &\quad \bigg(\Delta\sum_{(i,j)\in\mathcal{R}} \mathbb{E}_{z_i}\{\tilde z_{ik}\}\mathbb{E}_{z_j}\{\tilde z_{jm}\}\mathbb{E}_{a_{ij}}\{a_{ij}\} + \delta_\lambda\beta_{km}^{(r-1)}\bigg)\lambda_{km} + \text{cst.}
\end{align}
Using the distributions derived for $\hat{q}^{(r-1)}(z_i)$ and assumed form for $\hat{q}^{(r-1)}(a_{ij})$, we see that the optimal CAVI approximating distribution is $\hat{q}^{(r)}(\lambda_{km}) = \mathrm{Gamma}\bigg(\lambda_{km};\alpha_{km}^{(r)},\beta_{km}^{(r)}\bigg)$ with
\begin{align}
    \alpha_{km}^{(r)} &= \delta_\lambda\bigg(\alpha_{km}^{(r-1)} - 1\bigg) +  \sum_{(i,j)\in\mathcal{R}}\tau^{(r-1)}_{ik}\tau^{(r-1)}_{jm}\sigma_{ij}^{(r-1)}x_{ij}(I_r) + 1, \\
    \beta_{km}^{(r)} &= \delta_\lambda\beta_{km}^{(r-1)} + \Delta\sum_{(i,j)\in\mathcal{R}}\tau^{(r-1)}_{ij}\tau^{(r-1)}_{jm}\sigma^{(r-1)}_{ij},
\end{align}
for $k,m\in\mathcal{K}$.

\subsection{Approximation of $q(\rho_{k'm'})$}
Ignoring any terms that do not contain $\rho_{k'm'}$, we derive the expectation as:
\begin{align}
    \mathbb{E}_{-\rho_{k'm'}}&\{\log p(x,z,z',a,\lambda,\rho,\pi,\mu)\} = \mathbb{E}_{-\rho_{k'm'}}\{\log p(a|z,\rho) + \delta_\rho \log p(\rho)\} + \text{cst.} \\
    &= \sum_{(i,j)\in\mathcal{R}} \mathbb{E}_{z_i'}\{\tilde z_{ik'}'\}\mathbb{E}_{z_j'}\{\tilde z_{jm'}'\}\bigg[\mathbb{E}_{a_{ij}}\{a_{ij}\}\log\rho_{k'm'} + \mathbb{E}_{a_{ij}}\{1 - a_{ij}\}\log(1 - \rho_{k'm'})\bigg]\ + \\ &\quad \delta_\rho\bigg(\eta_{k'm'}^{(r-1)} - 1\bigg)\log \rho_{k'm'} + \delta_\rho\bigg(\zeta_{k'm'}^{(r-1)} - 1\bigg)\log (1 - \rho_{k'm'}) + \text{cst.} \\
    &= \bigg[\sum_{(i,j)\in\mathcal{R}} \mathbb{E}_{z_i'}\{\tilde z_{ik'}'\}\mathbb{E}_{z_j'}\{\tilde z_{jm'}'\}\mathbb{E}_{a_{ij}}\{a_{ij}\} + \delta_\rho\bigg( \eta_{k'm'}^{(r-1)} - 1\bigg)\bigg]\log\rho_{k'm'}\ + \\
    &\quad \bigg[\sum_{(i,j)\in\mathcal{R}} \mathbb{E}_{z_i'}\{\tilde z_{ik'}'\}\mathbb{E}_{z_j'}\{\tilde z_{jm'}'\}\mathbb{E}_{a_{ij}}\{1 - a_{ij}\} + \delta_\rho\bigg(\zeta_{k'm'}^{(r-1)} - 1\bigg)\bigg]\log (1 - \rho_{k'm'}) + \text{cst.}
\end{align}
Taking the exponential and computing expectations with respect to the derived distributions for $\hat{q}^{(r-1)}(z_i')$ and the assumed form for $\hat{q}^{(r-1)}(a_{ij})$, it follows that the optimal CAVI approximation takes $\hat{q}^{(r)}(\rho_{k'm'}) = \mathrm{Beta}\bigg(\rho_{k'm'}; \eta^{(r)}_{k'm'}, \zeta_{k'm'}^{(r)}\bigg)$, with
\begin{align}
    \eta^{(r)}_{k'm'} &= \delta_\rho\bigg(\eta_{k'm'}^{(r-1)} - 1\bigg) + \sum_{(i,j)\in\mathcal{R}} \nu^{(r-1)}_{ik'}\nu^{(r-1)}_{jm'}\sigma_{ij} + 1, \\
    \zeta^{(r)}_{k'm'} &= \delta_\rho\bigg(\zeta_{k'm'}^{(r-1)} - 1\bigg) + \sum_{(i,j)\in\mathcal{R}}\nu^{(r-1)}_{ik'}\nu^{(r-1)}_{jm'}\left(1 - \sigma^{(r-1)}_{ij}\right) + 1,
\end{align}
for $k',m'\in\mathcal{K}'$.

\subsection{Approximation of $q(\pi)$}
Ignoring any terms that do not contain $\pi$, we derive the expectation as:
\begin{align}
    \mathbb{E}_{-\pi}&\{\log p(x,z,z',a,\lambda,\rho,\pi,\mu)\} =\mathbb{E}_{-\pi}\{\delta_z\log p(z|\pi) + \delta_\pi \log p(\pi)\}\\ 
    &= \mathbb{E}_{-\pi}\bigg\{\delta_z\sum_{i\in\mathcal{V}}\sum_{k\in\mathcal{K}}\tilde z_{ik}\log \pi_k + \delta_\pi \sum_{k\in\mathcal{K}}\bigg(\gamma_k^{(r-1)} - 1\bigg)\log \pi_k\bigg\} + \text{cst.} \\
    &= \sum_{k\in\mathcal{K}}\bigg\{\sum_{i\in\mathcal{V}}\mathbb{E}_{z_i}\{\tilde z_{ik}\} + \delta_\pi\bigg(\gamma_k^{(r-1)} - 1\bigg)\bigg\}\log\pi_k + \text{cst.}
\end{align}
Using the derived forms for $\hat{q}^{(r-1)}(z_i)$, we see that the optimal CAVI choice is $\hat{q}^{(r)}(\pi) = \mathrm{Dirichlet}\left(\pi; \gamma^{(r)}\right)$, where
\[
\gamma^{(r)}_k = \delta_\pi\bigg(\gamma^{(r-1)}_k - 1\bigg) + \delta_z\sum_{i\in\mathcal{V}}\tau^{(r-1)}_{ik} + 1,
\]
for $k\in\mathcal{K}$.

\subsection{Approximation of $q(\mu)$}
Ignoring any terms that do not contain $\mu$, we compute the necessary expectation as:
\begin{align}
    \mathbb{E}_{-\mu}&\{\log p(x,z,z',a,\lambda,\rho,\pi,\mu)\} =\mathbb{E}_{-\mu}\{\log p(z'|\mu) + \delta_\mu \log p(\mu)\}\\ 
    &= \mathbb{E}_{-\mu}\bigg\{\sum_{i\in\mathcal{V}}\sum_{k'\in\mathcal{K}'}\tilde z_{ik'}'\log \mu_{k'} + \delta_\mu \sum_{k'\in\mathcal{K}'}\bigg(\xi_{k'}^{(r-1)} - 1\bigg)\log \mu_{k'}\bigg\} + \text{cst.} \\
    &= \sum_{k'\in\mathcal{K}'}\bigg\{\sum_{i\in\mathcal{V}}\mathbb{E}_{z_i'}\{\tilde z_{ik'}'\} + \delta_\mu\bigg(\xi_{k'}^{(r-1)} - 1\bigg)\bigg\}\log\mu_{k'} + \text{cst.}
\end{align}
Taking the exponential, and using the derived form for $\hat{q}^{(r-1)}(z_i')$, we see that the optimal CAVI distribution is $\hat{q}^{(r)}(\mu) = \mathrm{Dirichlet}\left(\mu; \xi^{(r)}\right)$, where
\begin{equation}
\xi^{(r)}_{k'} = \delta_\mu\bigg(\mu^{(r-1)}_{k'} - 1\bigg) + \sum_{i\in\mathcal{V}}\nu^{(r-1)}_{ik^\prime} + 1,
\end{equation}
for $k'\in\mathcal{K}$.

\subsection{Approximation of $q(z_i')$}
Ignoring terms that do not contain $z_i^\prime$, we derive the expectation as:
\begin{align}
    &\mathbb{E}_{-z_i'}\left\{\log p(x,z,z',a,\lambda,\rho,\pi,\mu)\right\} = \mathbb{E}_{-z_i'}\left\{\log p(a|z',\rho) + \log p(z'|\mu)\right\} + \text{cst.} \\
    &= \mathbb{E}_{-z_i'}\Bigg\{\sum_{i',j'\in\mathcal{V}}\sum_{k',m'\in\mathcal{K}'}z_{ik'}'z_{jm'}'\bigg[a_{i'j'}\log\rho_{k'm'} + (1 - a_{i'j'})\log(1 - \rho_{k'm'})\bigg] + \sum_{k'\in\mathcal{K}'}\tilde z_{i'k'}'\log \mu_{k'}\Bigg\} + \text{cst.} \\
    &= \mathbb{E}_{-z_i'}\Bigg\{\sum_{j\in\mathcal{V}}\sum_{k',m'\in\mathcal{K}'}\tilde z_{jk'}'\tilde z_{im'}'\bigg[a_{ij}\log\rho_{k'm'} + (1 - a_{ij})\log(1 - \rho_{k'm'})\bigg] + \\
    &\qquad \qquad \sum_{j'\in\mathcal{V}}\sum_{k',m'\in\mathcal{K}'}\tilde z_{ik'}'\tilde z_{j'm'}'\bigg[a_{j'i}\log\rho_{k'm'} + (1 - a_{j'i})\log(1 - \rho_{k'm'})\bigg] - \\
    &\qquad\qquad \sum_{k'\in\mathcal{K}'}\tilde z_{ik'}'\bigg[a_{ii}\log\rho_{k'k'} + (1 - a_{ii})\log(1 - \rho_{k'k'})\bigg] + \sum_{k'\in\mathcal{K}'}\tilde z_{ik'}'\log \mu_{k'}\Bigg\} + \text{cst.}\\
    &= \sum_{k'\in\mathcal{K}'}\tilde z_{ik'}'\Bigg\{\sum_{j\in\mathcal{V}}\sum_{m'\in\mathcal{K}'}\tilde z_{jm'}'\bigg[\mathbb{E}_{a_{ij}}\{a_{ij}\}\mathbb{E}_{\rho_{k'm'}}\{\log\rho_{k'm'}\} + \mathbb{E}_{a_{ij}}\{1 - a_{ij}\}\mathbb{E}_{\rho_{k'm'}}\{\log(1-\rho_{k'm'})\}\  + \\
    &\qquad \qquad \mathbb{E}_{a_{ji}}\{a_{ji}\}\mathbb{E}_{\rho_{m'k'}}\{\log \rho_{m'k'}\} + \mathbb{E}_{a_{ji}}\{1 - a_{ji}\}\mathbb{E}_{\rho_{m'k'}}\{\log(1-\rho_{m'k'})\}\bigg]\bigg(1 - \mathbb{I}_{\{k'\}}(m')\mathbb{I}_{\{i\}}(j)\bigg) + \\
    &\qquad \qquad \mathbb{E}_{a_{ii}}\{a_{ii}\}\mathbb{E}_{\rho_{k'k'}}\{\log\rho_{k'k'}\} + \mathbb{E}_{a_{ii}}\{(1 - a_{ii})\}\mathbb{E}_{\rho_{k'k'}}\{\log(1 - \rho_{k'k'})\} + \mathbb{E}_{\mu_{k'}}\left\{\log \mu_{k'}\right\}\Bigg\}\ + \text{cst.} 
\end{align}
Taking the exponential, and using the distributions derived for $\hat{q}^{(r-1)}(\rho_{k'm'})$ and the assumed form of $\hat{q}^{(r-1)}(a_{ij})$, it follows that the optimal CAVI choice of $\hat{q}^{(r)}(z_i')$ is $\hat{q}^{(r)}(z_i') = \mathrm{Categorical}(z_i';\nu_i)$, where
\begin{align}
    &\nu^{(r)}_{ik'} \propto \exp\Bigg\{\sum_{j\in\mathcal{V}}\sum_{m'\in\mathcal{K}'}\nu^{(r-1)}_{jm'}\bigg[\sigma^{(r-1)}_{ij}\bigg(\psi\left(\eta^{(r-1)}_{k'm'}\right) - \psi\left(\eta^{(r-1)}_{k'm'} + \zeta^{(r-1)}_{k'm'}\right)\bigg)\ + \\
    &\quad \left(1 - \sigma^{(r-1)}_{ij}\right)\bigg(\psi\left(\zeta^{(r-1)}_{k'm'}\right) - \psi\left(\eta^{(r-1)}_{k'm'} + \zeta^{(r-1)}_{k'm'}\right)\bigg) + \sigma^{(r-1)}_{ji}\bigg(\psi\left(\eta^{(r-1)}_{m'k'}\right) - \psi\left(\eta^{(r-1)}_{m'k'} + \zeta^{(r-1)}_{m'k'}\right)\bigg)\ + \\
    &\quad \left(1 - \sigma^{(r-1)}_{ji}\right)\bigg(\psi\left(\zeta^{(r-1)}_{m'k'}\right) - \psi\left(\eta^{(r-1)}_{m'k'} + \zeta^{(r-1)}_{m'k'}\right)\bigg)\bigg]\bigg(1 - \mathbb{I}_{\{k'\}}(m')\mathbb{I}_{\{i\}}(j)\bigg)\ + \\
    &\quad \sigma^{(r-1)}_{ii}\bigg(\psi\left(\eta^{(r-1)}_{k'k'}\right) - \psi\left(\eta^{(r-1)}_{k'k'} + \zeta^{(r-1)}_{k'k'}\right)\bigg) + \left(1 - \sigma^{(r-1)}_{ii}\right)\bigg(\psi\left(\zeta^{(r-1)}_{k'k'}\right) - \psi\left(\eta^{(r-1)}_{k'k'} + \zeta^{(r-1)}_{k'k'}\right)\bigg)\ + \\
    & \quad \psi\left(\xi^{(r-1)}_{k'}\right) - \psi\left(\sum_\ell \xi^{(r-1)}_\ell\right)\Bigg\},
\end{align}
for $i\in\mathbb{V}$ and $k'\in\mathcal{K}'$.

\section{Derivation of CAVI updates for the Dirichlet process prior}
\label{app:cavi_GEM}
Under our model specification, the term for $\log p(z_i|u)$ can be rewritten using indicators as
\begin{align}
    \log p(z_i|u) &= \sum_{\ell=1}^\infty\bigg[\mathbb{I}\{z_i=\ell\}\log(u_\ell) + \mathbb{I}\{z_i > \ell\}\log(1 - u_\ell)\bigg].
\end{align}
The joint loglikelihood then becomes
\begin{align}
    \log p(x,z,u,\lambda) &= (\nu - 1)\sum_{\ell=1}^\infty \log(1 - u_\ell) + (\alpha_{km} - 1)\sum_{k,m=1}^\infty \bigg[\log (\lambda_{km}) - \beta_{km}\lambda_{km}\bigg] \ + \\
    &\sum_{i\in\mathcal{V}}\sum_{\ell=1}^\infty\bigg[\mathbb{I}\{z_i=\ell\}\log(u_\ell) \ + \mathbb{I}\{z_i > \ell\}\log(1 - u_\ell)\bigg]\ + \\
    &\sum_{(i,j)\in\mathcal{E}}\sum_{k,m=1}^\infty \tilde z_{ik}\tilde z_{jm}\bigg[x_{ij}(I_r) \log(\lambda_{km}) - \Delta\lambda_{km}\bigg]\ + \text{cst.}.
\end{align}

\subsection{Approximation of $q(z_i)$}
Recalling that $q(u_L = 1) = 1$, it follows that $\mathbb{I}\{z_i > L\} = 0$. We can thus compute the necessary expectation as:
\begin{align}
    \mathbb{E}_{-z_i}&\{\log p(x,z,u,\lambda)\} = \mathbb{E}_{-z_i}\bigg\{\delta_z\sum_{\ell = 1}^L\bigg[\mathbb{I}\{z_i = \ell\}\log(u_\ell) + \mathbb{I}\{z_i > \ell\} \log(1 - u_\ell)\bigg] + \\
    &\quad \sum_{j:(i,j)\in\mathcal{E}}\sum_{k,m=1}^{L} \tilde z_{ik}\tilde z_{jm}\bigg(x_{ij}(I_r)\log(\lambda_{km}) - \Delta\lambda_{km}\bigg)\ + \\
    &\quad \sum_{j':(j',i)\in\mathcal{E}}\sum_{k,m=1}^{L} \tilde z_{j'k}\tilde z_{im}\bigg(x_{j'i}(I_r)\log(\lambda_{km}) - \Delta\lambda_{km}\bigg) - 
    \sum_{k=1}^L \tilde z_{ik}\bigg(x_{ii}(t)\log (\lambda_{kk}) - \Delta\lambda_{kk}\bigg)\bigg\} + \mathrm{cst.} \\
    &= \sum_{k = 1}^L\tilde z_{ik}\Bigg\{\delta_z\left[\mathbb{E}_{u_k}\{\log(u_k)\} + \sum_{\ell=1}^{k-1} \mathbb{E}_{u_\ell}\{\log(1 - u_\ell)\}\right] + x_{ii}(t)\mathbb{E}_{\lambda_{kk}}\{\log (\lambda_{kk})\} - \Delta\mathbb{E}_{\lambda_{kk}}\{\lambda_{kk}\} + \\
    &\quad \sum_{m=1}^{L}\Bigg[ \sum_{j:(i,j)\in\mathcal{E}} \mathbb{E}_{z_j}\{\tilde z_{jm}\}\bigg(x_{ij}(I_r)\mathbb{E}_{\lambda_{km}}\{\log(\lambda_{km})\} - \Delta\mathbb{E}_{\lambda_{km}}\{\lambda_{km}\}\bigg)\bigg(1 - \mathbb{I}_{\{k\}}(m)\mathbb{I}_{\{i\}}(j)\bigg)\ + \\
    &\sum_{j':(j',i)\in\mathcal{E}} \mathbb{E}\{\tilde z_{j'm}\}\bigg(x_{j'i}(I_r)\mathbb{E}_{\lambda_{mk}}\{\log(\lambda_{mk})\} - \Delta\mathbb{E}_{\lambda_{mk}}\{\lambda_{mk}\}\bigg)\bigg(1 - \mathbb{I}_{\{k\}}(m)\mathbb{I}_{\{i\}}(j')\bigg)\Bigg]\Bigg\} + \mathrm{cst.}
\end{align}
Using the distributions derived for $\hat{q}^{(r-1)}(\lambda_{km})$ and $\hat{q}^{(r-1)}(u_\ell)$, it follows that the optimal CAVI distribution is $\hat{q}^{(r)}(z_i) = \mathrm{Categorical}\left(z_i; \tau^{(r)}_i\right)$,
with
\begin{align}
    \tau^{(r)}_{ik} &\propto \exp\Bigg\{\delta_z\left[\psi\left(\omega^{(r-1)}_k\right) - \psi\left(\omega^{(r-1)}_k + \nu^{(r-1)}_k\right) + \sum_{\ell=1}^{k-1} \bigg(\psi\left(\nu^{(r-1)}_\ell\right) - \psi\left(\omega^{(r-1)}_\ell + \nu^{(r-1)}_\ell\right)\bigg)\right]\ + \\
    &\quad x_{ii}(I_r)\bigg(\psi\left(\alpha^{(r-1)}_{kk}\right) - \log\left(\beta^{(r-1)}_{kk}\right)\bigg) - \Delta\frac{\alpha^{(r-1)}_{kk}}{\beta^{(r-1)}_{kk}} + \\
    &\sum_{m=1}^{L}\Bigg[ \sum_{j:(i,j)\in\mathcal{E}} \tau^{(r-1)}_{jm}\bigg(x_{ij}(I_r)\left(\psi\left(\alpha^{(r-1)}_{km}\right) - \log\left(\beta^{(r-1)}_{km}\right)\right) - \Delta\frac{\alpha^{(r-1)}_{km}}{\beta^{(r-1)}_{km}}\bigg)\bigg(1 - \mathbb{I}_{\{k\}}(m)\mathbb{I}_{\{i\}}(j)\bigg)\ + \\
    &\sum_{j':(j',i)\in\mathcal{E}} \tau^{(r-1)}_{j'm}\bigg(x_{j'i}(I_r)\left(\psi\left(\alpha^{(r-1)}_{mk}\right) - \log\left(\beta^{(r-1)}_{mk}\right)\right) - \Delta\frac{\alpha^{(r-1)}_{mk}}{\beta^{(r-1)}_{mk}}\bigg)\bigg(1 - \mathbb{I}_{\{k\}}(m)\mathbb{I}_{\{i\}}(j')\bigg)\Bigg]\Bigg\},
    \label{app_eqn:tau_fp}
\end{align}
for $i\in\mathcal{V}$ and $k\in\mathcal{K}$.

\subsection{Approximation of $q(\lambda_{km})$}

Recall that $q(u_L=1) = 1$, and so $\pi(u_{L+1}) = \pi(u_{L+2}) = \dots$. This truncation removes the need for an infinite sum. In this way, we can compute the expectation for the CAVI approximation as:
\begin{align}
    \mathbb{E}_{-\lambda_{km}}\{\log p(x,z,u,\lambda)\} &=
    \sum_{k,m=1}^\infty\mathbb{E}_{-\lambda_{km}}\bigg[ \delta_{\lambda}\left(\alpha^{(r-1)}_{km}-1\right)\log\lambda_{km} - \delta_{\lambda}\beta^{(r-1)}_{km}\lambda_{km}\ + \\
    &\quad \sum_{(i,j)\in\mathcal{E}}\tilde z_{ik}\tilde z_{jm}\bigg\{x_{ij}(I_r)\log(\lambda_{km}) - \Delta\lambda_{km}\bigg\}\bigg] + \text{cst.} \\
    &= \sum_{k,m=1}^{L}\bigg[\bigg(\delta_{\lambda}\left(\alpha^{(r-1)}_{km} - 1\right) + \sum_{(i,j)\in\mathcal{E}}\mathbb{E}_{z_i}\{\tilde z_{ik}\}\mathbb{E}_{z_j}\{\tilde z_{jm}\}x_{ij}(I_r)\bigg)\log(\lambda_{km})\ - \\
    &\quad \bigg(\delta_{\lambda}\beta^{(r-1)}_{km} + \Delta\sum_{(i,j)\in\mathcal{E}}\mathbb{E}_{z_i}\{\tilde z_{ik}\}\mathbb{E}_{z_j}\{\tilde z_{jm}\}\bigg)\lambda_{km}\bigg] + \text{cst.}
\end{align}
Using the derived distributions for $\hat{q}^{(r-1)}(z_i)$, it follows that the optimal CAVI approximation is $q^{(r-1)}(\lambda_{km}) = \mathrm{Gamma}\left(\lambda_{km}; \alpha_{km}^{(r)},\beta_{km}^{(r)}\right)$, where
\begin{align}
    \alpha_{km}^{(r)} &= \delta_{\lambda}\left(\alpha^{(r-1)}_{km} - 1\right) + \sum_{(i,j)\in\mathcal{E}}\tau^{(r-1)}_{ik}\tau^{(r-1)}_{jm}x_{ij}(I_r) + 1,
    \label{app_eqn:alpha}\\
    \beta_{km}^{(r)} &= \delta_\lambda\beta_{km}^{(r-1)} + \Delta\sum_{(i,j)\in\mathcal{E}}\tau_{ik}^{(r-1)}\tau_{jm}^{(r-1)}
    \label{app_eqn:beta},
\end{align}
for $k,m\in\mathcal{K}$.

\subsection{Approximation of $q(u_i)$}
Again, using that $q^{(r-1)}(u_L = 1) = 1$, we have $\log(1 - u_L) = 0$
and $\mathbb{I}_{\{k\}}\{z_j\}= 0$, for $k>L$, the required expectation becomes:
\begin{align}
    \mathbb{E}_{-u_j}&\{\log p(x,z,u,\lambda)\} = \mathbb{E}_{-u_j}\bigg\{\mathbb{I}\{j < L\}\bigg[\delta_u\left(\omega_j^{(r-1)} - 1\right)\log(u_j) + \delta_u\left(\nu_j^{(r-1)}-1\right)\log(1 - u_j)\  + \\
    &\quad \delta_z\sum_{i\in\mathcal{V}}\bigg(\mathbb{I}_{\{j\}}\{z_i\}\log(u_j) + \sum_{k>j}\mathbb{I}_{\{k\}}\{z_i\}\log(1 - u_j)\bigg)\bigg\} + \text{cst.}. \\
    &= \mathbb{I}\{j < L\}\bigg[\log(u_j)\bigg(\delta_z\sum_{i\in\mathcal{V}}\tau^{(r-1)}_{ij} + \delta_u\left(\omega_j^{(r-1)} - 1\right)\bigg)\ + \\
    &\quad \log(1 - u_j)\bigg(\delta_z\sum_{i\in\mathcal{V}}\sum_{k=j+1}^L\tau^{(r-1)}_{ik} + \delta_u\left(\nu^{(r-1)} - 1\right)\bigg)\bigg] \ + \mathrm{cst.}
\end{align}
where we have used the derived forms for $\hat{q}^{(r-1)}(z_i)$ to compute the indicator expectations. It follows that the optimal CAVI approximation is $\hat{q}^{(r)}(u_j) = \mathrm{Beta}\left(u_j; \omega_j^{(r)},\nu_j^{(r)}\right)$, where
\begin{align}
    \omega_j' &= \delta_z\sum_{i\in\mathcal{V}}\tau_{ij}^{(r-1)} + \delta_u\left(\omega_j^{(r-1)} - 1\right) + 1,
    \label{app_eqn:omega}\\
    \nu_j' &= \delta_z\sum_{i\in\mathcal{V}} \sum_{k=j+1}^L\tau_{ik}^{(r-1)} + \delta_u\left(\nu_j^{(r-1)} - 1\right) + 1,
    \label{app_eqn:nu}
\end{align}
for $j \in \{1,\dots,L\}$.

\end{appendices}

\end{document}